\documentclass[final]{dmtcs-episciences} 

\usepackage[english]{babel}

\usepackage[utf8]{inputenc}
\usepackage[T1]{fontenc}

\usepackage{amssymb}
\usepackage{amsmath}
\usepackage{amsthm}

\usepackage[vlined,ruled,linesnumbered]{algorithm2e}

\usepackage[numbers]{natbib}
\usepackage{graphicx}
\usepackage{microtype}
\usepackage{paralist}
\usepackage{ifdraft}

\newcommand{\leo}[1]{#1}

\newcommand{\leojournal}[1]{#1}
 \newcommand{\markj}[1]{#1}
\newcommand{\mjnew}[1]{#1}

\newtheorem{theorem}{Theorem}
\newtheorem{definition}{Definition}
\newtheorem{lemma}{Lemma}
\newtheorem{observation}{Observation}
\newtheorem{corollary}{Corollary}

\def\noqed{\expandafter\def\csname endproof\endcsname{\medskip}\expandafter\def\csname endclaimproof\endcsname{\medskip}}
\def\qedhere{\qed\aftergroup\noqed}

\ifx\claim\undefined
  \def\claimqedsymbol{$\lrcorner$}
  \def\claimqed{\hfill\claimqedsymbol}
  \def\claimqedhere{\claimqed\aftergroup\noqed}
  \newtheorem{claim}{Claim}
  \newenvironment{claimproof}{{\noindent\textbf{Proof:}\expandafter\def\csname qedsymbol\endcsname{\claimqedsymbol}}}{\claimqed} 
\fi

\makeatletter 
\newcommand\nobreakpar{\par\nobreak\@afterheading} 
\makeatother

\usepackage[capitalise, nameinlink]{cleveref}

\title{Embedding phylogenetic trees  in networks of low treewidth}

\author{Leo van Iersel\affiliationmark{1} 
\and Mark Jones\affiliationmark{1}
\and Mathias Weller\affiliationmark{2}}

\affiliation{Delft Institute of Applied Mathematics, Delft University of Technology, The Netherlands\\
CNRS, LIGM (UMR 8049), Université Gustave Eiffel, Champs-sur-Marne, France}

\newcommand{\mjrev}[1]{{#1}}


\usepackage{xcolor}
\colorlet{darkgreen}{green!50!black}
\colorlet{dg}{darkgreen}
\colorlet{medgray}{gray!75}
\colorlet{lightgray}{gray!30}
\colorlet{bagcol}{gray!70}
\colorlet{pastcol}{gray!20}
\definecolor{linkcol}{rgb}{0,0,0.4} 
\definecolor{citecol}{rgb}{0.5,0,0} 

\usepackage{tikz}
\usetikzlibrary{arrows.meta,shapes,positioning,calc,decorations.pathreplacing,decorations.markings,decorations.pathmorphing,patterns}
\pgfdeclarelayer{background2}
\pgfdeclarelayer{background}
\pgfdeclarelayer{foreground}
\pgfsetlayers{background2,background,main,foreground}

\tikzstyle{hidden}=[opacity=0]
\tikzstyle{fade}=[opacity=0.2]
\tikzstyle{nonsol}=[dashed]

\tikzstyle{bold}=[draw, line width=2pt]
\tikzstyle{optional}=[dashed]
\tikzstyle{path}=[decorate, decoration={snake, amplitude=.6mm}]

\tikzstyle{small}=[inner sep=2pt]
\tikzstyle{tiny}=[inner sep=1.7pt]
\tikzstyle{textnode}=[inner sep=0pt]

\tikzstyle{triangle}=[draw, regular polygon, regular polygon sides=3]
\tikzstyle{vertex}=[circle, draw, fill=white]
\tikzstyle{reti}=[vertex, fill=black]
\tikzstyle{leaf}=[vertex, rectangle]
\tikzstyle{leaf2}=[vertex, regular polygon, regular polygon sides=3]

\tikzstyle{smallvertex}=[vertex, small]
\tikzstyle{smallleaf}=[leaf, inner sep=3.3pt]
\tikzstyle{smallleaf2}=[leaf2, inner sep=1.7pt]
\tikzstyle{smalltriangle}=[triangle, inner sep=1.5pt]
\tikzstyle{smallreti}=[reti, small]

\tikzstyle{tinyvertex}=[vertex, tiny]

\tikzstyle{normal}=[smallvertex, fill=black]

\tikzstyle{edge}=[draw,-]
\tikzstyle{matching}=[edge,line width=3pt]
\tikzstyle{arc}=[draw,arrows={-Latex[length=6pt]}]
\tikzstyle{boldarc}=[draw, bold, arrows={-Latex[length=10pt]}]
\tikzstyle{revarc}=[draw, arrows={Latex[length=6pt]-}]
\tikzstyle{boldrevarc}=[draw, bold, arrows={Latex[length=10pt]-}]

\tikzstyle{bag}=[bagcol, line width=15pt]
\tikzstyle{past}=[draw=gray, fill=pastcol]

\newcommand{\nextnode}[5][vertex]{\node[small#1] (#2) at ($(#3)+(#4)$) {} edge[revarc,#5] (#3);}

\newcommand{\mybox}[2]{
  \noindent\begin{tikzpicture}
    \node[minimum width=\linewidth-5pt, draw, rounded corners, text width=\linewidth-15pt] (a){#2};
    \node[fill=white, xshift=1em, anchor=west] at (a.north west) {#1};
  \end{tikzpicture}%
}

\newcommand{\myboxprobdef}[6][Question]{
  \label{#6}%
  \noindent\mybox{%
    \ifthenelse{\equal{#4}{}}{}{{#4}\ifthenelse{\equal{#5}{}}{}{ ({#5})}}
  }{%
    \begin{compactdesc}
      \item [Input:] {#2}
      \item [#1:] {#3}
    \end{compactdesc}
  }
}

\definecolor{linkcol}{rgb}{0,0,0.4} 
\definecolor{citecol}{rgb}{0.5,0,0} 
\hypersetup{
bookmarksopen=false,
pdftoolbar=false,
pdfmenubar=true,
pdfhighlight=/O,
colorlinks=true,
pdfpagemode=UseNone,
pdfpagelayout=SinglePage,
pdffitwindow=true,
linkcolor=linkcol,
citecolor=citecol,
urlcolor=linkcol
}

\newcommand{\T}{\mathcal{T}}

\newcommand{\Y}{\mathcal{Y}}

\newcommand{\pathSet}{\mathcal{P}}

\newcommand{\past}{\textsc{past}\xspace}
\newcommand{\future}{\textsc{future}\xspace}
\newcommand{\pleft}{\textsc{left}\xspace}
\newcommand{\pright}{\textsc{right}\xspace}
\newcommand{\true}{\textsc{true}\xspace}
\newcommand{\false}{\textsc{false}\xspace}
\let\emb\phi
\let\emphmore\emph

\newcommand{\Nin}{N_{\textsc{in}}}
\newcommand{\Tin}{T_{\textsc{in}}}

\newcommand{\NinTin}{D_\textsc{in}(\Nin, \Tin)}
\newcommand{\Disp}[1]{D(N_{#1}, T_{#1})}

\newcommand{\gecs}{\theta} 

\def\emptyset{\varnothing}

\begin{document}
\publicationdata{vol. 25:2 }{2023}{4}{10.46298/dmtcs.10116}{2022-10-03; 2022-10-03; 2023-05-16}{2023-06-21}
\maketitle

\begin{abstract}
  Given a rooted, binary phylogenetic network and a rooted, binary phylogenetic tree, can the tree be embedded into the network?
  This problem, called \textsc{Tree Containment}, arises when validating networks constructed by phylogenetic inference methods.
  We present the first algorithm for (rooted) \textsc{Tree Containment} using the treewidth~$t$ of the input network~$N$ as parameter,
  showing that the problem can be solved in $2^{O(t^2)}\cdot|N|$~time and space.
\end{abstract}

\section{Introduction}

\subsection{Background: phylogenetic trees and networks}

\looseness=-1
Phylogenetic trees and networks are graphs used to represent evolutionary relationships. In particular, a \emphmore{rooted phylogenetic network} is a directed acyclic graph with distinctly labelled leaves, a unique root and no indegree-1 outdegree-1 vertex. Here, we will only consider rooted binary phylogenetic networks, which we will call \emph{networks} for short. The labels of the leaves (indegree-1 outdegree-0 vertices) can, for example, represent a collection of studied biological species, and the network then describes how they evolved from a common ancestor (the root, a unique indegree-0 outdegree-2 vertex). Vertices with indegree~2 and outdegree~1 are called \emphmore{reticulations} and represent events where lineages combine, for example 
the emergence of new hybrid species. All other vertices have indegree~1 and outdegree~2. A network without reticulations is a \emphmore{phylogenetic tree}.

\subsection{The \textsc{Tree Containment} problem}

\looseness=-1
The evolutionary history of a small unit of hereditary information (for example a gene, a fraction of a gene or (in linguistics) a word) can often be described by a phylogenetic tree. This is because at each reticulation,
each unit is inherited from only one parent.
Hence, if we trace back the evolutionary history of such a hereditary unit in the network, we see that its phylogenetic tree can be embedded in the network. This raises the fundamental question: given a phylogenetic network and a phylogenetic tree, can the tree be embedded into the network? This is called the \textsc{Tree Containment} problem (see \cref{fig:TC}).
To formalize this problem, we say that a network $N$ \emphmore{displays} a tree $T$ if some subgraph 
of $N$ is a subdivision of $T$.

\myboxprobdef%
{phylogenetic network $\Nin$ and tree $\Tin$, both on the same set of leaf labels}%
{Does $\Nin$ display $\Tin$?}%
{\textsc{Tree Containment}}{TC}{def:TCshort}%

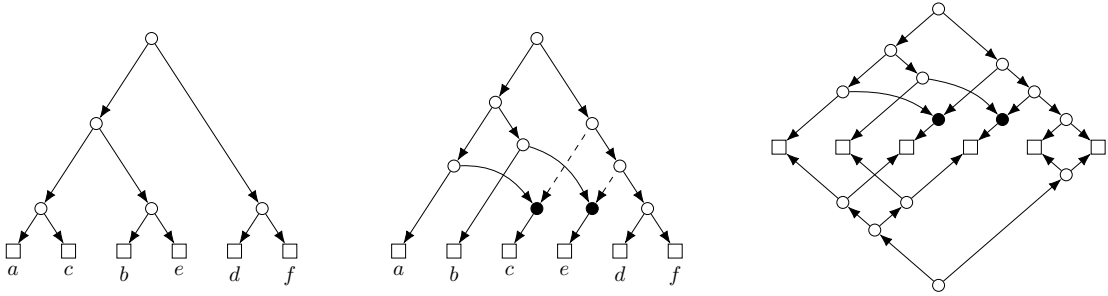
\begin{figure}[t]
  \centering
  \scalebox{.8}{\begin{tikzpicture}[xscale=.65, yscale=-1]
    \node[smallvertex] (root) at (0,0) {};
    \nextnode{0}{root}{45:4}{}
    \nextnode[leaf, label=below:$f$]{f}{0}{45:1}{}
    \nextnode[leaf, label=below:$d$]{d}{0}{135:1}{}

    \nextnode{1}{root}{135:2}{}
    \nextnode{10}{1}{45:2}{}
    \nextnode[leaf, label=below:$b$]{b}{10}{135:1}{}
    \nextnode[leaf, label=below:$e$]{e}{10}{45:1}{}
    \nextnode{11}{1}{135:2}{}
    \nextnode[leaf, label=below:$a$]{a}{11}{135:1}{}
    \nextnode[leaf, label=below:$c$]{c}{11}{45:1}{}
  \end{tikzpicture}}
  \hfill
  \scalebox{.8}{\begin{tikzpicture}[xscale=.65, yscale=-1]
    \node[smallvertex] (root) at (0,0) {};
    \nextnode{0}{root}{45:2}{}
    \nextnode{00}{0}{45:1}{}
    \nextnode{000}{00}{45:1}{}
    \nextnode[leaf, label=below:$f$]{f}{000}{45:1}{}
    \nextnode[reti]{001}{00}{135:1}{nonsol}
    \nextnode[leaf, label=below:$e$]{e}{001}{135:1}{}
    \nextnode[leaf, label=below:$d$]{d}{000}{135:1}{}
    \nextnode[reti]{01}{0}{135:2}{nonsol}
    \nextnode[leaf, label=below:$c$]{d}{01}{135:1}{}

    \nextnode{1}{root}{135:1.5}{}
    \nextnode{10}{1}{45:1}{}
    \nextnode[leaf, label=below:$b$]{b}{10}{135:2.5}{}
    \nextnode{11}{1}{135:1.5}{}
    \nextnode[leaf, label=below:$a$]{a}{11}{135:2}{}

    \foreach \u/\v in {10/001, 11/01} \draw[arc] (\u) to[bend right=20] (\v);
  \end{tikzpicture}}
  \hfill
  \scalebox{.8}{\begin{tikzpicture}[xscale=.75, yscale=-.65]
    \node[smallvertex] (root) at (0,0) {};
    \nextnode{0}{root}{45:2}{}
    \nextnode{00}{0}{45:1}{}
    \nextnode{000}{00}{45:1}{}
    \nextnode[leaf]{f}{000}{45:1}{}
    \nextnode[reti]{001}{00}{135:1}{}
    \nextnode[leaf]{e}{001}{135:1}{}
    \nextnode[leaf]{d}{000}{135:1}{}
    \nextnode[reti]{01}{0}{135:2}{}
    \nextnode[leaf]{c}{01}{135:1}{}

    \nextnode{1}{root}{135:1.5}{}
    \nextnode{10}{1}{45:1}{}
    \nextnode[leaf]{b}{10}{135:2.5}{}
    \nextnode{11}{1}{135:1.5}{}
    \nextnode[leaf]{a}{11}{135:2}{}

    \foreach \u/\v in {10/001, 11/01} \draw[arc] (\u) to[bend right=20] (\v);

    \nextnode{ac}{a}{45:2}{arc}
    \nextnode{be}{b}{45:2}{arc}
    \nextnode{df}{d}{45:1}{arc}
    \nextnode{acbe}{ac}{45:1}{arc}
    \nextnode{Trt}{acbe}{45:2}{arc}
    
    \foreach \u/\v in {ac/c, be/e, df/f, acbe/be, Trt/df} \draw[arc] (\u) -- (\v);
\end{tikzpicture}}
  \caption{%
  \textbf{Left}: a phylogenetic tree $T$.
  \textbf{Middle}: a phylogenetic network $N$ displaying $T$ (solid lines indicate an embedding of $T$; black nodes indicate reticulations).
  \textbf{Right}: the display graph $D(N,T)$ of $N$ and $T$ (see \cref{sec:chal})
  with the network part drawn on top and the tree part drawn on the bottom.
  Note that vertices of the display graph are not labelled. In the figure, the leaves (square vertices) are ordered in the same way as in~$N$.}
  \label{fig:TC}
\end{figure}
\noindent

\subsubsection{Motivation}

Apart from being a natural and perhaps one of the most fundamental questions regarding phylogenetic networks, the \textsc{Tree Containment} problem has direct applications in phylogenetics. The main application is the validation of phylogenetic network inference methods. After constructing a network, one may want to verify whether it is consistent with the phylogenetic trees. For example, if a heuristic method is used to generate a network for a genomic data set, and tree inference methods are used to generate trees for each gene, then the quality of the produced network can be assessed by computing the fraction of the gene trees that can be embedded into it. In addition, one may want to find the actual embeddings for visualisation purposes and/or to assess the importance of each network arc.

However, our main motivation for studying \textsc{Tree Containment} is that it is a first step towards the wider application of treewidth based approaches in phylogenetics (see \cref{sec:tw,sec:chal}). The techniques we develop are not exclusively designed for \textsc{Tree Containment} but intended to be useful also for other problems such as \textsc{Network Containment}~\citep{janssen2021cherry} and \textsc{Hybridization Number}~\citep{bordewich2007computing,van2016hybridization,van2016kernelizations,van2013quadratic}. The former is the natural generalization of \textsc{Tree Containment} in which we have two networks as input and want to decide whether one can be embedded into the other. It can, in particular, be used to decide whether two networks are isomorphic. In the latter problem, \textsc{Hybridization Number}, the input consists of a set of phylogenetic trees, and the aim is to construct a network with at most~$k$ reticulations that embeds each of the input trees. Although this will certainly be non-trivial, we expect that at least part of the approach we introduce here can be applied to those and other problems in phylogenetics.
We also believe our approach may have application to problems outside of phylogenetics, involving the reconciliation of multiple related graphs. 

\subsection{Treewidth}\label{sec:tw}

The parameters that are most heavily used for phylogenetic network problems (see eg. \cref{sec:previous work}) are the reticulation number and the level.
This is true not only for \textsc{Tree Containment} but more generally in the phylogenetic networks literature.
Although these parameters are natural, their downside is that they are not necessarily much smaller than the input size.
This is why we study a different parameter here.

The \emphmore{treewidth} of a graph measures its tree-likeness (see definition below), similarly to the reticulation number and level. In that sense, it is also a natural parameter to consider in phylogenetics, where networks are often expected to be reasonably tree-like. A major advantage of treewidth is that it is expected to be much smaller than the reticulation number and level.
In particular, 
there exist classes of networks for which the treewidth is at most a constant factor times the square-root of the level (see~\citep{kelk2018treewidth} for an example).
Moreover, a broad range of advanced techniques have been developed for designing FPT algorithms for graph problems when the parameter is the treewidth~\citep{bodlaender1988dynamic,cygan2011solving,BCKN15,eiben2021measuring}.
For these reasons, the treewidth has recently been studied for phylogenetics problems \citep{Janssen2019Treewidth,Kelk2016MSOL,kelk2018treewidth,scornavacca2021treewidth} and related width parameters have been proposed~\citep{berry2020scanning}. However, using treewidth as parameter for phylogenetic problems poses major challenges and, therefore, there are still few algorithms in phylogenetics that use treewidth as parameter \mjnew{(see~\cref{sec:chal})}.

\looseness=-1
\markj{One of the most famous results in treewidth is Courcelle's Theorem~\citep{COURCELLE199012,ARNBORG1991308}, which states, informally, that graph problems expressible in monadic second-order logic (MSOL) can be solved in linear time on graphs of bounded treewidth. This makes MSOL formulations  a powerful tool for establishing FPT results.
For practical purposes, it is often preferable to establish a concrete algorithm, since the running times derived via Courcelle's theorem are dominated by a tower of exponentials of size bounded in the treewidth.}

It will be convenient to define a \emphmore{tree decomposition} of a graph~$G=(V,E)$ as a rooted tree, where each vertex of the tree is called a \emph{bag} and is assigned a partition $(P,S,F)$ of~$V$, \mjnew{where $S$ is a separator between $P$ and $F$}.
We will refer to~$S$ as the \emph{present} of the bag. The set~$P$ is equal to the union of the presents of all 
descendant bags (minus the elements of $S$) and we refer to it as the \emph{past} of the bag. The set~$F=V\setminus (S\cup P)$ is referred to as the \emph{future} of the bag.
For each edge of the graph, there is at least one bag for which both endpoints of the edge are in the present of the bag. Finally, for each~$v\in V$, the bags that have~$v$ in the present form a non-empty connected subtree of the tree decomposition. 
The \emph{width} of a tree decomposition is one less than the maximum size of any bag's present and the \emph{treewidth} $tw(G)$ of a graph~$G$ is the minimum width of any tree decomposition of $G$. The treewidth of a phylogenetic network or other directed graph is the treewidth of the underlying undirected graph.

Our dynamic programming works with \emphmore{nice} tree decompositions,
in which the root is assigned $(V,\emptyset,\emptyset)$ and each bag assigned $(P,S,F)$ has exactly one of four types:
\emph{Leaf} bags have $P=S=\emptyset$ (hence~$F=V$) and have no child,
\emph{Introduce} bags have a single child assigned $(P,S\setminus\{z\},F\cup\{z\})$ for some $z\in S$,
\emph{Forget} bags have a single child assigned $(P\setminus\{z\},S\cup\{z\},F)$ for some $z\in P$, and
\emph{Join} bags have two children  
assigned $(L,S,F \cup R)$ and $(R,S, F \cup L)$ respectively, where $(L,R)$ is a partition of $P$.
When the treewidth is bounded by a constant, \citep{bodlaender1996linear} showed that a minimum-width tree decomposition can be found in linear time and~\citep{Kloks1994Treewidth} showed that a nice tree decomposition of the same width can be obtained in linear time.
Regarding approximation,
it is known that, for all graphs~$G$, tree decompositions of width $O(tw(G))$ can be computed in time single-exponential in $tw(G)$~\citep{cygan2015parameterized,bodlaender2016c,Kor21}
and tree decompositions of width 
$O(tw(G)\sqrt{\log tw(G)})$ can be computed in polynomial time \citep{feige08}.

\subsection{Challenges}\label{sec:chal}

One of the main challenges of using treewidth as parameter in phylogenetics is that the central goal in this field is to infer phylogenetic networks and, thus, the network is not known \emph{a priori} so a tree decomposition cannot be constructed easily. A possible strategy to overcome this problem is to work with the \emphmore{display graph} (see \cref{fig:TC}). Consider a problem taking as input a set of trees, such as \textsc{Hybridization Number}. Then, the \emph{display graph} of the trees is obtained by taking all trees and identifying leaves with the same label. Now we have a graph in the input and hence we can compute a tree decomposition. Moreover, in some cases, there is a strong relation between the treewidth of the display graph and the treewidth of an optimal network~\citep{grigoriev2014low,kelk2018treewidth,Janssen2019Treewidth}.

A few instances of exploiting (tree decompositions of) the display graph of input networks for algorithm design have been published.
Famously, Bryant and Lagergren~\citep{BryantLagergren2006} designed MSOL formulations solving the \textsc{Tree Consistency} problem on display graphs,
which have been improved by
a concrete dynamic programming on a given tree decomposition of the display graph~\citep{baste2017}.
Kelk et al.~\citep{Kelk2016MSOL} also developed MSOL formulations on display graphs for multiple incongruence measures on trees, based on so-called ``agreement forests''.

For the \textsc{Tree Containment} problem,
MSOL formulations acting on the display graph have been used to prove fixed-parameter tractability with respect to the treewidth~\citep{Janssen2019Treewidth}.
Analogously to the work of Baste et al.~\citep{baste2017} for \textsc{Tree Consistency},
we develop in this manuscript a concrete dynamic programming algorithm for \textsc{Tree Containment} acting on display graphs.

\looseness=-1
\textsc{Tree Containment} is conceptually similar to \textsc{Hybridization Number} in the sense that the main challenge is to decide which tree vertices correspond to which vertices of the other trees (for \textsc{Hybridization Number}) or network vertices (for \textsc{Tree Containment}). However, \textsc{Hybridization Number} is even more challenging since the network may contain vertices that do not correspond to any 
input
vertex~\citep{van2016hybridization}. Therefore, \textsc{Tree Containment} is a natural first problem to develop techniques for, aiming at extending them to \textsc{Hybridization Number} and other problems in phylogenetics in the long run.

That being said, solving \textsc{Tree Containment} parameterized by treewidth poses major challenges itself. Even though the general idea of dynamic programming on a tree decomposition is clear, its concrete use for \textsc{Tree Containment} is severely complicated by the fact that the tree decomposition does not know the correspondence between tree vertices and network vertices.
For example, when considering a certain bag of the tree decomposition, a tree vertex that is in the present of that bag may have to be embedded into a network vertex that is in the past or in the future. It may also be necessary to map vertices from the future of the tree to the past of the network and vice versa. Therefore, it will not be possible to ``forget the past'' and ``not worry about the future''. In particular, this makes it much more challenging to bound the number of possible assignments for a given bag. We will do this by bounding the number of ``time-travelling'' vertices by a function of the treewidth. We will describe these challenges in more detail in \cref{sec:approach}.

\subsection{Previous work}\label{sec:previous work}
\textsc{Tree Containment} was shown to be NP-hard~\citep{kanj2008seeing}, even for 
tree-sibling, time-consistent, regular networks~\citep{van2010locating}.
On the positive side, polynomial-time algorithms were found for other restricted classes,
including tree-child networks~\citep{gambette2016,gambette2018solving,gunawan2018solving,gunawan2019compression,van2010locating,weller2018linear}.
The first non-trivial FPT algorithm for \textsc{Tree Containment} on general networks had running time~$O(2^{k/2}n^2)$,
where the parameter~$k$ is the number of reticulations in the network~\citep{kanj2008seeing}.
Another algorithm was proposed by~\citep{gunawan2016program} with the same parameter, but it is only shown to be FPT for a restricted class of networks.
Since the problem can be split into independent subproblems at non-leaf cut-edges~\citep{van2010locating},
the parameterization can be improved to the largest number of reticulations in any biconnected component (block),
also called the \emph{level} of the network.
Further improving the parameterization,
the maximum number~$t^*$ of ``unstable component-roots'' per biconnected component was considered and
an algorithm (working also in the non-binary case) was found with running time $O(3^{t^*}|N||T|)$~\citep{weller2018linear}.
Herein, a parameterization ``improves'' over another if the first is provably smaller than (a function of) the second in any input network.

\looseness=-1
Several generalizations and variants of the \textsc{Tree Containment} problem have been studied.
The more general \textsc{Network Containment} problem asks to embed a \emph{network} in another and
has been shown to be solvable in polynomial time on a restricted network class~\citep{janssen2021cherry}.
When allowing multifurcations and non-binary reticulations, two variants of \textsc{Tree Containment} have been considered:
In the \textsc{firm} version, each non-binary node (``polytomy'') of the tree has to be embedded in a polytomy of the network whereas,
in the \textsc{soft} version, polytomies may be ``resolved'' into binary subtrees in any way~\citep{bentert2018tree}.
Finally, the unrooted version of \textsc{Tree Containment} was also shown to be NP-hard but fixed-parameter tractable when the parameter is the reticulation number (the number of edges that need to be deleted from the network to obtain a tree) \citep{van2018unrooted}.
While this version of the problem is also known to be fixed-parameter tractable with respect to the treewidth of the network~\citep{Janssen2019Treewidth},
the work does not explicitly describe an algorithm and the implied running time depends on Courcelle's theorem~\citep{Cou97} which makes practical implementation virtually impossible.

\looseness=-1
Since the notion of ``display'' closely resembles that of ``topological minor'' (with the added constraint that the embedding must respect leaf-labels),
\textsc{Tree Containment} can be understood as a special case of a variant of the well-known \textsc{Topological Minor Containment} (TMC) problem
for directed graphs.
TMC is known to be NP-complete in general by reduction from \textsc{Hamiltonian Cycle}
and previous algorithmic results focus on the undirected variant, parameterized by the \emph{size}~$h$ of the sought topological minor~$H$
(corresponding to the input tree for \textsc{Tree Containment}).
In particular, undirected TMC can be solved in $f(h)n^{O(1)}$~time~\citep{GKMW11,FLP+20}.
However, the dependency of the function $f$ on $h$ makes such algorithms impractical for all but small values of $h$. By contrast, in \textsc{Tree Containment} the input tree may be assumed to typically have a size comparable to the overall input size.
In the directed case, even the definition of ``topological minor'' has been contested~\citep{GHK+16} and we are aware of little to no algorithmic results.
In \textsc{Tree Containment}, part of the embedding of the host tree in the guest network is fixed by the leaf-labeling.
If the node-mapping is fixed for \emph{all} nodes of the host, then directed TMC generalizes the \textsc{Disjoint Paths} problem~\citep{FHW80},
which is NP-complete for 2 paths or in case the host network is acyclic.
Indeed, one can show \textsc{Tree Containment} to be NP-hard in a similar fashion~\citep{kanj2008seeing}.

\subsection{Our contribution}
\looseness=-1
In this paper, we present
an FPT-algorithm for \textsc{Tree Containment} parameterized by the treewidth of the input network. Our algorithm is one of the first (constructive) FPT-algorithms for a problem in phylogenetics parameterized by treewidth. We believe that this is an important development as the treewidth can be much smaller than other parameters such as reticulation number and level which are easier to work with. We see this algorithm as an important step towards the wide application of treewidth-based methods in phylogenetics.

\section{Preliminaries}

\subsection{Reformulating the problem}

\looseness=-1
A key concept throughout this paper will be \emphmore{display graphs}~\citep{BryantLagergren2006}, which are the graphs formed from the union of a phylogenetic tree  and a phylogenetic network by identifying leaves with the same labels. 
Throughout this paper we will let $\Nin$ and $\Tin$ denote the respective input network and tree in our instance of {\sc Tree Containment}.
The main object of study will be the ``display graph'' of $\Nin$ and $\Tin$.
For the purposes of our dynamic programming algorithm, we will often consider graphs that are not exactly this display graph,
but may be thought of as roughly corresponding to subgraphs of it (though they are not exactly subgraphs; see \cref{sec:isolabelling}).
In order to incorporate such graphs as well,
we will define display graphs in a slightly more general way than that usually found in the literature.
In particular, we allow for the two ``sides'' of a display graph to be disconnected,
and for some leaves to belong to one side but not the other.

\begin{definition}[display graph]
A \emph{display graph} is a 
directed acyclic graph $D = (V,A)$, with specified subsets $V_T,V_N\subseteq V$ such that $V_T\cup V_N = V$, satisfying the following properties:
\begin{compactitem}
    \item The graph $T := D[V_T]$ is an out-forest;
    \item Every vertex has in- and out-degree at most $2$ and total degree at most $3$;
    \item Any vertex in $V_N\cap V_T$ has out-degree $0$ and in-degree at most $1$ in both $T$ and $N := D[V_N]$.
\end{compactitem}
Herein, we call $T$ the \emph{tree side} and $N$ the \emph{network side} of $D$ and
we will use the term~$D(N,T)$
to denote a display graph with network side $N$ and tree side $T$.
\end{definition}

Given a phylogenetic network~$\Nin$ and phylogenetic tree~$\Tin$ with the same leaf-label set,
we define $\NinTin$ to be the display graph formed by taking the disjoint union of $\Nin$ and $\Tin$ and
identifying pairs of leaves that have the same label.
We note that, while the leaves of $\Nin$ and $\Tin$ were originally labelled, this labelling does not appear in $\NinTin$. Labels are used to construct $\NinTin$, but in the rest of the paper we will not need to consider them. Indeed, such labels are relevant to the \textsc{Tree Containment} problem only insofar as they establish a relation between the leaves of $\Tin$ and $\Nin$, and this relation is now captured by the structure of $\NinTin$.

We now reformulate the \textsc{Tree Containment} problem in terms of an \emphmore{embedding function} on a display graph.
Unlike the standard definition of an embedding function (see, e.g., \citep{van2010locating}),
which is defined for a phylogenetic network~$N$ and tree~$T$, our definition of an embedding function applies directly to the display graph $D(N,T)$.
Because of our more general definition of display graphs, our definition of an embedding function will also be more general than that found in the literature. The key idea of an embedding function remains the same, however:
it shows how a subdivision of $T$ may be viewed as a subgraph of $N$.

\begin{definition}[embedding function]\label{def:embed}
Let $D$ be a display graph with network side $N$ and tree side $T$, and
let $\pathSet(N)$ denote the set of all directed paths in $N$.
An \emph{embedding function on $D$} is a function $\emb:V(T)\cup A(T) \to V(N) \cup \pathSet(N)$ such that:
\begin{compactenum}[(a)]
    \item\label{it:node&path}
        for each $u \!\in\! V(T)$, $\emb(u) \!\in\! V(N)$ and,
        for each $uv \!\in\! A(T)$, $\emb(uv)$ is a directed $\emb(u)$-$\emb(v)$-path in $N$;
    \item\label{it:injective}
        for any distinct $u,v \in V(T)$, $\emb(u)\neq \emb(v)$;
    \item\label{it:leaves} for any $u \in V(T) \cap V(N)$, $\emb(u) = u$;
    \item\label{it:paths disjoint}
        the paths $\{\emb(uv) \mid uv\in A(T)\}$ are arc-disjoint;
    \item\label{it:shared nodes}
        for any distinct $p,q\in A(T)$, $\emb(p)$ and $\emb(q)$ share a vertex $z$
        only if $p$ and $q$ share a vertex $w$ with $z = \emb(w)$;
\end{compactenum}
\end{definition}
Note that the standard definition of an embedding of a phylogenetic tree~$T$
into a phylogenetic network~$N$ (see e.g.~\citep{van2010locating}) coincides with the definition
of an embedding function on $D(N,T)$.
Property~(\ref{it:shared nodes}) ensures that, while the embeddings of arcs $uv,vw_1,vw_2$ can all meet at $\emb(v)$, the embeddings of different tree arcs cannot otherwise meet. (In particular, the path $\emb(uv)$ cannot end at a reticulation that is also an internal vertex of $\emb(u'v')$, something that is otherwise allowed by properties~(\ref{it:node&path})--(\ref{it:paths disjoint}).)

\begin{lemma}\label{lem:display equiv}
  A phylogenetic network~$\Nin$ displays a phylogenetic tree~$\Tin$ if and only if there is an embedding function on $\NinTin$.
\end{lemma}
\begin{proof}
Suppose first that $\Nin$ displays $\Tin$, that is,
there is a subgraph~$\Tin'$ of $\Nin$ that is a subdivision of $\Tin$.
Then, every vertex of $\Tin$ corresponds to a vertex of $\Tin'$ and
every arc $uv$ in $\Tin$ corresponds to a directed path in $\Tin'$
between the vertices corresponding to $u$ and $v$.
Let $\emb$ denote this correspondence relation.
Then, \eqref{it:leaves} of \cref{def:embed} follows from the definition of $\NinTin$ while
the other properties follow from $\emb$ being an isomorphism (of a subdivision of $\Tin$ into $\Tin'$).

Conversely, suppose that there is an embedding function~$\emb$ on $\NinTin$ and
let $\Tin'$ be the subgraph formed by the arcs of $\Nin$ that are part of some path $\emb(uv)$ for an arc $uv$ in $T$.
Then, it can be verified that $\Tin'$ is a subdivision of $\Tin$.
\end{proof}

\noindent
In light of \cref{lem:display equiv},
we may henceforth view \textsc{Tree Containment} as the following problem:

\myboxprobdef[Task]%
{phylogenetic network $\Nin$ and phylogenetic tree $\Tin$ with the same leaf-label set}%
{Find an embedding function on $\NinTin$.}%
{\textsc{Tree Containment}}{TC}{def:TC}%

\subsection{Overview of our approach}\label{sec:approach}

We study \textsc{Tree Containment} parameterized by the treewidth of the input network $\Nin$.
A key tool will be the following theorem
from~\citep{Janssen2019Treewidth}. In this theorem, the display graph $D_{\textsc{in}}(N_u,T_u)$ for unrooted~$N_u$ and~$T_u$ is defined analogously to $\NinTin$ for rooted $N_{\textsc{in}}, T_{\textsc{in}}$ -- that is, it is the (undirected) graph derived from the disjoint union of~$N_u$ and~$T_u$ by identifying leaves with the same label.

\begin{theorem}[\citep{Janssen2019Treewidth}]\label{thm:tw bound}
Let~$N_u$ and~$T_u$ be an unrooted binary phylogenetic network and tree, respectively, with the same leaf-label set. If $N_u$ displays $T_u$ then $tw(D_{\textsc{in}}(N_u,T_u))\leq 2tw(N_u) + 1$.
\end{theorem}
%
By \cref{thm:tw bound}, we suppose that the display graph $\NinTin$ has treewidth at most $2k+1$, where $k$ is the treewidth of the underlying undirected graph~$N_u$ of~$\Nin$, as otherwise $tw(\NinTin) = tw(D_{\textsc{in}}(N_u,T_u)) > 2k+1$ and $N_u$ does not display the unrooted version~$T_u$ of~$\Tin$, implying that~$\Nin$ does not display~$\Tin$.

\looseness=-1
As is often the case for treewidth parameterizations, we will proceed via a dynamic programming on a tree decomposition, in this case a tree decomposition of $\NinTin$.
Recall that we view a bag $(P,S,F)$ in the tree decomposition as partitioning the vertices of $\NinTin$ into past, present and future.
A typical dynamic programming approach is to store, 
for each bag,
some set of information about the present, while forgetting most information about the past, and not yet caring about what happens in the future. The resulting information is stored in a ``signature'', and the algorithm works by calculating which signatures are possible on each bag, in a bottom-up manner.
This approach is complicated by the fact that the sought-for embedding of $\Tin$ into $\Nin$ may not map the past/present/future of $\Tin$ into the past/present/future (respectively) of $\Nin$. Vertices from the past of $\Tin$ may be embedded in the future of $\Nin$, or vice versa.
Thus, we have to store more information than we might at first think. 
In particular, it is not enough to store information about which present vertices of $\Tin$ are embedded in which present vertices of $\Nin$
(indeed, depending on the bag, it may be that none of them are).
As such, our notion of a ``signature'' has to track how vertices from the past of $\Tin$ are embedded in the present and future of $\Nin$, and which vertices from the past of $\Nin$ contain vertices from the present or future of $\Tin$. Vertices of the past which are mapped to vertices of the past, on the other hand, can mostly be forgotten about.

\subsection{An informal guide to (compact) signatures}\label{sec:signatureOverview}

Roughly speaking, a \emph{signature}~$\sigma$ for a bag~$(P,S,F)$ in the tree decomposition of $\NinTin$
consists of the following items (see \cref{fig:containment struct} for an example):
\begin{enumerate}
  \item a display graph $D(N,T)$, some of whose vertices correspond (isomorphically) to $S\subseteq V(\NinTin)$,
    and the rest of which are labeled \past or \future
    (which we may think of as vertices corresponding to some vertex of $\NinTin$ in $P$ or $F$, respectively).
    We use a function~$\iota$ \mjnew{on $V(D(N,T))$} to capture both this correspondence and labelling,
    \mjnew{where~$\iota$ maps each vertex to an element of $S$ or a label from $\{\past, \future\}$.} We emphasize here that $D(N,T)$ is distinct from the input display graph $\NinTin$.
  \item an embedding~$\emb$ of $T$ in $N$ such that, for no arc~$uv$ of $T$,
    all of \mjnew{$V(\emb(uv))\cup\{u,v\}$} have the same label $y \in \{\past, \future\}$ under $\iota$.
\end{enumerate}
Signatures may be seen as ``partial embedding functions'' on parts of~$\NinTin$ in a straightforward way.
In particular, we call~$\sigma$ \emph{valid} for $(P,S,F)$ if, roughly speaking, $\emb$~corresponds (via $\iota$) to something
that can be extended to an embedding function on the subgraph of $\NinTin$ induced by the vertices~$P\cup S$ introduced below $(P,S,F)$.
In our dynamic programming algorithm, we build valid signatures for a bag~$x$ from valid signatures of the child bag(s) of $x$
(in particular, validity 
of a signature for $x$ is characterized by the validity of certain signatures for the child bag(s)).

Since the definition of a signature does not guarantee any bound on the number of vertices labeled \past or \future, iterating over all signatures for a bag~$(P,S,F)$ (in order to check their validity) exceeds FPT time.
Therefore we will instead consider ``compact'' signatures, whose number and size are bounded in the width~$|S|$ of the bag~$(P,S,F)$.
If $\NinTin$ admits an embedding function~$\emb^*$, then a compact signature corresponding to this embedding function exists.
In the following, we informally describe the compaction process for this hypothetical solution~$\emb^*$,
thus giving a rough idea of the definition of a ``compact'' signature.
At all times, the (tentative) signature will contain a display graph $D(N,T)$ (initially $D(N,T) = \NinTin$),
and an embedding function of $T$ into~$N$ (initially~$\emb^*$).
For a more complete description of our approach, see \cref{sec:filling} and,
for an illustration, see \cref{fig:SigExample}. 

\begin{figure}[t]
  \tikzstyle{embed}=[line width=2.2pt]%
  \tikzstyle{hl}=[gray, fill=lightgray, rounded corners=8pt]%
  \centering
  \begin{tikzpicture}[yscale=.6,xscale=.7]
    \node[smallvertex] (Nrt) {};
    \foreach [count=\j from 0] \leaves in {{B/A,C/D},{F/E,G/H}} {
      \nextnode{N\j}{Nrt}{-150+120*\j:3}{revarc,embed}
      \foreach [count=\i from 0] \lone/\ltwo in \leaves {
        \nextnode{N\j\i}{N\j}{-135+90*\i:1}{revarc,embed}
        \nextnode{N\j\i\i}{N\j\i}{-135+90*\i:1}{revarc,embed}
        \nextnode[reti]{Nr\j\i}{N\j\i\i}{-45-90*\i:1}{revarc,embed}
        \nextnode[leaf]{\lone}{Nr\j\i}{-90:.9}{revarc,embed}
        \node[smallleaf] (\ltwo) at ($(\lone)-(1.44-2.88*\i,0)$) {} edge[revarc,embed] (N\j\i\i);
        \nextnode{T\j\i}{\lone}{-135+90*\i:1}{arc}
        \draw[arc] (T\j\i) -- (\ltwo);
      }
      \nextnode{T\j}{T\j0}{-45:2}{arc}
      \draw[arc] (T\j) -- (T\j1);
      \foreach \u/\v in {0/1, 1/0} \draw[arc] (N\j\u) -- (Nr\j\v) {};
    }
    \nextnode{Trt}{T0}{-30:3}{arc}
    \draw[arc] (Trt) -- (T1);
    \node at ($(N0)+(180:1)$) {$F$};
    \node at ($(T1)+(0:1)$) {$P$};
    \node at ($(Nrt)+(-90:1.2)$) {$\Nin$};
    \node at ($(Trt)+(90:1.2)$) {$\Tin$};
    \begin{pgfonlayer}{background}      
      \draw[hl, dashed]
          ($(Nrt)+(90:.5)$) --
          ($(N0)+(135:.4)$) -- ($(N000)+(135:.4)$) -- ($(A)+(180:.5)$) -- ($(T0)+(-90:.6)$) -- ($(D)+(0:.5)$) -- ($(N011)+(45:.4)$)
          -- ($(N0)+(-10:.6)$) -- ($(Nrt)+(-90:.6)$) -- ($(N1)+(-170:.6)$) --
          ($(N100)+(180:.5)$) -- ($(Nr10)+(-90:.4)$) -- ($(Nr11)+(-90:.4)$) -- ($(N111)+(0:.5)$) -- ($(N1)+(45:.4)$) -- cycle;
      \draw[hl, dashed]
          ($(T1)+(-90:.6)$) -- ($(T10)+(135:.3)+(-135:.4)$) -- ($(T10)+(135:.3)+(45:.4)$) --
          ($(T1)+(90:.7)$) -- ($(T11)+(45:.3)+(135:.4)$) -- ($(T11)+(45:.3)+(-45:.4)$) -- cycle;
    \end{pgfonlayer}
  \end{tikzpicture}
  \hfill
  \begin{tikzpicture}[yscale=.6,xscale=.7]
    \node[smallvertex] (Nrt) {};
    \nextnode{N0}{Nrt}{-150:3}{revarc,embed}
    \nextnode{N1}{Nrt}{-30:3}{revarc,embed}
    \foreach [count=\i from 0] \lone/\ltwo in {F/E, G/H} {
      \nextnode{N1\i\i}{N1}{-135+90*\i:2}{revarc,embed}
      \nextnode[reti]{Nr1\i}{N1\i\i}{-45-90*\i:1}{revarc,embed}
      \nextnode[leaf]{\lone}{Nr1\i}{-90:.9}{revarc,embed}
      \node[smallleaf] (\ltwo) at ($(\lone)-(1.44-2.88*\i,0)$) {} edge[revarc,embed] (N1\i\i);
      \nextnode{T1\i}{\lone}{-135+90*\i:1}{arc}
      \draw[arc] (T1\i) -- (\ltwo);
    }
    \nextnode{T1}{T10}{-45:2}{arc}
    \draw[arc] (T1) -- (T11);

    \nextnode{Trt}{T1}{-150:3}{arc}
    \nextnode{T0}{Trt}{150:3}{revarc}
    \node[anchor=west] at ($(N1)+(110:1)$) {\future};
    \node[anchor=east] at ($(T1)+(150:1.4)$) {\past};
    \begin{pgfonlayer}{background}      
      \path[hl]
          ($(Nrt)+(90:.5)$) -- ($(N0)+(-150:.3)+(-150-90:.4)$) -- ($(N0)+(-150:.3)+(-150+90:.4)$)
          -- ($(Nrt)+(-90:.6)$) -- ($(N1)+(-170:.6)$) --
          ($(N100)+(180:.5)$) -- ($(Nr10)+(-90:.4)$) -- ($(Nr11)+(-90:.4)$) -- ($(N111)+(0:.5)$) -- ($(N1)+(45:.4)$) -- cycle;
      \path[hl]
        ($(T1)+(-90:.6)$) -- ($(T10)+(135:.3)+(-135:.4)$) -- ($(T10)+(135:.3)+(45:.4)$) --
        ($(T1)+(90:.7)$) -- ($(T11)+(45:.3)+(135:.4)$) -- ($(T11)+(45:.3)+(-45:.4)$) -- cycle;
      \node[draw=white, fill=lightgray, rounded corners=3pt, inner sep=8pt,label=above:{\future}] at (T0) {};
    \end{pgfonlayer}
  \end{tikzpicture}
  \caption{%
  \textbf{Left}:
  An example of a display graph $\NinTin$ for which $\Nin$ displays $\Tin$ as witnessed by the embedding function $\emb$ that is indicated by bold edges.
  Highlighting with dashed border represents the sets $P$ and $F$, for some bag $(P,S,F)$ in a tree decomposition of $\NinTin$.
  \textbf{Right}:
  A representation of the (compact) signature for $(P,S,F)$ derived from this solution.
  Vertices labelled $\past$ or $\future$ are highlighted in gray without border.}
  \label{fig:SigExample}
\end{figure}
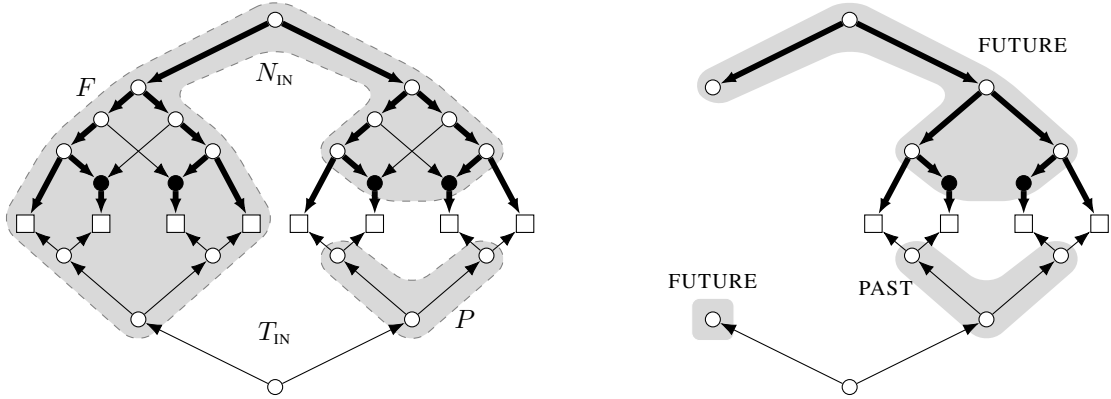

\begin{compactdesc}
  \item[Step~1]
    After initialization with $\emb^*$,
    we assign a label~\future to all vertices of $F$,
    and a label~\past to all vertices in $P$
    (Observe that no vertex labeled~\past will be adjacent to a vertex labeled~\future, since $S$ separates $P$ from $F$ in $\NinTin$).
    Then, we ``forget'' which vertices of $\NinTin$ the vertices labelled $\past$ or $\future$ correspond to.
    Our preliminary signature now contains 
    (1) a display graph~$D(N,T)$ whose vertices are either labelled $\future$ or $\past$ or correspond (isomorphically) to vertices in~$S\subseteq V(\NinTin)$
    (we refer the reader to \cref{sec:isolabelling} for a more formal description), as well as
    (2) an embedding function for $D(N,T)$ into $N$.

  \item[Step~2] \looseness=-1
    We now simplify the structure of the preliminary signature.
    The main idea is that, 
    if $a$ is an arc of $T$ with both endpoints labelled \past and
    all vertices in the path $\emb(a)$ are also labelled \past,
    then we can safely forget $a$ and all the arcs in $\emb(a)$.
    Intuitively, the information that $a$ will be embedded in $\emb(a)$  does not have any effect
    on the possible solutions one could construct on the part of $\NinTin$ that is ``above'' the bag~$(P,S,F)$.
    Similarly, we can forget any arc $a$ of $T$ whose endpoints, as well as every vertex in $\emb(a)$,
    are assigned the label \future.
    Intuitively, this is because this information should have no bearing on whether a solution exists with this signature
    for $\NinTin$ restricted to $P\cup S$.
    In a similar way, we forget any vertex $u\in V(T)$ and its embedding $\emb(u)$ if they are assigned the same label,
    provided that all their incident arcs can also be forgotten.
    We will call the vertices and arcs fulfilling these conditions ``redundant'' and we remove them from our tentative signature. 
    We can also safely delete the vertices and arcs of $N$ that are labelled $y \in \{\past, \future\}$ but are not part of the image of $\emb$.
    As a result, we now have that for any remaining vertex $u \in T$,
    either one of $\{u, \emb(u)\}$ is labelled \past and the other labelled \future,
    or some vertex element of $S$ must appear as either one of $\{u,\emb(u)\}$, a neighbor of $u$, or a vertex in the path $\emb(a)$
    for an incident arc $a$ of $u$.
    Thus, we have ``forgotten'' all the aspects of the embedding except those that involve vertices from the present in some way,
    or those where the embedding ``time-travels'' between the past and future
    (see \cref{sec:restriction} for a more formal description of this process).
 
  \item[Step~3]
    Finally, we may end up with long paths of vertices with in-degree and out-degree~$1$ that are labelled \past or \future in $N$
    (for example, if $u$ and $v$ are labelled \past, then $\emb(uv)$ may be a long path in~$N$ with all vertices labelled \future).
    Such long paths do not contain any useful information to us, we therefore compress these by suppressing vertices with in-degree and out-degree~$1$
    (This gives the \emph{compact signature}, see \cref{sec:compact}).
\end{compactdesc} 

\colorlet{pastcol}{gray!60!white}
\begin{figure}[tb]
    \centering
    \scalebox{.8}{%
    \begin{tikzpicture}[xscale=.9, yscale=.9]
        \tikzstyle{tofut}=[];
        \tikzstyle{iota}=[dotted];
        \tikzstyle{emb}=[gray, line width=6pt];
        \foreach \i/\n in {0/,10/In}{
            \pgfmathtruncatemacro{\fadeintensity}{(13-\i)}
            \tikzstyle{fade}=[opacity=\fadeintensity/10];
            \node[smallvertex] (Nrt\n) at (-\i,0) {};
            \nextnode{N0\n}{Nrt\n}{-45:1}{revarc}
            \nextnode{N00\n}{N0\n}{-135:1}{revarc}
            \nextnode{N01\n}{N0\n}{-45:1}{revarc}
            \nextnode[leaf]{la\n}{N00\n}{-100:1.32}{revarc}
            \nextnode[reti]{N001\n}{N00\n}{-45:1}{revarc}
            \nextnode[leaf]{lb\n}{N001\n}{-90:.6}{revarc}
            
            \nextnode[vertex, fade]{Nx\n}{Nrt\n}{45:1.1}{arc, fade}
            \nextnode[vertex, fade]{Nr\n}{Nx\n}{45:.9}{arc, fade}
            \nextnode[vertex, fade]{Nq\n}{Nx\n}{-47:1}{revarc, fade}
            \nextnode[vertex, fade]{Ny\n}{Nq\n}{-47:1}{revarc, fade}
            \nextnode[reti, fade]{Nz0\n}{Ny\n}{1,-.6}{revarc, fade}
            \nextnode[vertex, fade]{Nz1\n}{Ny\n}{.2,-.6}{revarc, fade}
            \node[smallvertex, fade] (Nz\n) at ($(Nrt\n)+(-135:1.53)$) {} edge[arc, fade] ($(Nz\n)+(-90:.6)$) edge[arc, fade] ($(Nz\n)+(-120:.6)$);
            
            \draw[arc, tofut, fade] (Nr\n) to ($(Nr\n)+(-47:1)$);
            \draw[arc, tofut, fade] (Nq\n) to[bend left=10] ($(Nq\n)+(-15:1.5)$);
            \foreach \i in {30,-30} \draw[arc, tofut, fade] (Nz1\n) -- ++(-90+\i:.5);
            \foreach \i/\s in {0/arc,120/revarc} \draw[\s, tofut, fade] (Nz0\n) -- ++(-90+\i:.5);
        
            \draw[arc, tofut] (N01\n) -- ++(-60:.7);
            \draw[arc, tofut] (Nrt\n) -- ++(-135:1.5);
            \draw[revarc, tofut] (Nrt\n) -- ++(45:1);
            \draw[arc] (N01\n) -- (N001\n);
        
            \nextnode{T00\n}{la\n}{-45:1}{arc}
            \nextnode{T0\n}{T00\n}{-60:1.5}{arc}
            \node[smallvertex] (Tr\n) at ($(T0\n)+(-150:1)$) {};
            \node[smallvertex] (Tx\n) at ($(T0\n)+(30:.9)$) {};
            \nextnode[vertex, fade]{Tz\n}{Tx\n}{.6,.6}{revarc, fade}

            \draw[arc] (T00\n) -- (lb\n);
            \draw[arc, tofut] (T0\n) -- ++(30:.8);
            \draw[revarc, tofut] (T0\n) -- ++(-150:1);
            \draw[arc, tofut, fade] (Tx\n) -- ++(100:.6);
            \foreach \i in {30,-30} \draw[arc, tofut, fade] (Tz\n) -- ++(90+\i:.5);
            \draw[arc, tofut, fade] (Tr\n) --++(130:.6);
            
            \foreach [count=\i from 0] \x in {3.2,4.2}
                \node[smallleaf] (l\i\n) at ($(la\n)+(\x,0)$) {}
                    edge[revarc] ($(l\i\n)+(100:.6)$) edge[revarc] ($(l\i\n)+(-70:.6)$);
            
            \coordinate (Nroot\n) at ($(Nrt\n)+(0,1)+(45:2)$);
            \coordinate (NTleft\n) at ($(Nroot\n)+(-130:6.7)$);
            \coordinate (NTright\n) at ($(Nroot\n)+(-50:6.7)$);
            \coordinate (Npast\n) at ($(Nrt\n)+(-140:1.1)$);
            \coordinate (NfutUp\n) at ($(Nrt\n)+(45:.7)$);
            \coordinate (NfutDn\n) at ($(N01\n)+(-45:.7)$);
            
            \coordinate (Tpast\n) at ($(T0\n)+(-135:1)$);
            \coordinate (Troot\n) at ($(Tpast\n)+(-90:.2)$);
            \coordinate (Tfut\n) at ($(T0\n)+(45:.8)$);
            
            \begin{pgfonlayer}{background}
                \draw[gray, thick, fill=lightgray, rounded corners] (Nroot\n) -- (NTleft\n) -- (NTright\n) -- cycle;
                \draw[gray, thick, fill=lightgray, rounded corners] (Troot\n) -- (NTleft\n) -- (NTright\n) -- cycle;
                
                \draw[gray, rounded corners, fill=pastcol] (Npast\n) -- ($(NTleft\n)+(.3,0)$) -- ($(NTleft\n)+(2,0)$) -- cycle;
                \draw[gray, rounded corners, fill=pastcol] (Tpast\n) -- ($(NTleft\n)+(.3,0)$) -- ($(NTleft\n)+(2,0)$) -- cycle;
                \draw[gray, rounded corners, fill=pastcol] ($(NfutUp\n)+(135:.2)$) -- ++(50:1.66) -- ($(NTright\n)+(-.8,.4)$) -- ($(NTright\n)+(-4.5,.4)$) -- ++(60:1.3) -- cycle;
                \draw[gray, rounded corners, fill=pastcol] ($(Tfut\n)+(-90:.6)$) -- ($(NTright\n)+(-.9,-.4)$) -- ($(NTright\n)+(-4.5,-.4)$) -- cycle;
            \end{pgfonlayer}
        }
        
        \begin{pgfonlayer}{background}
            \draw[emb] (Ny.center) -- (Nq.center) -- (Nx.center) -- (Nrt.center) -- (N0.center) -- (N00.center) -- (la.center);
            \draw[emb] (Nz0.center) -- (Ny.center) -- (Nz1.center);
            \draw[emb] (N00.center) -- (N001.center) -- (lb.center);
            \draw[emb] (Nr.center) -- (Nx.center);
            \draw[emb] (Nr.center) -- ++(-47:1);
            \draw[emb] (Nz0.center) -- ++(-90:.5);
            \foreach \i in {30,-30} \draw[emb] (Nz1.center) -- ++(-90+\i:.5);
            
            \draw[emb] (Tz.center) -- (Tx.center) -- (T0.center) -- (T00.center) -- (la.center);
            \draw[emb] (T00.center) -- (lb.center);
            \draw[emb] (Tz.center) -- (Tx.center);
            \draw[emb] (Tr.center) -- (T0.center);
            \draw[emb] (Tr.center) --++(130:.6);
            \draw[emb] (Tx.center) -- ++(100:.6);
            \foreach \i in {30,-30} \draw[emb] (Tz.center) -- ++(90+\i:.5);
            \foreach \i in {0,1}
                \draw[emb] ($(l\i.center)+(-70:.6)$) -- (l\i.center) -- ($(l\i.center)+(100:.6)$);
        \end{pgfonlayer}
        
        \node at ($(NrootIn)+(-170:.7)$) {$\Nin$};
        \node at ($(TrootIn)+(170:.7)$) {$\Tin$};
        \node at ($(Nroot)+(-170:.7)$) {$N$};
        \node at ($(Troot)+(170:.7)$) {$T$};
        \node at ($(Nrt)+(-4,1)$) {\LARGE$\iota$};
        \node[gray!80!black] at ($(N0)-(.6,0)$) {\LARGE $\emb$};
        \node[fill=pastcol, rounded corners] (past) at ($(Nroot)+(170:2)$) {\past};
        \node[fill=pastcol, rounded corners] (past2) at ($(Troot)+(10:3)$) {\past};
        \node[fill=pastcol, rounded corners] (future) at ($(Nroot)+(10:2)$) {\future};
        
        \foreach \u/\b in {Nrt/1, N0/1, N00/1, N01/1, N001/1, la/1, lb/1, T0/-1, T00/-1}
            \draw[thick,iota] (\u) to[bend right=\b*10] ($(\u)-(10,0)$);
        \draw[line width=1.5pt,iota] ($(NTleft)+(1.2,.5)$) to[bend left=10] (past);
        \draw[line width=1.5pt,iota] ($(NTright)+(-2,-.7)$) to[bend left=10] (past2);
        \draw[line width=1.5pt,iota] ($(NTright)+(-2,1)$) to[bend right=10] (future);
    \end{tikzpicture}
    }
    \caption{Example of a signature of a bag $(P,S,F)$. 
    The $S$-part of $D(\Nin,\Tin)$ is solid while the non-$S$ part is faded.
    The embedding $\emb$ (right, \leo{indicated with} gray edge-highlight) maps $T$ into $N$. The dotted arcs labelled $\iota$ show the isomorphism between part of $D(N,T)$ and $S\subseteq V(D(\Nin,\Tin))$. Note that the part of $D(N,T)$ that is not mapped to $S$ is not necessarily isomorphic to anything in $D(\Nin,\Tin)$.
    \label{fig:containment struct}}
\end{figure}
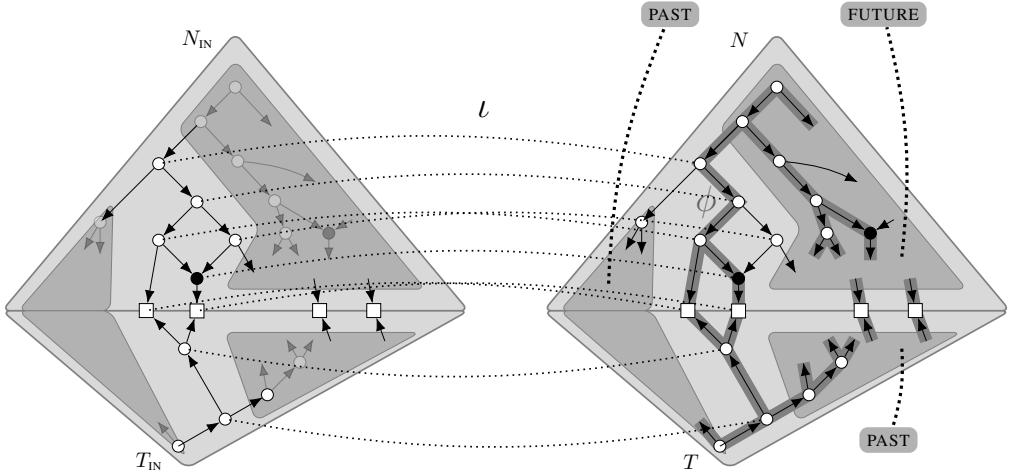

\subsection{Bounding the number of signatures}

We now outline the main arguments for why the number of possible (compact) signatures for a given bag~$(P,S,F)$ can be bounded in~$|S|$.
Such a bound on the number of signatures ensures that the running time of the algorithm is FPT,
because the number of calculations required for each bag is bounded by a function of the treewidth.

The main challenge is to bound the size of the display graph $D(N,T)$ in a given signature for $(P,S,F)$. Once such a bound is achieved, this immediately implies upper bounds (albeit quite large) for the number of possible display graphs and the number of possible embeddings, and hence on the number of possible signatures.
We will focus here on bounding the size of the tree part $T$.
Once a bound is found for $|T|$ it is relatively straightforward to use that to give a bound on $|N|$ (because the arcs of $N$ that are not used by the embedding of $T$ into $N$ are automatically deleted, unless they are themselves incident to a vertex in $S$, and because isolated vertices are deleted and long paths suppressed).

It can be seen that a vertex $u \in V(T)$ is redundant (and so would be deleted from the signature) unless one of the following properties holds:
\begin{inparaenum}[(1)]
\item $u \in S$, 
\item $\emb(u) \in S$,
\item $u$ is incident to an element of $S$
\item for some arc $a$ incident to $u$, the path $\emb(a)$ contains a vertex from $S$ or
\item $u$ and $\emb(u)$ have different labels from $\{\past, \future\}$.
\end{inparaenum}
Essentially if none of (1)--(4) holds, then all the vertices mentioned in those properties have the same label as either $u$ or $\emb(u)$, using the fact that $S$ separates the vertices labelled \past from the vertices labelled \future. If $u$ and $\emb(u)$ have the same label, then all these vertices have the same label, which is enough to show that $u$ is redundant.
It remains to bound the number of vertices satisfying one of these properties.
For the first four properties, it is straightforward to find a bound in terms of $|S|$.
The vertices satisfying the final property are ``time-travelling'' (in the sense that either $u$ is labelled \past and $\emb(u)$ \future, or $u$ is labelled \future and $\emb(u)$ \past).
Because of the bounds on the other types of vertices, it is sufficient to provide a bound on the number of \emph{lowest} time-travelling vertices in $T$.

To see the intuition why this bound should hold: consider some full solution on the original input, i.e. an embedding function on $\NinTin$, 
and suppose $u \in V(\Tin)$ is a lowest tree vertex for which $u \in P, \emb(u) \in F$ (thus in the corresponding signature, $u$ has label \past and $\emb(u)$ has label \future). Let $x \in V(\Nin)\cap V(\Tin)$ be some leaf descendant of $u$. Then there is path in $\Tin$ from $u$ to $x$, and a path in $\Nin$ from $\emb(u)$ to $\emb(x)=x$. Thus $\NinTin$ has an (undirected) path from $u$ to $\emb(u)$. As this is a path between a vertex in $P$ and a vertex in $F$, some vertex on this path must be in $S$ (since $S$ separates $P$ from $F$). Such a path must exist for every lowest time-travelling vertex $u$, and these paths are distinct. The existence of these paths can then be used to bound the number of lowest time-travelling vertices.

\section{Filling in the details}\label{sec:filling}

\subsection{Tracking the identity of vertices}\label{sec:isolabelling}

In the following, consider a bag~$(P,S,F)$ and
recall that the sought solution may embed parts of the past of $T$ in the future of $N$ or vice versa.
Thus, as previously indicated, a signature for $(P,S,F)$ will have to track information
about more vertices than those in $S$.
This leads to the following complexities in how we talk about the identity of vertices in signatures.

\looseness=-1
For the display graph $D(N,T)$ in a signature for $(P,S,F)$,
we want some vertices to correspond to specific vertices in $S\subseteq V(\NinTin)$
while, for other vertices we want to express the fact that they correspond to a vertex of $P$ or $F$,
without specifying which one.
Moreover, as our dynamic programming proceeds up the tree decomposition and we move from one signature to another, vertices will change between these two states.
When we move from a child bag~$(P_c,S_c,F_c)$ to a parent bag~$(P_p,S_p,F_p)$ and
a vertex~$x$ of $\NinTin$ ``leaves the bag'' (that is, $x\in S_c\setminus S_p$) then 
we want to ``forget'' that a vertex $v$ corresponds to $x$ (instead labelling $v$ \past),
while still tracking its adjacencies and embedding information.
Similarly, a vertex labelled \future may, at some point, 
come to correspond to a particular vertex in $S$.

In order to keep this information straight, we have to be careful in how we talk about the identity of these vertices.
In particular, the vertices of the display graph~$D(N,T)$ are different from those in $V(\NinTin)=P\uplus S\uplus F$,
but the signature carries a ``correspondence'' function~$\iota$ that, on some subgraph of $D(N,T)$,
acts as an isomorphism into the ``$S$-part'' of $\NinTin$ while, on other vertices,
acting as a labelling that assigns some vertices the label~\past and others the label~\future.
As such, we refer to $\iota$ as an ``isolabelling''.

\looseness=-1
In the next definition, we consider an arbitrary set $\Y$ of labels, rather than the set $\{\past, \future\}$.
While we will often have $\Y=\{\past,\future\}$, we also use isolabelings in other contexts besides signatures.

\begin{definition}[isolabelling]\label{def:isolabelling}
Let $\Y$ be a set of labels
and let $S\subseteq V(\NinTin)$.
An \emph{$(S,\Y)$-isolabelling}
on a display graph $D(N,T)$ is a function $\iota: V(D(N,T))\to S \cup \Y$ such that $S$ is a subset of the image of $\iota$
($\iota$ is ``surjective onto $S$'')
\footnote{Every $v\in S$ has a $u \in V(D(N,T))$ with $\iota(u) = v$,
but not every $y \in \Y$ may have a $u \in V(D(N,T))$ with $\iota(u) = y$.}
and, for any $u,v \in V(D(N,T))$ with $\iota(u), \iota(v) \in S$: 
\begin{compactenum}[(a)]
    \item\label{it:NT match}
        $\iota(u) \in V(\Nin)$ only if $u \in V(N)$ and $\iota(u) \in V(\Tin)$ only if $u \in V(T)$,
    \item\label{it:injectiveisol}
        $\iota(u) = \iota(v)$ only if $u = v$, and
    \item\label{it:isomorph}
        the arc $uv$ is in $D(N,T)$ if and only if $\iota(u)\iota(v)$ is in $\NinTin$.
\end{compactenum}
\end{definition}

\subsection{Signatures and containment structures}\label{sec:containmentStructure}

Before formally defining a signature for a bag $(P,S,F)$, let us remark that,
later in the paper, we will also discuss constructs which are similar in construction to signatures,
but with slightly different vertex and label sets (``partial solutions'' and ``reconciliations'').
For this reason, we define a more general structure, called an $(S,\Y)$-containment structure,
that encapsulates all these notions.
In the definition that follows, a signature for a bag $(P,S,F)$ corresponds to an $(S,\Y)$-containment structure,
where $S$~is the present of the bag $(P,S,F)$ and $\Y$ is the label set~$\{\past, \future\}$. 
For an intuitive understanding of a containment structure, the most important features are the display graph,
embedding function and isolabelling.

\begin{definition}[containment structure]
\label{def:containment struct}
    Let $\Y$ be a set of labels 
    and let $S\subseteq V(\NinTin)$.
    An \emph{$(S,\Y)$-containment structure} is a tuple $(D(N,T), \emb, \iota)$ consisting of\nobreakpar
    \begin{compactitem}
        \item a  display graph $D(N,T)$,  
        \item an embedding function $\emb$ on $D(N,T)$, and
        \item an $(S,\Y)$-isolabelling
        $\iota:V(D(N,T))\to 
        S \cup \Y$.
    \end{compactitem}
    In addition, each vertex~$u$ of $D(N,T)$ should satisfy the following properties:
    \begin{compactenum}[(a)]
        \item if $\iota(u)\in S$, then $u$ has the same in- and out-degree in $D(N,T)$ as $\iota(u)$ in $\NinTin$.
        \item if $u\in T$ and
            $\iota(u) \neq \iota(\emb(u))$,
            then $u$ has $2$ out-arcs in $D(N,T)$.
    \end{compactenum}
\end{definition}

\noindent
See \cref{fig:containment struct} for an example of an $(S,\{\past, \future\})$-containment structure $(D(N,T),\emb,\iota)$.
The last two properties of \cref{def:containment struct} are used for bookkeeping, and to help bound the size of the display graph.

\begin{definition}[signature]\label{def:signature}
Let $(P,S,F)$ be a bag in the tree decomposition of $\NinTin$.
Then any $(S, \{\past, \future\})$-containment structure is called a
\emph{signature for $(P,S,F)$}.
\end{definition}

Through most of the paper, $\Y = \{\past, \future\}$ and $S$ may be referred to as the ``present'' of the bag $(P,S,F)$.
However, our setup also allows us to talk about ``partial solutions'', where $S$ is replaced by $P \cup S$
(essentially merging the past and the present)
and only the label set $\Y=\{\future\}$ (corresponding to vertices in $F$) is used.
In some instances, we will use additional auxiliary labels; in particular when considering Join bags in \cref{sec:valid sigs}, we use the labels~\pleft and \pright to distinguish between the pasts of different bags.
Finally, \cref{def:containment struct} also allows capturing a solution for an instance of \textsc{Tree Containment}.
Indeed, if we replace the set $S$ with $V(\NinTin)$ and require that no vertex is mapped to a label of $\Y$,
then $D(N,T)$ is isomorphic to $\NinTin$ and the embedding function~$\emb$ of $T$ into $N$
also describes an embedding of $\Tin$ into $\Nin$.

\begin{lemma}\label{lem:TC equiv}
$(\Nin,\Tin)$ is a \textsc{Yes}-instance of {\sc Tree Containment}
if and only if
there is a $(V(\NinTin), \Y)$-containment structure $(D(N,T), \emb, \iota)$ with ${\iota}^{-1}(\Y) = \emptyset$.
\end{lemma}
\begin{proof}
First, suppose $(D(N,T), \emb, \iota)$ is a $(V(\NinTin), \Y)$-containment structure with ${\iota}^{-1}(\Y) = \emptyset$.
Then, for every vertex $u$ of $D(N,T)$, $\iota(u)$ is a vertex of $\NinTin$.
Moreover, for every vertex $v$ of $\NinTin$, there is $u \in V(D(N,T))$ such that $v = \iota(u)$ (see \cref{def:containment struct}).
Thus, $\iota$ is a bijective function between $V(D(N,T))$ and $V(\NinTin)$ and,
by \cref{def:isolabelling}\eqref{it:isomorph}, $uv$ is an arc in $D(N,T)$ if and only if $\iota(u)\iota(v)$ is an arc in $\NinTin$.
Thus, $\iota$ is an isomorphism, implying that $D(N,T)$ is isomorphic to $\NinTin$.
Moreover $N$ is isomorphic to $\Nin$ and $T$ is isomorphic to $\Tin$ (as $\iota$ maps vertices of $N$ to vertices of $\Nin$ and vertices of $T$ to vertices of $\Tin$ by \cref{def:isolabelling}\eqref{it:NT match}).
As $\emb$ is an embedding function on $D(N,T)$,
combining $\iota$ with $\emb$ yields an embedding function on $\NinTin$.

Conversely, suppose that there is an embedding function $\emb$ on $\NinTin$.
Let $\iota$ be the identity function on $V(\NinTin)$ and note that,
by \cref{def:isolabelling}, $\iota$ is a $(V(\NinTin), \Y)$-isolabelling on $V(\NinTin)$ and ${\iota}^{-1}(\Y) = \emptyset$.
Then, $(\NinTin, \emb, \iota)$ satisfies the first three conditions of an $(S,\Y)$-containment structure.
The remaining properties, concerning the degrees of vertices, follow from the fact that $\Nin$ and $\Tin$ are binary.
In particular, for any vertex $u$ in $\Tin$ with $\iota(u) \neq \iota(\emb(u))$, it holds that $u \neq \emb(u)$ and so $u$ is not a leaf of $\Tin$ and, thus, has out-degree~$2$ in $\NinTin$.
\end{proof}

\subsection{Formally defining the restriction}\label{sec:restriction}

While \cref{sec:signatureOverview} describes in broad strokes how a signature for a bag~$(P,S,F)$
could be derived from a solution for an instance of \textsc{Tree Containment}, we now make this process more precise. 
Recall that we first relabelled the vertices in $P$ and $F$ by \past and \future, respectively,
and then, ``redundant'' parts of the display graph were removed and the remaining long paths contracted.
Exactly which arcs and vertices should be removed is made precise in the notion of \emph{redundancy} defined below.

\begin{definition}[redundant]\label{def:redundant}
Let~$\chi := (D(N,T), \emb, \iota)$ be an $(S, \Y)$-containment structure, and let $y \in \Y$.
Then, we define the \emph{$y$-redundant} arcs and vertices of $D(N,T)$ as follows:
\begin{compactitem}
    \item A tree arc $uv \in A(T)$ is $y$-redundant if $\iota(u)=\iota(v) = y$ and $\iota(v') = y$ for all vertices $v'$ in the path $\emb(uv)$.
    \item A network arc $u'v' \in A(N)$ is $y$-redundant if $\iota(u')=\iota(v') = y$ and, 
   if~$u'v'$ is in the path $\emb(uv)$ for some tree arc~$uv\in A(T)$, 
    $uv$ is $y$-redundant.
    \item A tree vertex $v$ in $V(T)$ is $y$-redundant if
    $\iota(v) = \iota(\emb(v)) = y$, and
    $v$ and $\emb(v)$ are incident only to $y$-redundant arcs in $D(N,T)$.
    \item A network vertex $v'\in V(N)$ is $y$-redundant if
    $\iota(v') = y$ and $v'$ is incident only to $y$-redundant arcs and, 
   if $v'=\emb(v)$ for some  $v\in V(T)$,  $v$ is $y$-redundant.
\end{compactitem}
We say that an arc or vertex of $D(N,T)$ is \emph{redundant} if it is $y$-redundant for some $y \in \Y$.
When it is important to specify the containment structure~$\chi$, we say an arc or vertex is \emph{$y$-redundant with respect to $\chi$.}
\end{definition}
Just as we derived a signature from a solution in \cref{sec:signatureOverview}
by restricting our attention to a subset of vertices (in that case, the set $S$),
we can restrict any $(S,\Y)$-containment structure to an $(S',\Y)$-containment structure for any $S' \subseteq S$.
This will be a useful tool for deriving signatures for one bag from signatures for another bag,
as well as characterizing the ``validity'' of a signature.

\colorlet{pastcol}{gray!60!white}
\tikzstyle{pastarea}=[draw, gray, rounded corners, fill=pastcol]
\tikzstyle{treearea}=[gray, thick, fill=lightgray, rounded corners]
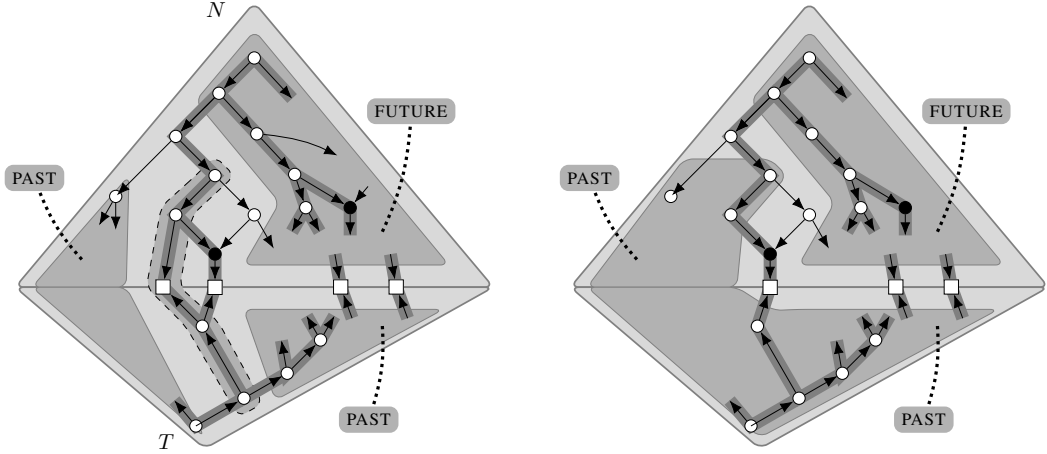
\begin{figure}[tb]
    \centering
    \scalebox{.82}{%
    \begin{tikzpicture}[xscale=.9, yscale=.9]
        \tikzstyle{tofut}=[];
        \tikzstyle{iota}=[dotted];
        \tikzstyle{emb}=[gray, line width=6pt];
        \foreach \i/\n in {0/b,10/}{
            \pgfmathtruncatemacro{\laopacity}{\i/10}
            \tikzstyle{fade}=[opacity=\laopacity]
            \node[smallvertex] (Nrt\n) at (-\i,0) {};
            \nextnode{N0\n}{Nrt\n}{-45:1}{revarc}
            \nextnode{N00\n}{N0\n}{-135:1}{revarc}
            \nextnode{N01\n}{N0\n}{-45:1}{revarc}
            \nextnode[leaf, fade]{la\n}{N00\n}{-100:1.32}{revarc, fade}
            \nextnode[reti]{N001\n}{N00\n}{-45:1}{revarc}
            \nextnode[leaf]{lb\n}{N001\n}{-90:.6}{revarc}
            
            \nextnode{Nx\n}{Nrt\n}{45:1.1}{arc}
            \nextnode{Nr\n}{Nx\n}{45:.9}{arc}
            \nextnode{Nq\n}{Nx\n}{-47:1}{revarc}
            \nextnode{Ny\n}{Nq\n}{-47:1}{revarc}
            \nextnode[reti]{Nz0\n}{Ny\n}{1,-.6}{revarc}
            \nextnode{Nz1\n}{Ny\n}{.2,-.6}{revarc}
            \node[smallvertex] (Nz\n) at ($(Nrt\n)+(-135:1.53)$) {} edge[arc, fade] ($(Nz\n)+(-90:.6)$) edge[arc, fade] ($(Nz\n)+(-120:.6)$);
            
            \draw[arc, tofut] (Nr\n) to ($(Nr\n)+(-47:1)$);
            \draw[arc, tofut, fade] (Nq\n) to[bend left=10] ($(Nq\n)+(-15:1.5)$);
            \foreach \i in {30,-30} \draw[arc, tofut] (Nz1\n) -- ++(-90+\i:.5);
            \draw[arc, tofut] (Nz0\n) -- ++(-90:.5);
            \draw[revarc, tofut, fade] (Nz0\n) -- ++(50:.5);
        
            \draw[arc, tofut] (N01\n) -- ++(-60:.7);
            \draw[arc, tofut] (Nrt\n) -- ++(-135:1.5);
            \draw[revarc, tofut] (Nrt\n) -- ++(45:1);
            \draw[arc] (N01\n) -- (N001\n);
        
            \nextnode{T00\n}{la\n}{-45:1}{arc, fade}
            \nextnode{T0\n}{T00\n}{-60:1.5}{arc}
            \node[smallvertex] (Tr\n) at ($(T0\n)+(-150:1)$) {};
            \node[smallvertex] (Tx\n) at ($(T0\n)+(30:.9)$) {};
            \nextnode{Tz\n}{Tx\n}{.6,.6}{revarc}

            \draw[arc] (T00\n) -- (lb\n);
            \draw[arc, tofut] (T0\n) -- ++(30:.8);
            \draw[revarc, tofut] (T0\n) -- ++(-150:1);
            \draw[arc, tofut] (Tx\n) -- ++(100:.6);
            \foreach \i in {30,-30} \draw[arc, tofut] (Tz\n) -- ++(90+\i:.5);
            \draw[arc, tofut] (Tr\n) --++(130:.6);
            
            \foreach [count=\i from 0] \x in {3.2,4.2}
                \node[smallleaf] (l\i\n) at ($(la\n)+(\x,0)$) {}
                    edge[revarc] ($(l\i\n)+(100:.6)$) edge[revarc] ($(l\i\n)+(-70:.6)$);
            
            \coordinate (Nroot\n) at ($(Nrt\n)+(0,1)+(45:2)$);
            \coordinate (NTleft\n) at ($(Nroot\n)+(-130:6.7)$);
            \coordinate (NTright\n) at ($(Nroot\n)+(-50:6.7)$);
            \coordinate (Npast\n) at ($(Nrt\n)+(-140:1.1)$);
            \coordinate (NfutUp\n) at ($(Nrt\n)+(45:.7)$);
            \coordinate (NfutDn\n) at ($(N01\n)+(-45:.7)$);
            
            \coordinate (Tpast\n) at ($(T0\n)+(-135:1)$);
            \coordinate (Troot\n) at ($(Tpast\n)+(-90:.2)$);
            \coordinate (Tfut\n) at ($(T0\n)+(45:.8)$);
            
            \node[fill=pastcol, rounded corners] (past\n) at ($(NTleft\n)+(80:2)$) {\past};
            \node[fill=pastcol, rounded corners] (past2\n) at ($(Troot\n)+(10:3)$) {\past};
            \node[fill=pastcol, rounded corners] (future\n) at ($(NTright\n)+(115:3.5)$) {\future};

            \draw[line width=1.5pt,iota] ($(NTleft\n)+(1.2,.5)$) to[bend left=10] (past\n);
            \draw[line width=1.5pt,iota] ($(NTright\n)+(-2,-.7)$) to[bend left=10] (past2\n);
            \draw[line width=1.5pt,iota] ($(NTright\n)+(-2,1)$) to[bend right=10] (future\n);
            
            \begin{pgfonlayer}{background}
                \draw[treearea] (Nroot\n) -- (NTleft\n) -- (NTright\n) -- cycle;
                \draw[treearea] (Troot\n) -- (NTleft\n) -- (NTright\n) -- cycle;
            \end{pgfonlayer}
        }

        \begin{pgfonlayer}{background}
            \draw[pastarea] (Npast) -- ($(NTleft)+(.3,0)$) -- ($(NTleft)+(2,0)$) -- cycle;
            \draw[pastarea] (Tpast) -- ($(NTleft)+(.3,0)$) -- ($(NTleft)+(2,0)$) -- cycle;
            \draw[pastarea] ($(NfutUp)+(135:.2)$) -- ++(50:1.66) -- ($(NTright)+(-.8,.4)$) -- ($(NTright)+(-4.5,.4)$) -- ++(60:1.3) -- cycle;
            \draw[pastarea] ($(Tfut)+(-90:.6)$) -- ($(NTright)+(-.9,-.4)$) -- ($(NTright)+(-4.5,-.4)$) -- cycle;
            
            \draw[pastarea] ($(Npastb)+(90:.3)$) -- ($(NTleftb)+(.3,0)$) -- ($(NTleftb)+(3,0)$) -- ($(N00b)+(-45:.3)$) -- ($(N0b)+(0:.3)$) -- ($(N0b)+(90:.3)$) -- cycle;
            \draw[pastarea] ($(NfutUpb)+(135:.2)$) -- ++(50:1.66) -- ($(NTrightb)+(-.8,.4)$) -- ($(NTrightb)+(-4.5,.4)$) -- ++(60:1.3) -- cycle;
            \draw[pastarea] (Tpastb) -- ($(NTleftb)+(.3,0)$) -- ($(NTleftb)+(3,0)$) -- ($(T00b)+(45:.4)$) -- ($(NTrightb)+(-4.5,-.4)$) -- ($(NTrightb)+(-.9,-.4)$) -- cycle;
            
            \draw[pastarea, draw=black, dashed] 
                ($(N0)+(180:.35)$) -- ($(N00)+(150:.3)$) -- ($(la)+(180:.3)$) -- ($(T00)+(180:.3)$) -- ($(T0)+(-105:.35)$) -- ($(T0)+(-15:.35)$) -- ($(T00)+(15:.3)$) -- ($(la)+(0:.3)$) -- ($(N00)+(0:.3)$) -- ($(N0)+(0:.35)$) -- ($(N0)+(90:.35)$) --cycle;
            \foreach \n in {,b}{
                \draw[emb] (Ny\n.center) -- (Nq\n.center) -- (Nx\n.center) -- (Nrt\n.center) -- (N0\n.center) -- (N00\n.center) -- (N001\n.center) -- (lb\n.center);
                \draw[emb] (Nz0\n.center) -- (Ny\n.center) -- (Nz1\n.center);
                \draw[emb] (Nr\n.center) -- (Nx\n.center);
                \draw[emb] (Nr\n.center) -- ++(-47:1);
                \draw[emb] (Nz0\n.center) -- ++(-90:.5);
                \foreach \i in {30,-30} \draw[emb] (Nz1\n.center) -- ++(-90+\i:.5);
            
                \draw[emb] (Tz\n.center) -- (Tx\n.center) -- (T0\n.center) -- (T00\n.center) -- (lb\n.center);
                \draw[emb] (Tz\n.center) -- (Tx\n.center) -- ++(100:.6);
                \draw[emb] (T0\n.center) -- (Tr\n.center) --++(130:.6);

                \foreach \i in {30,-30} \draw[emb] (Tz\n.center) -- ++(90+\i:.5);
                \foreach \i in {0,1}
                    \draw[emb] ($(l\i\n.center)+(-70:.6)$) -- (l\i\n.center) -- ($(l\i\n.center)+(100:.6)$);
            }
            \draw[emb] (N00.center) -- (la.center) -- (T00.center);
        \end{pgfonlayer}
        
        \node at ($(Nroot)+(-170:.7)$) {$N$};
        \node at ($(Troot)+(170:.7)$) {$T$};
    \end{tikzpicture}
    }
    \caption{Illustration of \cref{def:restrict} (except non-$\Y$ parts of $\iota$ and $\iota'$).
        \textbf{Left}: The example containment structure~$\psi$ of \cref{fig:containment struct}. The dashed area indicates $S'$ and $g$: for all $u$ in the dashed area, $g(\iota(u))=\past$ while $\iota(u)\notin\{\past,\future\}$.
        \textbf{Right}: The $g$-restriction of $\psi$.
        In accordance with \cref{def:restrict}, note how
        (a) arcs in the ``\past''-area of $N$ that are not mapped to by $\emb$ disappear in $N'$
        while arcs that are mapped by $\emb$ persist,
        (b) many of the arcs in $T$ are in the ``\past''-area, but embedded by $\emb$ into paths of $N$ that live (at least partially) in the ``\future''-area of $N$ and, therefore, also persist, and
        (c) a common leaf of $T$ and $N$ has been removed since all its incoming arcs have been deleted.
        Finally, as will be the case in the dynamic programming later on, the ``\past''- and ``\future''-areas never touch in $D(N',T')$.}
    \label{fig:restriction}
\end{figure}

\begin{definition}[restriction, see \cref{fig:restriction}]\label{def:restrict}
Let~$\chi = (D(N,T), \emb, \iota)$ be an $(S, \Y)$-containment structure,
let~$S' \subseteq S$, and
let~$g:S \cup \Y \to S' \cup \Y$
be some function such that for all $v \in S \cup \Y$, $g(v)=v$ if $v \in S'$ and $g(v) \in \Y$ otherwise.
Then, we call $g$ a \emph{restriction function}.
Further, the \emph{$g$-restriction} of $\chi$ is the tuple $(D(N',T'), \emb', \iota')$ constructed as follows:

With $\iota_g(u)$ abbreviating $g(\iota(u))$ for all $u \in V(D(N,T))$,
let $D(N',T')$ be the display graph derived from $D(N,T)$ 
by replacing the isolabelling $\iota$ with $\iota_g$,
and then exhaustively deleting redundant arcs and vertices (with respect to $(D(N,T), \emb, \iota_g)$) from $D(N,T)$.
Finally, we define $\iota'$ and $\emb'$ as the respective restrictions of $\iota_g$ to $D(N',T')$ and $\emb$ to $T'$.
\end{definition}

We will sometimes write $(S_1\to y_1, \dots, S_j \to y_j)$-restriction instead of $g$-restriction, where $y_1, \dots y_j$ are labels in $\Y$ and 
$g(u) := y_i$ for all $u \in S_i$, and $g(u):=u$ for all other $u$.
For example, we may write $(P\to \past)$-restriction when $g$ is the function that maps all vertices in $P$ to the label $\past$, and leaves other vertices unchanged.
If any $S_i$ is a singleton set~$\{z\}$, we permit ourselves to write $z\to y_i$ instead of $\{z\}\to y_i$.
In particular, we will see $(\pright\to \future)$-restrictions,
where the label $\pright$ is mapped to the label $\future$. 
Such restrictions may sometimes be used to ``merge'' labels.
Note that Steps~1 and 2 of the process described in \cref{sec:signatureOverview} is (roughly analogous to) the process of constructing the $(P\rightarrow \past, F\rightarrow \future)$-restriction of some $V(\NinTin, \emptyset)$-containment structure corresponding to a solution.

Before proving results relating to $g$-restrictions, we make the following observations about $g$-restrictions and redundant arcs and vertices, which are straightforward consequences of \cref{def:redundant} and, thus, we omit their proofs.

\begin{observation}\label{obs:redundantQ}
Any arc or vertex~$a$ of $D(N,T)$ is $y$-redundant with respect to $(D(N,T), \emb, \iota)$ if and only if $\iota(Q_a) = \{y\}$ for a set of vertices $Q_a$ that depends only on $D(N,T)$ and~$\emb$:
\begin{compactitem}
  \item If $a$ is a tree arc, then $Q_a = V(a)\cup V(\emb(a))$.
  \item If $a$ is a network arc, then $Q_a = Q_{a'}$ if $Q_a$ is part of a path $\emb(a')$, otherwise $Q_a = V(a)$.
  \item If $a$ is a tree vertex, then $Q_a = a \cup \emb(a) \cup \bigcup_{uv \in A(D(N,T)): a \in \{u,v\}}Q_{uv} $.
  \item If $a$ is a network vertex and $a = \emb(u')$, then $Q_a = Q_{u'} \cup \bigcup_{uv \in A(D(N,T)): a \in \{u,v\}}Q_{uv}$, otherwise  $Q_a = a \cup \bigcup_{uv \in A(D(N,T)): a \in \{u,v\}}Q_{uv}$.
\end{compactitem}
\end{observation}

\begin{observation}\label{obs:redundant tree}
If a tree arc $uv$ is $y$-redundant then so is every arc in $\emb(uv)$. 
A tree vertex $u$ is $y$-redundant if and only if $\emb(u)$ is $y$-redundant.
\end{observation}

\begin{lemma}\label{cor:restrictionIsAStructure}
Let $\chi$ be an $(S, \Y)$-containment structure,
let $S'\subseteq S$, and
let $\chi'$ be the $g$-restriction of $\chi$ for some restriction function $g:S\cup\Y \to S'\cup\Y$.
Then, $\chi'$ is an $(S',\Y)$-containment structure.
\end{lemma}
\begin{proof}
  To show that $\chi=(D(N,T),\emb,\iota)$ being an $(S,\Y)$-containment structure implies
  $\chi'=(D(N',T'),\emb',\iota')$ being an $(S',\Y)$-containment structure,
  we verify the five conditions of \cref{def:containment struct} individually.
  
  \begin{compactdesc}
    \item[$D(N',T')$ is a display graph.]
        As $D(N',T')$ was derived from $D(N,T)$ by deleting a subset of arcs and vertices,
        it is clear that $D(N',T')$ remains a display graph.
    \item[$\emb'$ is an embedding function.]
        As $\emb'$ is the restriction of $\emb$ to $T'$, and no arc in a path $\emb(uv)$ was deleted unless the arc $uv$ was deleted as well
        (and similarly for vertices $\emb(u)$),
        $\emb'$ is still a function that maps vertices of $T'$ to vertices of $N'$ and arcs of $T'$ to directed paths of $N'$.
        The other properties of an embedding function (such as all paths being arc-disjoint) follow immediately
        from the fact that these properties hold for $\emb$.
        Thus, $\emb'$ is an embedding function on $D(N',T')$.
    \item[$\iota'$ is an $(S',\Y)$-isolabelling.]
        Note that the properties of an isolabelling follow immediately from construction of $\iota'$,
        the fact that $g(v)=v$ for all $v\in S'$ and the fact that $\iota$ is an $(S,\Y)$-isolabelling.
        To see that the image of $\iota'$ in $\NinTin$ contains $S'$,
        we prove that, for every $z \in S'$, there is some $u \in V(D(N',T'))$ for which $\iota'(u) = z$.
        To see this, consider the vertex $u \in V(D(N,T))$ for which $\iota(u) = z$.
        Since $z \in S'$ we then have $g(\iota(u)) = g(z) = z \notin \Y$ and, thus, $u$ cannot become redundant in the construction of $\chi'$.
        Hence, $u$ is also a vertex of $D(N',T')$, and $\iota'(u) = g(\iota(u)) = z$ as required. 
    \item[Same degrees in $D(N',T')$ as in $\NinTin$.]
        Let $u\in V(D(N,T))$ with $\iota(u) \in S'$.
        Observe that no arc incident to $u$ is deleted when constructing $D(N',T')$ and,
        thus, $u$ has the same in-and out-degrees in $D(N',T')$ as it does in $D(N,T)$,
        that is, the same in- and out-degrees as $\iota(u) = \iota'(u)$.
    \item[Out-arcs in $D(N',T')$.]
        Let $u\in V(T')$ with $\iota'(u) \neq \iota'(\emb'(u))$.
        In particular, there is no $y \in \Y$ with $\iota'(u) = \iota'(\emb'(u)) = y$.
        Then, no arc incident to $u$ was deleted (as for any such arc $a$ the path $\emb(a)$ contains $\emb(u)$ by \cref{def:embed}\eqref{it:node&path}) and,
        so, $u$ has the same out-degree in $D(N',T')$ as in $D(N,T)$.
        Moreover, $u$ has out-degree~$2$ in $D(N,T)$ since $\iota(u) \neq \iota(\emb(u))$
        (otherwise $\iota'(u) = \iota'(\emb'(u))$ by construction of $\emb'$ and $\iota'$) and, thus, in $D(N',T')$.\qedhere
  \end{compactdesc}
\end{proof}

\begin{lemma}[transitivity of restrictions]\label{lem:transitivity-of-restrictions}
  Let $\chi$ be an $(S,\Y)$-containment structure for some $S$ and $\Y$,
  let $S''\subseteq S'\subseteq S$ and let
  $g':S\cup \Y \to S'\cup \Y$ and $g'':S'\cup \Y \to S'' \cup \Y$ be restriction functions.
  Let $\chi'$ be the $g'$-restriction of $\chi$ and
  let $\chi''$ be the $g''$-restriction of $\chi'$.
  Then, $\chi''$ is the $(g''\circ g')$-restriction of $\chi$.
\end{lemma}
\begin{proof}
  First, we show that $g:=g''\circ g'$ is a restriction function.
  To this end, let $v\in S\cup\Y$.
  If $v\in S''$, then $v\in S'\cap S$ and, thus, $g(v) = g''(g'(v)) = g''(v) = v$ since both $g'$ and $g''$ are restriction functions.
  Otherwise, $v\notin S''$ and either $v\in S'$, in which case $g'(v) = v \notin S''$, or $g'(v) \in \Y$.
  In either case $g'(v)\notin S''$ and so $g(v)=g''(g'(v)) \in \Y$.

  In the following,
  let $\chi'=(D(N',T'),\emb',\iota')$,
  let $\chi''=(D(N'',T''),\emb'',\iota'')$ and
  let $\iota_{g'}$ be as described in \cref{def:restrict},
  that is, $\iota_{g'}(u) = g'(\iota(u))$ for all $u \in V(D(N, T))$
  (and analogously for $\iota'_{g''}$ and $\iota_g$).
  Then, by \cref{def:restrict}, $\iota' = \iota_{g'}$ and $\iota'' = \iota'_{g''}$.
  Note that, for any $u\in V(D(N',T'))$,
  \begin{equation}
      \iota'_{g''}(u) = g''(\iota'(u)) = g''(\iota_{g'}(u)) = g''(g'(\iota(u))) = g(\iota(u)) = \iota_g(u).
      \label{eq:g and g''}
  \end{equation}

  \medskip
  Next, we show that $\chi''$ equals the $g$-restriction $\chi^* = (D(N^*,T^*), \emb^*, \iota^*)$ of $\chi$.
  To this end, we first prove that $D(N^*,T^*) = D(N'',T'')$.
  As both display graphs are subgraphs of $D(N,T)$, it is enough to show that any arc or vertex is deleted in the construction of $D(N^*,T^*)$ from $D(N,T)$ if and only if it is deleted in the construction of $D(N',T')$ from $D(N,T)$ or of $D(N'',T'')$ from $D(N',T')$.

  \begin{compactdesc}
  \item[Arcs:]
  Let $uv$ be an arc of $D(N,T)$.
  If $uv\in A(T)$, then let $Q_{uv} := \{u,v\} \cup V(\emb(uv))$
  (that is, $Q_{uv}$ contains $u$, $v$ and all vertices in the path $\emb(uv)$).
  If $uv$ is an arc of $N$, then let $Q_{uv}$ be the set containing $u$ and $v$ together with any $u'$ and $v'$ for which $uv$ is in the path $\emb(u'v')$ (by \cref{def:embed}, there is at most one such arc $u'v'$).
  By Observation~\ref{obs:redundantQ}, $uv$ is deleted in the construction of $D(N^*,T^*)$
  if and only if
  there is some $y\in Y$ with $\iota_g(Q_{uv}) = \{y\}$.

  First, assume that $uv$ is in $D(N'',T'')$ but not in $D(N^*,T^*)$. 
  Since $D(N'',T'')$ is a subgraph of $D(N',T')$, we know that $uv$ is also in $D(N',T')$.
  If $uv\in A(T)$, then all vertices of $\emb'(uv)=\emb(uv)$ are still in $D(N',T')$ as $\emb'$ is an embedding function on $D(N',T')$.
  If $uv\in A(N)$ and $Q_{uv} = \{u,v,u',v'\}$ and any of $u'$ and $v'$ is not in $D(N',T')$,
  then the arc $u'v'\in A(T)$ was deleted in the construction of $D(N',T')$ and,
  by Observation~\ref{obs:redundant tree}, so was $uv$, contradicting $uv$ being in $D(N',T')$.
  Thus, all vertices of $Q_{uv}$ are in $D(N',T')$ and, by \eqref{eq:g and g''}, $\iota'_{g''}(Q_{uv}) = \iota_g(Q_{uv})= \{y\}$ for some $y\in\Y$.
  Then $uv$ is deleted in the construction of $D(N'',T'')$.

  Second, assume that $uv$ is in $D(N^*,T^*)$ but not in $D(N'',T'')$.
  If $uv$ is in $D(N',T')$ then $uv$ is redundant with respect to $(D(N',T'),\emb',\iota')$,
  that is, $Q_{uv} \subseteq V(D(N',T'))$ and, by Observation~\ref{obs:redundantQ},
  there is some $y\in \Y$ with $\{y\} = \iota'_{g''}(Q_{uv}) = \iota_g(Q_{uv})$
  (using \eqref{eq:g and g''}).
  But then, $uv$ is also deleted in the construction of $D(N^*, T^*)$.
  If $uv$ is not in $D(N',T')$, then $\iota_{g'}(Q_{uv}) = \{y\}$ for some $y \in \Y$,
  implying $\iota_{g}(Q_{uv}) = g(\iota(Q_{uv})) = g''(g'(\iota(Q_{uv})) = g''(\iota_{g'}(Q_{uv})) = g''(\{y\}) = \{g''(y)\}$. As $y\in\Y$ implies $g''(y) \in \Y$, we know that $uv$ is deleted in the construction of $D(N^*, T^*)$.

  \item[Vertices:]
  Since each arc of $D(N,T)$ is in $D(N^*,T^*)$ if and only if it is in $D(N'',T'')$,
  all vertices with at least one incident arc in $D(N^*,T^*)$ are in $D(N'',T'')$ and vice versa.
  It remains to consider the isolated vertices. 
  To this end, let $v\in V(D(N,T))$ with no incident arcs in $D(N^*,T^*)$ and
  let $Q_v := \{v, \emb(v)\}$ if $v\in V(T)$ and
  let $Q_v$ be the set containing $v$ and any $u$ such that $v = \emb(u)$ if $v\in V(N)$.
  By definition of redundant vertices, either both vertices of $Q_v$ are deleted in the construction of $D(N^*,T^*)$  or neither is. As $D(N^*,T^*)$ and $D(N'',T'')$ have the same arcs, we may assume neither element of $Q_v$ has any incident arcs.
  Then, $v$ is deleted in the construction of $D(N^*,T^*)$ if and only if $\iota_g(Q_v) = \{y\}$ for some $y \in \Y$.
  Now, if $Q_v$ intersects $V(D(N',T')$, then $Q_v\subseteq V(D(N',T'))$.
  Then, by \eqref{eq:g and g''}, $\iota'_{g''}(Q_v) = \iota_g(Q_v)$ implying that
  $v$ is in $D(N'',T'')$ if and only if it is in $D(N^*,T^*)$.
  Otherwise, $Q_v\cap V(D(N',T')) = \emptyset$ (in particular, $v \notin V(D(N',T'))$),
  implying $\iota_{g'}(Q_v) = \{y\}$ for some $y \in \Y$.
  As in the Arc-case, $\iota_g(Q_v) = \{g''(y)\} = \{y'\}$ for some $y'\in \Y$.
  Thus, $v$ is neither in $D(N^*,T^*)$ nor in $D(N'',T'')$.
  \end{compactdesc}
  Since $\emb^*$ is the restriction of $\emb$ to $T^*$ and
  $\emb''$ is the restriction of $\emb$ to $T''$ (via $\emb'$), we also have $\emb^* = \emb''$.
  By \eqref{eq:g and g''}, $\iota'_{g''}(u) = \iota_g(u)$ for any vertex $u$ in $D(N',T')$ and,
  thus, $\iota''(u) = \iota^*(u)$ for any vertex $u$ in $D(N'',T'')$
  (as $\iota''$ and $\iota^*$ are just restrictions of $\iota'_{g''}$ and $\iota_g$ respectively).
  Therefore, $\iota^* = \iota''$ and we have proved that $\chi^* = \chi''$, as required.
\end{proof}


\subsection{Well-behaved containment structures}\label{sec:wellBehaved}

At this point, we note some additional properties that it will be helpful to assume for $(S, \Y)$-containment structures.
These properties are summarized in the concept of ``well-behavedness'' and,
in what follows, we will restrict our attention to such containment structures.
While not directly implied by the definition of a containment structure,
the properties effectively ensure that containment structures behave ``as expected''. In particular, our dynamic programming algorithm will work with well-behaved signatures.

\begin{definition}\label{def:wellBehaved}
An $(S, \Y)$-containment structure $(D(N,T),\emb, \iota)$ is called \emph{well-behaved} if
\begin{compactenum}[(a)]
    \item\label{it:WBnoRedundancy} $D(N,T)$ contains no redundant arcs or vertices;
    \item\label{it:WBNoLabelCrossings} For all arcs $uv$ in $D(N,T)$ with $\iota(u),\iota(v)\in\Y$, we have $\iota(u)=\iota(v)$;
    \item\label{it:WBSPaths} For all $u,v\in V(D(N,T))$ such that
    $\iota(u), \iota(v)\in S$ and
    $D(N,T)$ has a $u$-$v$-path,
    $\NinTin$ has an $\iota(u)$-$\iota(v)$-path.
\end{compactenum}
\end{definition}

Here we prove a number of properties of well-behaved containment structures, that allow us to focus on them going forward.
\begin{lemma}\label{lem:top-level-WB}
  Let $\chi:=(D(N,T), \emb, \iota)$ be a $(V(\NinTin), \Y)$-containment structure with $\iota^{-1}(\Y) = \emptyset$.
  Then, $\chi$ is well-behaved.
\end{lemma}
\begin{proof}
  Since $\iota^{-1}(\Y)=\emptyset$ and $S=V(\NinTin)$, we know that $\iota$ is an isomorphism between $D(N,T)$ and $V(\NinTin)$, implying \cref{def:wellBehaved}\eqref{it:WBSPaths},
  while~\eqref{it:WBnoRedundancy} and~\eqref{it:WBNoLabelCrossings} are direct consequences of $\iota^{-1}(\Y)=\emptyset$.
\end{proof}

\begin{lemma}\label{lem:WBrestriction}
  Let $\chi = (D(N,T), \emb, \iota)$ be a well-behaved $(S, \Y)$-containment structure and
  let $g$ be a restriction function such that,
  for all arcs $uv$ of $D(N,T)$ with $g(\iota(u)),g(\iota(v))\in\Y$, we have $g(\iota(u))=g(\iota(v))$.
  Then, the $g$-restriction~$\sigma$ of $\chi$ is also well-behaved.
\end{lemma}
Essentially, this condition says that relabelling according to $g$ does not immediately create any arcs $uv$ for which $u$ and $v$ are labelled with different elements of $\Y$.

\begin{proof}
  We will show that $\sigma = (D(N',T'), \emb', \iota' = \iota\circ g)$ satisfies \cref{def:wellBehaved}\eqref{it:WBnoRedundancy}-\eqref{it:WBSPaths}.
  
  \begin{compactdesc}
    \item[\eqref{it:WBnoRedundancy}]
      follows by \cref{def:restrict},
      as redundant arcs and vertices are deleted when constructing~$\sigma$.
    \item[\eqref{it:WBNoLabelCrossings}:]
      Let $uv$ be an arc of $D(N',T')$ with $\iota'(u),\iota'(v)\in\Y$.
      Then, by the condition of the lemma, $\iota'(u)=\iota'(v)$.
    \item[\eqref{it:WBSPaths}:]
      Let $u,v \in V(D(N',T'))$ such that $\iota'(u), \iota'(v) \notin\Y$
      (that is, $\iota'(u),\iota'(v)\in S'$ for $S':=\operatorname{img}(g)\setminus\Y$) and
      there is a $u$-$v$-path~$p$ in $D(N',T')$.
      Then, by \cref{def:restrict}, $\iota(u) = g(\iota(u))=\iota'(u)$ and, similarly, $\iota(v) = \iota'(v)$.
      Moreover, as $D(N',T')$ is a subgraph of $D(N,T)$, the latter also contains~$p$.
      Finally, since $\chi$ is well-behaved, there is a path in $\NinTin$ from $\iota(u)=\iota'(u)$ to $\iota(v)=\iota'(v)$.\qedhere
  \end{compactdesc}
\end{proof}

\subsection{Partial solution and valid signatures}\label{sec:partialSolution}

As with any dynamic programming algorithm, we need some way to decide which signatures are ``correct''
before we have actually found a solution.
As such, we need a notion of a ``partial solution''.
Much as we may think of a signature for a bag~$(P,S,F)$
as corresponding to the $(P\to\past,F\to\future)$-restriction of some solution,
we may think of a partial solution as the $(F\to\future)$-restriction of some solution.
That is, a partial solution is a $(P\cup S, \{\future\})$-containment structure
that roughly corresponds to what would happen if we took a solution and ``forgot'' some of the details about the vertices in $F$.
A partial solution is then a ``witness'' for a given signature for $(P,S,F)$
if that signature can in turn be derived from the partial solution by ``forgetting'' details about the vertices in $P$.
In this case we call the signature ``valid''.
This is defined precisely below.

\begin{definition}[partial solution, signature, valid]\label{def:partial-solution-signature}
Let $(P,S,F)$ be a bag in the tree decomposition of $\NinTin$.
Then
any $(P\cup S, \{\future\})$-containment structure is called \emph{$F$-partial solution} (or simply \emph{partial solution}) for $(P,S,F)$.
A well-behaved signature $\sigma$ for $(P,S,F)$ is called \emph{valid} if
$\sigma$ is the $(P\to\past)$-restriction of a well-behaved partial solution $\psi$ for $(P,S,F)$.
We call $\psi$ a \emph{witness} for $\sigma$.
\end{definition}

\begin{lemma}\label{cor:WB-signatures-exist}
  Let $(\Nin,\Tin)$ be a \textsc{Yes}-instance of \textsc{Tree Containment} and
  let $(P,S,F)$ be a bag in the tree decomposition of $\NinTin$.
  Then, there is a well-behaved $F$-partial solution $\psi$, and a valid signature $\sigma$ for $(P,S,F)$.
\end{lemma}
\begin{proof}
  By \cref{lem:top-level-WB}, there is a well-behaved $(V(\NinTin), \Y)$-containment structure $\psi^*$
  with ${\iota}^{-1}(\Y) = \emptyset$.
  Clearly, $\psi^*$ is also a $(V(\NinTin), \emptyset)$-containment structure.
  Now, let $\psi$ be the $(F \to \future)$-restriction of $\psi^*$ and
  note that $\psi$ is an $F$-partial solution.
  Clearly, $x,y\in\{\future\}\Rightarrow x=y$ and, so, \cref{lem:WBrestriction} applies, showing that $\psi$ is well-behaved.
  Finally, let $\sigma$ be the $(P \to \past)$-restriction of $\psi$ which is valid since $\psi$ is well-behaved.
\end{proof}

\begin{lemma}\label{cor:top-level-sig}
  $(\Nin,\Tin)$ is a \textsc{Yes}-instance of \textsc{Tree Containment}
  if and only if
  there is a valid signature $\sigma := (D(N,T), \emb, \iota)$ for $(V(\NinTin), \emptyset, \emptyset)$
  with ${\iota}^{-1}(\future) = \emptyset$.
\end{lemma}
\begin{proof}
  By \cref{lem:TC equiv} and \cref{lem:top-level-WB},
  $(\Nin,\Tin)$ is a \textsc{Yes}-instance 
  if and only if
  there is a well-behaved $V(\NinTin), \Y)$-containment structure $\psi = (D(N,T), \emb, \iota)$ 
  with ${\iota}^{-1}(\{\Y\}) = \emptyset$, or equivalently
  a well-behaved $\emptyset$-partial solution  $\psi = (D(N,T), \emb, \iota)$ 
  with ${\iota}^{-1}(\{\future\}) = \emptyset$.
  First, suppose such a partial solution exists and
  let $\sigma := ((D(N',T'), \emb', \iota')$ be the $(V(\NinTin) \to \past)$-restriction of $\psi$.
  Then, $\sigma$ is a valid signature for $(V(\NinTin), \emptyset, \emptyset)$ and, by construction, ${\iota'}^{-1}(\future) = \emptyset$.
  (In fact $D(N',T')$ is the empty graph since all arcs and vertices of $D(N,T)$ become $\past$-redundant; but we do not use that fact here).

  For the converse, consider a valid signature $\sigma := (D(N,T), \emb, \iota)$ for $(V(\NinTin), \emptyset, \emptyset)$
  with ${\iota}^{-1}(\future) = \emptyset$.
  By \cref{def:partial-solution-signature}, 
  $\sigma$ is the $(V(\NinTin) \to \past)$-restriction of a well-behaved $\emptyset$-partial solution $\psi := ((D(N',T'), \emb', \iota')$.
  It remains to show that ${\iota'}^{-1}(\future) = \emptyset$. 
  Towards a contradiction, assume that $D(N',T')$ has a vertex~$u$ with $\iota'(u) = \future$.
  Since $u\notin\iota^{-1}(\future)=\emptyset$, $u$ is $\future$-redundant after applying $V(\NinTin) \to \past$.
  However, applying $V(\NinTin) \to \past$ cannot make a vertex $\future$-redundant that was not previously $\future$-redundant
  (as no new vertex gains the label~$\future$).
  Thus, $u$ is $\future$-redundant in $\psi$, contradicting \cref{def:wellBehaved}\eqref{it:WBnoRedundancy}.
\end{proof}

\noindent
\cref{lem:WBrestriction}, \cref{cor:WB-signatures-exist} and \cref{cor:top-level-sig} show that an instance 
  $(\Nin,\Tin)$ of \textsc{Tree Containment} is a \textsc{Yes}-instance
  if and only if
  there is a well-behaved valid signature $\sigma := (D(N,T), \emb, \iota)$ for the root bag $(V(\NinTin), \emptyset, \emptyset)$
  with ${\iota}^{-1}(\future) = \emptyset$.
Thus in order to solve an instance of \textsc{Tree Containment},
it is enough to decide for each bag~$(P,S,F)$ in the tree decomposition of $\NinTin$,
and for each well-behaved signature $\sigma$ for $(P,S,F)$, whether $\sigma$ is valid.

\subsection{Determining valid signatures}\label{sec:valid sigs}

With the formal definitions of valid signatures and restrictions in place,
we can now show how to determine whether a well-behaved signature for a bag~$(P,S,F)$ is valid,
assuming we know this for all signatures on the child bag(s).
As is common in dynamic programming techniques, we take advantage of the structure of a nice tree decomposition.
The following lemmas describe the exact conditions for Leaf, Forget and Introduce bags
while the additional terminology required for Join bags is deferred to \cref{sec:joinBags}.

\begin{lemma}\label{lem:leafBag}
  Let $(P,S,F)$ correspond to a Leaf bag in the tree decomposition i.e.\
  $P = S = \emptyset$, $F = V(\NinTin)$ and $(P,S,F)$ has no children.
  Let $\sigma := (D(N,T), \emb, \iota)$ be a well-behaved signature for $(P,S,F)$.
  Then, $\sigma$ is valid
  if and only if
  $\iota^{-1}(\past) = \emptyset$.
\end{lemma}
\begin{proof}
  Suppose first that $\sigma$ is valid.
  By \cref{def:partial-solution-signature},
  $\sigma$ is the $(P\to \past)$-restriction of a well-behaved $F$-partial solution (that is, $(P\cup S, \{\future\})$-containment structure)
  $\psi := D(N',T'), \emb', \iota')$.
  Then $\iota'(u) \neq \past$ for every vertex $u$ in $D(N',T')$ and, as $P = \emptyset$,
  this remains true after applying $P \to \past$.
  It follows that $\iota(u) \neq \past$ for every vertex $u \in D(N,T)$, as required.

  Conversely, suppose $\iota^{-1}(\past) = \emptyset$.
  Since $P = \emptyset$, the $(S, \{\past, \future\})$-containment structure $\sigma$ is also a $(P\cup S, \{\future\})$-containment structure,
  that is, $\sigma$ is an $F$-partial solution.
  To show that $\sigma$ is valid, we prove that $\sigma$ is the $(P\to \past)$-restriction of itself, that is, $\sigma$ is a witness for $\sigma$.
  To this end, observe that applying $P\to \past$ does not change the isolabelling and,
  by \cref{def:wellBehaved}, $\sigma$ contains no redundant arcs or vertices.
  Thus, applying $P\to \past$ and removing redundant arcs and vertices does not change~$\sigma$.
\end{proof}

\begin{lemma}\label{lem:forgetBag}
  Let $(P,S,F)$~correspond to a Forget bag in the tree decomposition with child bag~$(P',S',F)$,
  i.e. $P = P' \cup \{z\}$ and $S = S'\setminus \{z\}$ for some~$z\in S'$.
  Let $\sigma$ be a well-behaved signature for~$(P,S,F)$.
  Then, $\sigma$ is valid
  if and only if
  $\sigma$ is the $(z\to \past)$-restriction of a valid signature~$\sigma'$ for~$(P',S',F)$.
\end{lemma}
\begin{proof}
  Suppose first that $\sigma$ is the $(z\to \past)$-restriction of some valid signature $\sigma'$ for $(P',S',F)$.
  By \cref{def:partial-solution-signature}, $\sigma'$ has a witness $\psi'$
  (that is, $\psi'$ is a well-behaved $F$-partial solution for $(P',S',F)$ whose $P'\to\past$-restriction is $\sigma'$).
  By \cref{lem:transitivity-of-restrictions}, $\sigma$ is the $(P\to \past)$-restriction of $\psi'$, and so $\sigma$ is valid.

  For the converse, suppose that $\sigma$ is valid and
  let $\psi$ be a witness of $\sigma$
  (that is, $\psi$ is a well-behaved $F$-partial solution for $(P,S,F)$
  and $\sigma$ is the $(P\to \past)$-restriction of $\psi$).
  Then, the $(P'\to \past)$-restriction $\sigma'$ of $\psi$ is a valid signature for $(P',S',F)$ and,
  by \cref{lem:transitivity-of-restrictions}, the $(z\to \past)$-restriction of $\sigma'$ is the $(P\to \past)$-restriction of $\psi$, that is, $\sigma$.
\end{proof}

\begin{lemma}\label{lem:introduceBag}
  Let $(P,S,F)$~correspond to an Introduce bag in the tree decomposition with child bag~$(P,S',F')$,
  i.e.\ $S' = S\setminus \{z\}$ and $F' = F \cup \{z\}$ for some~$z\in S$.
  Let $\sigma$ be a well-behaved signature for~$(P,S,F)$.
  Then, $\sigma$ is valid 
  if and only if
  the $(z\to \future)$-restriction~$\sigma'$ of $\sigma$ is a valid signature for~$(P,S',F')$.
\end{lemma}
\begin{proof}
  Suppose first that $\sigma$ is a valid signature and
  let $\psi$ be a witness for $\sigma$
  (that is, $\psi$ is a well-behaved $F$-partial solution for which $\sigma$ is the $(P\to \past)$-restriction).
  Let $\psi'$ be the $(z \to \future)$-restriction of $\psi$ which,
  by \cref{lem:WBrestriction}, is a well-behaved $F'$-partial solution. 
  Let $\sigma^*$ denote the $(P\to \past)$-restriction of $\psi'$ and note that $\sigma^*$ is a valid signature for $(P,S',F')$.
  By \cref{lem:transitivity-of-restrictions},
  $\sigma^*$ is also the $(P\to \past, \{z\}\to \future)$-restriction of $\psi$ and
  the $(z\to\future)$-restriction $\sigma'$ of $\sigma$ is also the $(P\to \past, \{z\}\to \future)$-restriction of $\psi$.
  Thus, $\sigma' = \sigma^*$ and so $\sigma'$ is valid, as required (see \cref{fig:IntroduceBagEasyDir.}).
 
  \begin{figure}[t]
     \centering
\includegraphics[scale = 0.75]{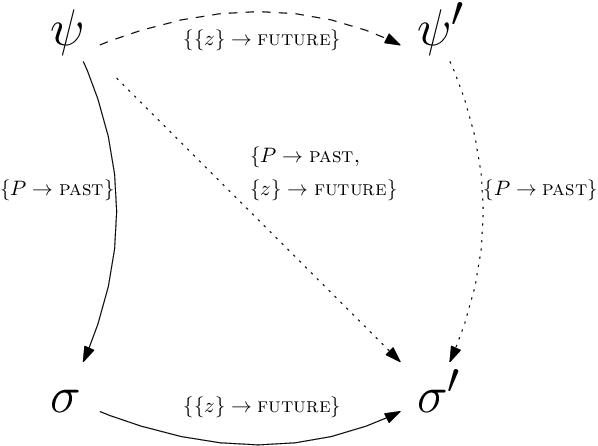}
     \caption{Illustration of proof for Introduce bags, validity of $\sigma$ implies validity of $\sigma'$. Solid lines are restriction relations that we may assume; the dashed line shows the construction of $\psi'$ from $\psi$; dotted lines are relations we can infer using transitivity.}
     \label{fig:IntroduceBagEasyDir.}
 \end{figure}
 
  For the converse,
  let $\sigma =: (D(N,T), \emb, \iota)$ and
  let $\sigma' =: (D(N',T'), \emb', \iota')$ be a valid signature for $(P,S',F')$.
  By \cref{def:partial-solution-signature},
  $\sigma'$ is the $(P\to \past)$-restriction of some well-behaved $F'$-partial solution $\psi' := (D(N_0',T_0'), \emb_0', \iota_0')$.
  In the following, $\sigma$, $\sigma'$ and $\psi'$ will guide us in constructing an $F$-partial solution $\psi$
  of which $\psi'$ is the $(z\to \future)$-restriction and, more importantly,
  of which $\sigma$ is the $(P\to \past)$-restriction.
  Then, this implies that $\sigma$ is valid.
 
  To begin the construction, let $V_z$ and $A_z$ denote the the set of vertices and arcs, respectively, that are in $D(N,T)$ but not in $D(N',T')$
  (that is, the vertices and arcs that are deleted when deriving $\sigma'$ from $\sigma$).
  Letting $g$ be the function $z \to \future$,
  we may assume that
  $g(\iota(v)) = \future$ for all vertices~$v$ in $V_z$ or $V(A_z)$ and
  $g(\iota(V(\emb(uv)))) = \{\future\}$ for any tree arc~$uv$ in $A_z$ as well.
  Similarly, let $V_p$ and $A_p$ be the vertices and arcs, respectively, that are deleted from $D(N_0', T_0')$ in the construction of $D(N',T')$. 
  Letting $h$ be the function $P\to \past$, we may assume that 
  $h(\iota(v)) = \past$ for all vertices~$v$ in $V_p$ or $V(A_p)$, and
  $h(\iota(V(\emb(uv)))) = \{\past\}$ for any tree arc~$uv$ in $A_p$ as well.
  
 We now construct the $F$-partial solution $\psi := (D(N_0,T_0), \emb_0, \iota_0)$ as follows.
 First, let $D(N_0, T_0)$ be the graph derived from $D(N_0',T_0')$ by adding all vertices and arcs of $V_z$ and $A_z$. 
 Equivalently, we may say we construct~$D(N_0,T_0)$ by adding the arcs and vertices of $V_p$ and $A_p$ to $D(N,T)$,
 or by adding the arcs and vertices of $V_z,V_p,A_z,A_p$ to $D(N',T')$.
 Let $\emb_0$ be defined as the 'union' of $\emb$ and $\emb_0'$ - that is, $\emb_0(uv) = \emb(uv)$ if $uv$ is an arc in $T$, and $\emb_0(uv) = \emb_0'(uv)$ if $uv$ is an arc in $T_0'$ (if $uv$ is an arc of both $T$ and $T_0'$, then these are the same, as $uv$ is in $T'$ and $\emb(uv) = \emb'(uv) = \emb_0'(uv)$).
 Similarly $\emb_0(u) = \emb(u)$ if $u$ is a vertex in $T$, and $\emb_0(u) = \emb_0'(u)$ if $u$ is a vertex in $T_0'$.
 Finally let $\iota_0:V(D(N_0,T_0))\to P\cup S \cup \{\future\}$ be defined as follows: $\iota_0(v) = \iota(v)$ if $v\in V_z$ or $\iota'(v) = \future$, $\iota_0(v) = \iota_0'(v)$ if $v \in V_p$ or $\iota'(v) = \past$, and $\iota_0(v) = \iota_0'(v) = \iota'(v) = \iota(v)$ otherwise. (Note that $\iota_0(v) = \future$ is possible for some vertices, if $\iota(v) = \future$, but $\iota_0(v) = \past$ is not possible as $\iota_0'(v) \neq \past$ for any $v$.)

 \begin{figure}[t]
     \centering
 \includegraphics[scale = 0.75]{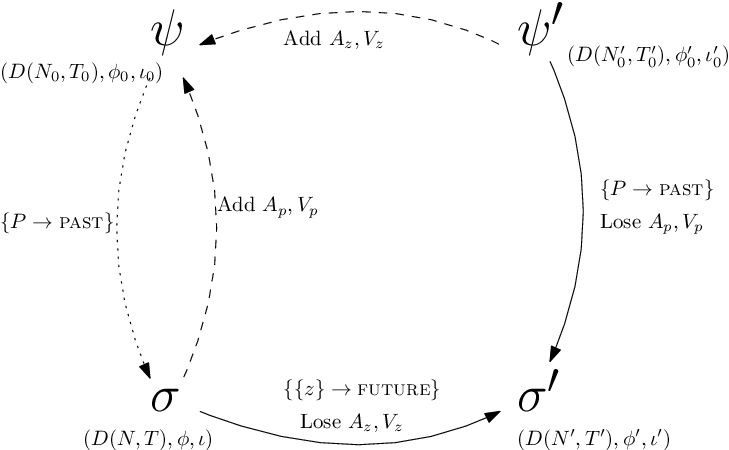}
     \caption{Illustration of proof for Introduce bags, validity of $\sigma'$ implies validity of $\sigma$. Solid lines are restriction relations that we may assume; the dashed lines show the construction of $\psi$ from $\psi'$ or $\sigma$ (we do not describe the construction of $\emb_0$ or $\iota_0$ in the figure, only the construction of $D(N_0,T_0)$ by adding arcs and vertices); the dotted line shows the relation we want to prove, that $\sigma$ is the $\{P\to \past\}$-restriction of $\psi$. }
     \label{fig:IntroduceBagConstructiveDir.}
 \end{figure}

 We note here that for any vertex $v$ in $D(N_0, T_0)$ with $\iota_0(v) \in S \cup \{\future\}$, $v$ is a vertex of $D(N,T)$ and $\iota_0(v) = \iota(v)$. Indeed, $v$ cannot be in $V_p$, nor can it hold that $\iota'(v) = \past$ if $v$ is in $D(N',T')$, as both cases would imply $\iota_0(v)=\iota_0'(v) \in P$. By construction, in all other cases $\iota_0(v) = \iota(v)$.
 Furthermore, for any arc $a$ incident to $v$ in $D(N_0,T_0)$, $a$ is an arc in $D(N,T)$ (since $a \in A_p$ would imply $\iota_0'(v)\in P$ and hence either $v \in V_p$ or $\iota'(v) = \past$), and thus $\emb_0(a)=\emb(a)$.
 A similar argument shows that for any vertex $v$ in $D(N_0,T_0)$ with $\iota(u) \in P$, $v$ is a vertex of $D(N_0',T_0')$, $\iota_0(v) = \iota_0'(v)$, and $\emb_0(a) = \emb_0'(a)$ for any arc $a$ in $D(N_0,T_0)$ incident to $v$. 
 
We will now show that $\psi$ is indeed a well-behaved $F$-partial solution; afterwards we will show that $\sigma$ is the $\{P\to \past\}$-restriction of $\psi$, from which it follows that $\sigma$ is valid. (A similar argument can be used to show that $\psi'$ is the $\{\{z\} \to \future\}$-restriction of $\psi$, but we will not need this fact so we do not prove it here.)

\begin{claim}
  $\psi$ is a well-behaved $F$-partial solution.
\end{claim}
\begin{claimproof}
We first show that $\psi$ satisfies one of the properties of a well-behaved $(P\cup S, \{\future\})$-containment structure

\begin{itemize}
    \item {\bf for any $u,v \in V(D(N_0, T_0))$ with $\iota_0(u), \iota_0(v) \in P \cup S$, if $D(N_0, T_0)$ has a path from $u$ to $v$ then there is a path from $\iota_0(u)$ to $\iota_0(v)$ in $\NinTin$:}
    Suppose for a contradiction that this is not the case, and let $u, v \in V(D(N_0,T_0))$ be such that $\iota_0(u), \iota_0(v) \in P \cup S$ and there is a path from $u$ to $v$ in $D(N_0,T_0))$, but no path from $\iota_0(u)$ to $\iota_0(v)$ in $\NinTin$, choosing $uv$ such that the length of the $u-v$ path is minimal. If $uv$ is an arc in $D(N_0,T_0)$, then $uv$ is an arc in $D(N,T)$ or $D(N_0', T_0')$, and the existence of a path from $\iota(u)$ to $\iota(v)$ in $\NinTin$ follows from the fact that $\sigma$ or $\sigma_0'$ is well-behaved.
    Otherwise, by choice of $u,v$ we may assume all internal vertices $u'$ on the path from $u$ to $v$ satisfy $\iota_0(u') = \future$. Then every arc on this path is an arc in $D(N,T)$ and so again, we have a path from $\iota_0(u) = \iota(u)$ to $\iota_0(v) = \iota(v)$ in $\NinTin$ by the fact that $\sigma$ is well-behaved.
     
    \item {\bf $D(N_0,T_0)$ is a display graph:} 
        \begin{itemize}
            \item {\bf $D(N_0, T_0)$ is acyclic: }  Suppose for a contradiction that $D(N_0,T_0)$ contains a directed cycle. If $\iota_0(u) \in P\cup S$ for some vertex $u$ in this cycle, then as shown above the existence of a path from $u$ to $u$ in $D(N_0,T_0)$ implies the existence of a path from $\iota_0(u)$ to $\iota_0(u)$ in $\NinTin$, contradicting the fact that $\NinTin$ is acyclic.
            
            It remains to consider the case that every vertex $z$ in this cycle has $\iota_0(z) = \future$. But in this case every arc of the cycle is also an arc in $D(N,T)$, and so $D(N,T)$ contains a cycle, a contradiction as $D(N,T)$ is a display graph.
            
            \item {\bf $T_0$ is an out-forest:} 
            We first observe that for every vertex $u$ in $D(N_0,T_0)$, either all its incident arcs in $D(N_0,T_0)$ are also arcs of $D(N,T)$, or they are all arcs of $D(N_0',T_0')$. Indeed, suppose for a contradiction that this is not the case, then there exist arcs $a \in A_p$, $a'\in A_z$ that share a vertex $u$. But by construction, $a \in A_p$ implies that $u\in V_p$ or $\iota'(u) = \past$ and so $\iota_0(u) \in P$, while $a' \in A_z$ implies $\iota_0(u) \in S \cup \{\future\}$, a contradiction.
            
            Now it remains to observe that every vertex in $T_0$ has in-degree at most $1$, as either all its incident arcs are in $T$ or they are all in $T_0'$. This, together with the fact that $D(N_0,T_0)$ is acyclic, implies that $T_0$ is an out-forest.
            
            \item {\bf Every vertex in $D(N_0, T_0)$ has in- and out-degree at most $2$ and total degree at most $3$:}
           As argued above, every vertex in $D(N_0,T_0)$ has all its incident arcs in $D(N,T)$ or in $D(N_0', T_0')$, and so this property follows from the fact that it holds for $D(N,T)$ and $D(N_0', T_0')$
           
            \item {\bf Any vertex in $V(N_0)\cap V(T_0)$ has out-degree $0$ and in-degree at most $1$ in each of $T_0$ and $N_0$:}
            Again, this follows from the fact that all incident arcs of a vertex in $D(N_0,T_0)$ belong to one of $D(N,T)$, $D(N_0', T_0')$.
        \end{itemize}
    
    \item {\bf $\emb_0$ is an embedding function on $D(N_0,T_0)$:}
    \begin{itemize}
        \item {\bf For each $u\in V(T_0)$, $\emb_0(u)\in V(N_0)$ and,
        for each arc $uv \in A(T_0)$, $\emb_0(uv)$ is a directed $\emb_0(u)$-$\emb_0(v)$-path in $N_0$:} This follows immediately from the construction of $\emb_0$ and the fact that $\emb$ and $\emb_0'$ are embedding functions.
        
        \item {\bf for any distinct $u,v \in V(T_0)$, $\emb_0(u)\neq \emb_0(v)$}: 
        Note that if a tree vertex $u$ is in $V_p$ then, by definition of $\past$-redundant, so is $\emb_0'(u)$. Similarly if $u \in V_z$ then $\emb(u) \in V_z$. So now for two vertices $u,v \in V(T_0)$, if $u \in V_p$ and $v \in V_z$ then $\emb_0(u) = \emb_0'(u) \in V_p$ and $\emb_0(v) = \emb(v) \in V_z$, and so $\emb_0(u) \neq \emb_0(v)$.
        Otherwise, $u,v$ are either both in $D(N,T)$ or both in $D(N_0', T_0')$, and $\emb_0(u) \neq \emb_0(v)$ follows from the fact that $\emb$ and $\emb_0'$ are embedding functions.
    
        \item {\bf for any $u \in V(T_0) \cap V(N_0)$, $\emb_0(u) = u$:}
        Follows immediately from the construction of $\emb_0$.
        \item {\bf the paths $\{\emb_0(uv) \mid uv\in A(T_0)\}$ are arc-disjoint:} 
        Observe that by construction that if a tree arc $uv$ is in $A_p$, then so is every arc in $\emb_0(uv) = \emb_0'(uv)$. Similarly if $uv$ is in $A_z$ then so is every arc in $\emb_0(uv) = \emb(uv)$.
        So consider two distinct tree arcs $uv, u'v' \in A(T_0)$.
        If $uv \in A_p, u'v'\in A_z$ then the arcs of the paths $\emb_0(uv), \emb_0(u'v')$ are in $A_p, A_z$ respectively, so $\emb_0(uv), \emb_0(u'v')$ are arc-disjoint. Otherwise, we may assume both $uv$ and $u'v'$ are arcs in one of $D(N,T), D(N_0', T_0')$, from which the claim follows by the fact that $\emb$ and $\emb_0'$ are embedding functions.
        
        \item {\bf for any distinct $p,q\in A(T_0)$, $\emb_0(p)$ and $\emb_0(q)$ share a vertex $z'$
        only if $p$ and $q$ share a vertex $w$ with $z' = \emb_0(w)$:} 
        Suppose $\emb_0(p)$ and $\emb_0(q)$ share a vertex $z'$. As argued previously, all incident arcs to $z'$ must be in one of $D(N,T), D(N_0',T_0')$. Thus in particular, $z'$ cannot have incident arcs from both $A_p$ and $A_z$. Then since $\emb_0(p)$ and $\emb_0(q)$ both contain arcs incident to $z'$, we must have that $p,q \notin A_p$ or $p,q \notin A_z$. Then either $\emb_0(p) = \emb(p)$ and $\emb_0(q)=\emb(q)$, or $\emb_0(p) = \emb_0'(p)$ and $\emb_0'(q)$. Then the claim follows from the fact that $\emb$ and $\emb_0'$ are embedding functions.
    \end{itemize}

    \item {\bf $\iota_0$ is a $(P\cup S, \{\future\})$-isolabelling:}
    \begin{itemize}
        \item {\bf For $u \in V(D(N_0, T_0))$ with $\iota_0(u) \neq \future$, $\iota_0(u) \in V(\Nin)$ only if $u \in V(N_0)$ and $\iota_0(u) \in V(\Tin)$ only if $u \in V(T_0)$ :}  Follows immediately from the construction of $\iota_0$.
        \item {\bf For $u,v \in V(D(N_0, T_0))$ with $\iota_0(u), \iota_0(v) \neq \future$, $\iota_0(u) = \iota_0(v)$ only if $u = v$:} 
        Suppose $u \neq v$ and $\iota_0(u), \iota_0(v) \neq \future$; we will show $\iota_0(u) \neq \iota_0(v)$. If $u \in V_p$ and $v \in V_z$, then by construction $\iota_0(u) = \iota(u) \in \{z, \future\}$ (and so in fact $\iota_0(u)=z$) and $\iota_0(v) = \iota_0'(v) \in P$. Thus $\iota_0(u) \neq \iota_0(v)$. Otherwise, we may assume $u$ and $v$ are both vertices in $D(N,T)$ or in $D(N_0', T_0')$. Then $\iota_0(u) \neq \iota_0(v)$ follows from the fact that $\iota$ and $\iota_0'$ are isolabelings.

        \item {\bf For $u,v \in V(D(N_0, T_0))$ with $\iota_0(u), \iota_0(v) \neq \future$, the arc $uv$ is in $D(N_0,T_0)$ if and only if $\iota_0(u)\iota_0(v)$ is in $\NinTin$ :} Follows immediately from the construction of $\iota_0$ and the fact that $\iota, \iota_0'$ are isolabelings.
        \item {\bf $\iota_0$ is surjective onto $P \cup S$:}
        Consider any $w \in P\cup S$. If $w \in P$, there is $u \in V(D(N_0', T_0'))$ with $\iota_0'(u) = w$. For such $u$ we have either  $u \in V_p$ or $\iota'(u) = \past$. In either case we have $\iota_0(u) = \iota_0'(u) = w$. Similarly if $w \in S$, there is $u \in V(D(N,T))$ with  $\iota(u) = w$ and either $\iota'(u) = \iota(u)$ (if $w \in S'$) or $u \in V_z$ or $\iota'(u) = \future$ (if $w = z$), and so $\iota_0(u) = \iota(u) = w$.
    \end{itemize}
    
    \item {\bf Each vertex $u$ with $\iota_0(u) \in P\cup S$ has the same in- and out-degree in $D(N_0,T_0)$ as $\iota_0(u)$ in $\NinTin$:}
    As previously shown, all incident arcs of $u$ in $D(N_0,T_0)$ belong to at least one of $D(N,T), D(N_0',T_0')$. Moreover if $\iota_0(u) = \iota(u)$ then all incident arcs are in $D(N,T)$ and if $\iota_0(u) = \iota_0'(u)$ then all incident arcs are in $D(N_0', T_0')$.
    Then this property follows from the fact that $\sigma$ and $\sigma_0'$ are containment structures.

    \item {\bf Each vertex $u$ of $T_0$ with $\iota_0(u) \neq \iota(\emb_0(u))$ has $2$ out-arcs in $D(N_0,T_0)$ :}
    In the case that $\iota_0(u) \in P \cup S$, this follows from previously-shown properties. We may assume $u$ is an internal vertex of $T_0$ (as otherwise $u \in V(T_0)\cap V(N_0)$ and $\emb_0(u) = u$). Then $\iota_0(u)$ is also an internal vertex of $\Tin$, and $u$ has the same in-and out-degree as $\iota_0(u)$. Thus in particular $u$ has out-degree $2$, as $T_0$ is binary.
    
    For the case that $\iota_0(u) = \future$, by construction $\iota_0(u) = \iota(u)$. Furthermore $\iota(\emb(u))\neq \future$ (as this would imply $\iota_0(\emb_0(u)) = \iota(\emb(u) = \future = \iota_0(u)$, a contradiction). Then as $\sigma$ is a $(S,\{\past, \future\})$-containment structure, $u$ has out-degree $2$ in $D(N,T)$, and therefore in $D(N_0,T_0)$.
    
    \medskip
    The above conditions show that $\psi$ is a $(P\cup S, \{\future\})$-containment structure, i.e. an $F$-partial solution. Next we show that $\psi$ is well-behaved:
    
    \item {\bf $D(N_0, T_0)$ contains no redundant arcs or vertices:} Suppose for a contradiction that $D(N_0,T_0)$ contains a $\future$-redundant arc or vertex $a$. We will show that such an arc or vertex is also redundant w.r.t $\sigma$, a contradiction as $\sigma$ is well-behaved.
    Consider the case that $a$ is an arc $uv$ Then $\iota_0(u) = \iota_0(v) = \future$. It follows by construction that $u,v$ are vertices of $D(N,T)$ and $\iota(u) = \iota(v) = \future$. In addition we have that $\iota_0(u') = \future$ for any vertex $u'$ on the path $\emb_0(uv) = \emb(uv)$, and hence $\emb(u') = \future$ for such $u'$. It follows that $uv$ is $\future$-redundant w.r.t $\sigma$, the desired contradiction. Essentially the same arguments can also be made for redundant network arcs and redundant tree or network vertices.

    \item {\bf For any $y,y' \in \{\future\}$ with $y \neq y'$, there is no arc $uv$ in $D(N_0, T_0)$ for which $\iota_0(u) = y$ and $\iota_0(v) = y'$:} This follows immediately from the fact that  $|\{\future\}| = 1$ so there are no such $y, y'$.
    
    \item {\bf For any $u,v,$ in $V(D(N_0, T_0))$ with $\iota_p(v), \iota_0(v) \in P\cup S$, if $D(N_0, T_0)$ has a path from $u$ to $v$ then $\NinTin$ has a path from $\iota_0(u)$ to $\iota_0(v)$:} This has already been shown, at the start of the proof for this claim.\claimqedhere
\end{itemize}
\end{claimproof}

We now have that $\psi$ satisfies all the conditions of a well-behaved $F$-partial solution.
Finally we need to show that $\sigma$ is the $\{P \to \past\}$-restriction of $D(N_0, T_0)$. 
    
\begin{claim}
    $\sigma$ is the $\{P \to \past\}$-restriction of $D(N_0, T_0)$
\end{claim}    
\begin{claimproof} 
    Let $\sigma'' = (D(N'',T''), \emb'', \iota'')$ denote the $\{P\to \past\}$-restriction of $\psi$.
    We show $\sigma = \sigma''$ for their three elements individually.

    {\bf Equality of $D(N'',T'')$ and $D(N,T)$}:
    To show this, it is enough to show that that $D(N'',T'')$ is $D(N_0,T_0)$ with the arcs of $A_p$ and vertices of $V_p$ removed, i.e. these arcs and vertices are exactly the ones that become redundant when constructing the $\{P \to \past\}$-restriction of $\psi$.  Since these were exactly the arcs and vertices that were added to $D(N,T)$ to produce $D(N_0,T_0)$, this is enough to show that $D(N'',T'') = D(N,T)$.
    
    Let $g$ be the restriction function $\{P\to \past\}$.
    We claim that for the isolabelling $\iota_u \circ g$, the redundant arcs are exactly those of $A_p$.
    
    Indeed, consider any arc $uv$ in $A_p$. Since $uv \in A_p$, it must have been made $\past$-redundant in the construction of $\sigma'$ from $\psi'$.
    So if $uv$ is a tree arc, then $\iota_0'(z) \in P$ for all $z$ in $\{u,v\}\cup V(\emb_0'(uv))$ (as all these vertices are labelled $\past$ after applying $\{P\to \past\}$).
    Then by construction of $\psi$, we have $\emb_0(uv) = \emb_0'(uv)$, and also $\iota_0(z) = \iota_0'(z)$ for all $z$ in $\{u,v\}\cup V(\emb_0'(uv))$ (since either $z \in V_p$ or $\iota'(z)=\past$).
    Thus $\iota_0(z) \in P$ for all $z$ in $\{u,v\}\cup V(\emb_0(uv))$, and so $uv$ becomes $\past$-redundant after applying $\{P\to \past\}$.
    For a network arc $uv$ in $A_p$, a similar argument holds, but we need to be careful if there is a tree arc $u'v'$ for which $uv$ is in $\emb(u'v')$ (in particular, we would have a problem if $u'v'$ is not an arc in $D(N',T')$, as $u'v'$ could conceivably prevent $uv$ from becoming $\past$-redundant). For such an arc $u'v'$, note that $u'v'$ must also be in $A_p$ (otherwise $u'v'$ is an arc in $D(N',T')$ but $\emb'(uv)=\emb_0(uv)$ is not a path in $D(N',T')$, a contradiction). It follows then that $\iota_0'(u'),\iota_0(v') \in P$. We also have that $\iota_0'(u),\iota_0(v) \in P$. So by a similar argument to tree arcs, we have that $uv$ is redundant after applying $\{P\to \past\}$ in $\psi$, as required.
    
    Conversely consider any redundant arc $uv$ in $D(N_0,T_0)$ after applying $\{P\to \past\}$. As $\psi$ is well-behaved, we may assume  any such arc is $\past$-redundant.
    Then $\iota_0(u), \iota_0(v) \in P$. 
    This implies among other things that $uv$ is not in $A_z$ (as that would require $\iota_0(u), \iota_0(v) \in F \cup \{\future\}$) so $uv$ is also an arc in $D(N_0', T_0')$.
    Similarly if $uv$ is a tree arc then all arcs in $\emb_0(uv)$ as also in $D(N_0', T_0')$, and if $uv$ is a network arc that is part of a path $\emb_0(u'v')$, then $u'v'$ is also an arc in $D(N_0,T_0')$.
    Furthermore $\iota_0'(z) = \iota_0(z)$ for any vertex in one of these arcs (as $\iota'(z)\neq \future$ and $\iota'(z) \notin V_z$).
    It is then easy to see that, just as $uv$ is redundant in $\psi$ after applying $\{P\to \past\}$, $uv$ is also redundant in $\psi'$ after applying $\{P\to \past\}$, and so $uv$ is in $A_p$.
    
   We have now shown that $A_p$ is exactly the set of arcs in $\psi$ that are redundant after applying $\{P\to \past\}$. We now consider the vertices.
    First consider a vertex $v \in V_p$. Then as $V_p$ is in $D(N_0',T_0')$ but not $D(N',T')$, all incident arcs of $v$ are in $A_p$. Thus $v$ becomes isolated after removing redundant arcs from $D(N,T)$. As $v \in V_p$, by construction $\iota_0(v)=\iota_0'(v) \in P$. Similarly this holds for any $v'$ such that $\emb_0'(v') = v$ or $\emb_0'(v)=v'$.  As such, $v$ is $\past$-redundant after applying $\{P \to \past\}$, as required.
    
    Conversely, suppose $v$ is $\past$-redundant w.r.t. $(D(N_0,T_0), \emb_0, \iota_0 \circ \{P \to \past\})$. Then $v$ is isolated after removing $A_p$, and $\iota_0(v), \iota_0(v') \in P$, for $v'$ any vertex such that $\emb_0(v)=v'$ or $\emb_0(v')=v$. These vertices are also in $D(N_0', T_0')$ (they cannot be in $V_z$ as that would require them being labelled with something in $\{z,\future\}$).
    Then by construction $\iota_0'(v) = \iota_0(v), \iota_0'(v') = \iota_0(v')$ are in $P$ as well, and as such $v$ is $\past$-redundant w.r.t $(D(N_0',T_0'), \emb_0', \iota_0' \circ g)$, and so $v \in V_p$, as required.

    {\bf Equality of $\emb''$ and $\emb$}:
    Here we use the fact that $D(N'',T'') = D(N,T)$. Consider any arc $uv$ in $T'' = T$.
    Then by construction, $\emb''(uv) = \emb_0(uv)$, and furthermore $\emb_0(uv) = \emb(uv)$ as $uv$ is an arc in $T$.
    Thus $\emb''(uv) = \emb(uv)$ for all arcs $uv$ in $T''=T$, and so $\emb'=\emb$.

    {\bf Equality of $\iota''$ and $\iota$}:
    Consider any vertex $u$ in $D(N'',T'') = D(N,T)$, and suppose first that $\iota_0(u) = \iota(u)$. Then $\iota_0(u) \notin P$ (as $\iota(u) \in S \cup \{\past, \future\}$), from which it follows that $\iota''(u) = \iota_0(u) = \iota(u)$.
    If on the other hand $\iota_0(u) \neq \iota(u)$, by construction of $\iota_0$ this can only happen if $\iota'(u) = \past$ (we do not need to consider the case $u \in V_p$ as we know $u$ is a vertex of $D(N,T)$). In this case $\iota_0(u) = \iota_0'(u)$, which must be in $P$ (as otherwise $\iota'(u) = \iota_0'(u) \neq \past$.) Then by construction $\iota''(u) = \past$, and also $\iota(u) = \iota'(u) = \past$.
    Thus $\iota''(u) = \iota(u)$ for all vertices $u$ in $D(N,T)$.
\end{claimproof}

As we have now shown that $\psi$ is a well-behaved $F$-partial solution and $\sigma$ is the $\{P\to \past\}$-restriction of $\sigma$, we have that $\sigma$ is valid, as required.
\end{proof}

\subsection{Validity for Join bags}\label{sec:joinBags}

In order to characterize validity for Join bags, we need to introduce a third type of $(S,\Y)$-containment structure.
In our notation we may think of a Join bag as being expressed by the tuple $(P,S,F)$ where $P$ can be decomposed into $L\cup R$, such that the ``left'' child bag is $(L, S, F\cup R)$ and the ``right'' child bag is $(R,S,F\cup L)$.
We want to characterize the validity of a signature $\sigma$ for a Join bag in terms of the validity of signatures for each of the child bags. The main idea is to find two signatures $\sigma_L$ and $\sigma_R$ (for different child bags) which are `compatible', in the sense that the two witnesses for these signatures can be combined, and such that the resulting partial solution is a witness for $\sigma$.
In order to facilitate this characterization, it will be useful to define a ``3-way'' analogue of a signature, called a \emph{reconciliation}, in which we use the labels $\{\pleft,\pright,\future\}$ instead of $\{\past, \future\}$. This is defined below.

\begin{definition}[reconciliation]
  Let $(L\cup R, S, F)$ be a Join bag in the tree decomposition of $\NinTin$
  with child bags $(L, S, F\cup R)$ and $(R,S,F \cup L)$.
  Then, we call an $(S, \{\pleft,\pright, \future\})$-containment structure
a  \emph{reconciliation} for $(L\cup R, S, F)$.
  Such a reconciliation $\mu$ for $(L\cup R, S, F)$ is called \emph{valid}
  if it is the $(L\to \pleft, R\to \pright)$-restriction of a well-behaved $F$-partial solution
  (i.e.\ a $(L\cup R \cup S, \{\future\})$-containment structure).
\end{definition}

\begin{lemma}\label{lem:joinBagOverall}
  Let $(L\cup R, S, F)$ be a Join bag with child bags $(L, S, F\cup R)$ and $(R,S,F \cup L)$,
  and let $\sigma$ be a signature for $(L\cup R, S, F)$.
  Then, $\sigma$ is valid
  if and only if
  there is a reconciliation $\mu$ for $(L\cup R, S, F)$
  and valid signatures $\sigma_L$ and $\sigma_R$ for $(L, S, F \cup R)$ and $(R, S, F \cup L)$, respectively, such that
  \begin{compactenum}[(a)]
    \item $\sigma$ is the $(\{\pleft, \pright\} \to \past)$-restriction of $\mu$,
    \item $\sigma_L$ is the $(\pleft \to \past, \pright \to \future)$-restriction of $\mu$, and 
    \item $\sigma_R$ is the $(\pright \to \past, \pleft \to \future)$-restriction of $\mu$.
  \end{compactenum} 
\end{lemma}

To prove this lemma, we first show the following.

\begin{lemma}\label{cor:valid-recon-WB}
Any valid reconciliation is well-behaved.
\end{lemma}
\begin{proof}
 Let $\mu$ be a valid reconciliation for $(L\cup R,S,F)$, and let $\psi = (D(N,T), \emb, \iota)$  be a well-behaved $F$-partial solution for which $\mu$ is the $g$-restriction, where $g$ is the function $\{L \to \pleft, R \to \pright\}$.
 We first show that there is no arc $uv$ in $D(N,T)$ with $g(\iota(u)) = y, g(\iota(v)) = y'$ for any $y,y' \in \{\pleft, \pright, \future\}$ with $y\neq y'$. 
 Recall that by properties of a tree decomposition, there are no arcs between $F$ and $L\cup R$ in $\NinTin$, nor between $L$ and $R$.
 So now consider a vertex $u \in V(D(N,T))$; we will show 
 that if $g(\iota(u)) \in \{\pleft, \pright\}$ then $u$ has no neighbour $v$ with $g(\iota(v)) \in \{\pleft, \pright, \future\}\setminus\{g(\iota(u))\}$, which is enough to show the claim.
 If $\iota(u) \in L$, then all neighbours of $\iota(u)$ in $\NinTin$ are in $L \cup S$. Furthermore as $\psi$ is a well-behaved $(L\cup R \cup S, \{\future\})$-containment structure, the degree of $u$ in $D(N,T)$ is equal to the degree of $\iota(u)$ in $\NinTin$, and for each neighbour $v'$ of $\iota(u)$ in $\NinTin$, there is a neighbour $v$ of $u$ in $D(N,T)$ with $\iota(v)=v'$. 
 It follows that $\iota(v) \in L \cup S$ and thus $g(\iota(v)) \in S\cup \{\pleft\}$ for any neighbour of $u$ in $D(N,T)$, while $g(\iota(u)) = \pleft$.
 A similar argument shows that if $\iota(u)\in R$, then $g(\iota(v)) \in S\cup \{\pright\}$ for any neighbour of $u$ in $D(N,T)$ and $g(\iota(u)) = \pright$.
 Finally if $\iota(u)\in S\cup \{\future\}$, then $g(\iota(u)) = \iota(u) \notin \{\pleft, \pright\}$, and we are done.
 
 As no vertex labelled $\pleft$ or $\pright$ by $\iota \circ g$ has a neighbour with 
 a different label from~$\{\pleft,\pright,\future\}$, there are in fact no arcs $uv$ in $D(N,T)$ with $g(\iota(u)) = y, g(\iota(v)) = y'$ for any $y,y' \in \{\pleft, \pright, \future\}$ with $y\neq y'$.
We can therefore apply \cref{lem:WBrestriction} to see that $\sigma$ is well-behaved.
\end{proof}

We next motivate the definition of a reconciliation by showing that the validity of a signature for  $(L\cup R, S, F)$ can be characterized by the validity of reconciliations for  $(L\cup R, S, F)$.

\begin{lemma}\label{lem:joinBag}
Let $(L\cup R, S, F)$ be a Join bag with child bags $(L, S, F\cup R)$ and $(R,S,F \cup L)$, and let $\sigma$ be a signature for $(L\cup R, S, F)$.
Then $\sigma$ is valid if and only if there is a valid reconciliation $\mu$ for $(L\cup R, S, F)$ such that $\sigma$ is the $\{\{\pleft, \pright\} \to \past\}$-restriction of $\mu$.
\end{lemma}
\begin{proof}
 Suppose first that $\sigma$ is valid. Then there is a well-behaved $(L\cup R \cup S, \{\future\})$-containment structure $\psi$ (an $F$-partial solution) such that $\sigma$ is the $\{L\cup R \to \past\}$-restriction of $\psi$. Now let $\mu$ be the $\{L \to \pleft, R \to \pright\}$-restriction of $\psi$. By \cref{cor:restrictionIsAStructure}, $\mu$ is an $(S, \{\pleft, \pright, \future\})$-containment structure, and by construction $\mu$ is valid.
 Now let $\sigma'$ be the $\{\{\pleft, \pright\}\to \past\}$-restriction of $\mu$.
  Then transitivity implies that $\sigma'$ is also the $\{L\cup R \to \past\}$-restriction of $\psi$, that is $\sigma' = \sigma$.
  Thus $\sigma$ is the $\{\{\pleft, \pright\}\to \past\}$-restriction of $\mu$, as required.
  
  Conversely, suppose there is a reconciliation  $\mu$ for $(L\cup R, S, F)$ such that $\sigma$ is the $\{\{\pleft, \pright\} \to \past\}$-restriction of $\mu$.
  Then as $\mu$ is valid, there is a well-behaved $F$-partial solution $\psi$ such that $\mu$ is the $\{L\to \pleft, R\to \pright\}$-restriction of $\psi$. Then again by transitivity, $\sigma$ is also the $\{L\cup R\to \past\}$-restriction of $\psi$. Thus $\sigma$ is valid.
\end{proof}

Now we show how the validity of a reconciliation $\mu$ for $(L\cup R, S, F)$ can be characterized by the validity of signatures for the child bags.

\begin{lemma}\label{lem:reconciliationToChildren}
Let $(L\cup R, S, F)$ be a Join bag with child bags $(L, S, F\cup R)$ and $(R,S,F \cup L)$, and let $\mu$ be a well-behaved reconciliation for $(L\cup R, S, F)$.
Let  $\sigma_L$ be the $\{\pleft \to \past, \pright \to \future\}$-restriction of $\mu$, and $\sigma_R$ the $\{\pright \to \past, \pleft \to \future\}$-restriction of $\mu$.
If $\mu$ is valid, then $\sigma_L$ is a valid signature for $(L, S, F \cup R)$ and 
$\sigma_R$  is a valid signature for $(R, S, F \cup L)$.
\end{lemma}

\begin{proof}
 Suppose that $\mu$ is valid. Then there is a well-behaved $F$-partial solution $\psi$ such that $\mu$ is the $\{L\to \pleft, R\to \pright\}$-restriction of $\mu$.
 By construction and \cref{cor:restrictionIsAStructure}, $\sigma_L$ is an $(S, \{\past,\future\})$-containment structure, and thus a signature for $(L, S, F\cup R)$. 
 By transitivity, $\sigma_L$ is the $\{L\to \past, R\to \future\}$-restriction of $\psi$. We will show that $\sigma_L$ is a valid signature for $(L,S,F\cup R)$.
 
Let $\psi_L$ be the $\{R\to \future\}$-restriction of $\psi$. By construction, $\psi_L$ is an $F\cup R$-partial solution and by \cref{lem:WBrestriction} $\psi_L$ is well-behaved. Moreover by transitivity, the $\{L \to \past\}$-restriction of $\psi_L$ is also the $\{L\to \past, R\to \future\}$-restriction of $\psi$. That is, the $\{L \to \past\}$-restriction of $\psi_L$ is $\sigma_L$, and so $\sigma_L$ is valid for $(L, S, F\cup R)$.
(See \cref{fig:ReconciliationEasyDir}.)

\begin{figure}[t]
     \centering
\includegraphics[scale = 0.75]{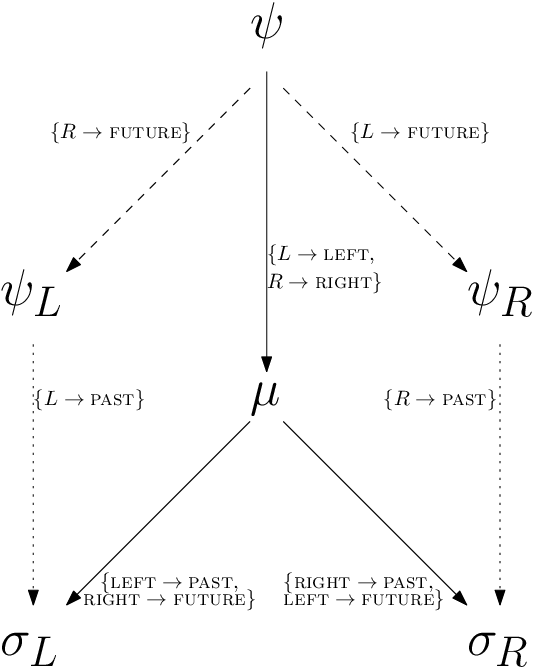}
     \caption{Illustration of proof for reconciliations on Join bags, validity of $\mu$ implies validity of $\sigma_L$ and $\sigma_R$. Solid lines are restriction relations that we may assume; the dashed lines shows the construction of $\psi_L$ and $\psi_R$ from $\psi$; dotted lines are relations we can infer using (multiple uses of) transitivity.}
     \label{fig:ReconciliationEasyDir}
 \end{figure}

A similar argument shows that $\sigma_R$ is a valid signature for $(R,S, F \cup R)$.
\end{proof}

\begin{lemma}\label{lem:childrentoReconciliation}
Let $(L\cup R, S, F)$ be a Join bag with child bags $(L, S, F\cup R)$ and $(R,S,F \cup L)$, and let $\mu$ be a well-behaved reconciliation for $(L\cup R, S, F)$.
Let  $\sigma_L$ be the $\{\pleft \to \past, \pright \to \future\}$-restriction of $\mu$, and $\sigma_R$ the $\{\pright \to \past, \pleft \to \future\}$-restriction of $\mu$.
If  $\sigma_L$ is a valid signature for $(L, S, F \cup R)$ and 
$\sigma_R$  is a valid signature for $(R, S, F \cup L)$, then $\mu$ is valid.
\end{lemma}

\begin{proof}%
Let $\psi_L$ be a well-behaved $F\cup R$-partial solution for which $\sigma_L$ is the $(L\to \past)$-restriction, and similarly
 let $\psi_R$ be a well-behaved $F \cup L$-partial solution for which $\sigma_R$ is the $(R\to \past)$-restriction. 
 Our strategy is to combine $\psi_L$ and $\psi_R$ into an $F$-partial solution $\psi$, and then show that $\mu$ is the $(L\to \pleft, R\to \pright)$-restriction of this $\psi$. 
 
 In what follows, for a given containment structure $\sigma$,  we will write the display graph, embedding function and isolabelling of $\sigma$ as $D(N_{\sigma}, T_{\sigma}), \emb_{\sigma}$ and $\iota_{\sigma}$.
 Thus we have  $\mu = (D(N_{\mu}, T_{\mu}), \emb_{\mu}, \iota_{\mu})$,
   $\sigma_L = (D(N_{\sigma_L}, T_{\sigma_L}), \emb_{\sigma_L}, \iota_{\sigma_L})$, etc.
  Recall that by construction, $D(N_{\sigma_L}, T_{\sigma_L})$ is a subgraph of both $D(N_{\mu}, T_{\mu})$ and $D(N_{\psi_L}, T_{\psi_L})$. 
  Similarly, 
  $D(N_{\sigma_R}, T_{\sigma_R})$ is a subgraph of both $D(N_{\mu}, T_{\mu})$ and $D(N_{\psi_R}, T_{\psi_R})$.

 Now let $A_L$ be the set of arcs, and $V_L$ the set of vertices, that become redundant and are therefore deleted from $D(N_{\psi_L}, T_{\psi_L})$ when deriving the $\{L\to \past\}$-restriction $\sigma_L$ from $\psi_L$. That is,  $A_L = A(D(N_{\psi_L},T_{\psi_L}))\setminus A(D(N_{\sigma_L},T_{\sigma_L}))$
 and $V_L = V(D(N_{\psi_L},T_{\psi_L}))\setminus V(D(N_{\sigma_L},T_{\sigma_L}))$.
 We note that by construction, $\iota_{\psi_L}(v) \in L$ for any $v \in V_L \cup V(A_L)$.
 Similarly, let $A_R,V_R$ be the sets of arcs and vertices that are deleted from $D(N_{\psi_R}, T_{\psi_R})$ in the construction of $\sigma_R$ from $\psi_R$. 
 Let $A_R', V_R'$ be the set of vertices that are deleted from $D(N_{\mu}, T_{\mu})$ in the construction of $\sigma_L$ from $\mu$, and let $A_L', V_L'$ be the set of vertices that are deleted from $D(N_{\mu}, T_{\mu})$ in the construction of $\sigma_R$ from $\mu$.
 (See \cref{fig:ReconciliationConstructiveDir.}.)

 Finally let $A_S$ be the arcs of $D(N_{\mu}, T_{\mu})$ that are not in $A_L'\cup A_R'$,
 and $V_S$ the vertices of $D(N_{\mu}, T_{\mu})$ not in $V_L' \cup V_R'$. Note that the arcs and vertices of $V_S,A_S$ appear in all of $D(N_{\mu}, T_{\mu})$,  $D(N_{\sigma_L}, T_{\sigma_L})$, $D(N_{\sigma_R}, T_{\sigma_R})$, $D(N_{\psi_L}, T_{\psi_L})$, $D(N_{\psi_R}, T_{\psi_R})$.

 Observe that if $v \in V_R'\cup V(A_R')$ then $\iota_{\mu}(v) \in \{\pright,\future\}$. Indeed, when deriving $\sigma_L$ the $\{\pleft \to \past, \pright \to \future\}$-restriction of $\mu$, no arcs or vertices will become $\past$-redundant, since they would previously have been $\pleft$-redundant in $\mu$ (here we use the fact that $\iota_{\mu}$ did not already label any vertices $\past$, unlike $\future$.) So the only vertices and arcs that are removed are ones that become $\future$-redundant, i.e. previously had their vertices labelled $\pright$ or $\future$.
 By a similar argument, if $v \in V_L'\cup V(A_L')$ then $\iota_{\mu}(v) \in \{\pleft,\future\}$.
 Furthermore, we can show that $A_L'\cap A_R' = \emptyset$ and $V_L'\cap V_R' = \emptyset$.
 Indeed, recall that for any arc $a$ in $D(N_{\mu}, T_{\mu})$, there is a set of vertices $Q_a$ such that $a$ is $y$-redundant w.r.t $(D(N_{\mu}, T_{\mu}), \emb_{\mu}, \iota')$ if and only if $\iota(Q_a) = \{y\}$, for any isolabelling $\iota$.  Note that we cannot have $\iota_{\mu}(Q_a) = \{\future\}$, as this implies that $a$ is redundant w.r.t $\mu$, a contradiction as $\mu$ is well-behaved. So if $a$ is $\future$-redundant after applying $\{L\to \past, R \to \future\}$ (i.e. if $a \in A_R'$) then there is at least one $z\in Q_a$ with $\iota_{\mu}(z) \in R$. But then this implies that $a$ is not $\future$-redundant after applying $\{R \to \past, L \to \future\}$, so $a$ is not in $A_L'$. A similar argument shows that $V_L'$ and $V_R'$ are disjoint.
 
 \begin{figure}[t]
    \centering
 \includegraphics[scale = 0.75]{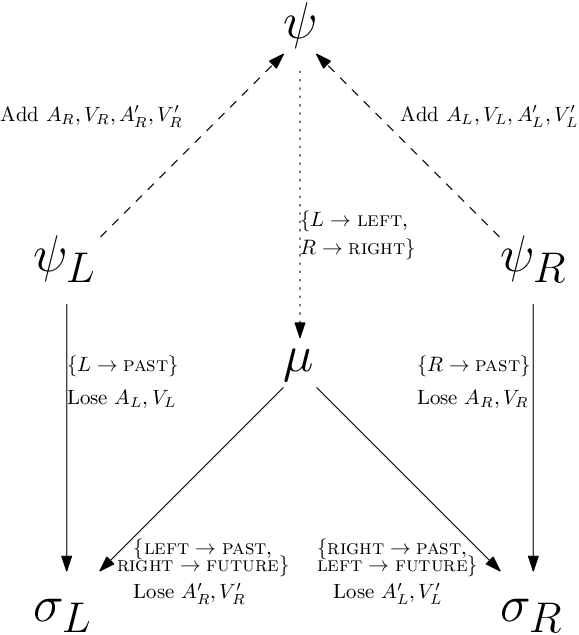}
     \caption{Illustration of proof for reconciliations on Join bags, validity of $\sigma_L$ and $\sigma_R$ implies validity of $\mu$. Solid lines are restriction relations that we may assume; the dashed lines show the construction of $\psi$ from $\psi_L$ or $\psi_R$ (we do not describe the construction of $\emb_0$ or $\iota_0$ in the figure, only the construction of $D(N_0,T_0)$ by adding arcs and vertices); the dotted line shows the relation we want to prove, that $\mu$ is the $\{L\to \pleft, R \to \pright\}$-restriction of $\psi$.}
     \label{fig:ReconciliationConstructiveDir.}
 \end{figure}

 \medskip
 
 We now describe the construction of an $F$-partial solution $\psi = (D(N,T), \emb, \iota)$.
 
 Let $D(N, T)$ be the display graph with vertex set $V_S\cup V_L' \cup V_R' \cup V_L \cup V_R$ and arc set $A_S \cup A_L' \cup A_R' \cup A_L \cup A_R$.
 (We keep the partition of these sets into network side and tree side the same as before.)
 That is, $D(N, T)$ is $D(N_{\psi_L}, T_{\psi_L})$ with the arcs of $A_R \cup A_R'$ and vertices of $V_R \cup V_R'$ added; equivalently we may say 
 $D(N, V)$ is $D(N_{\psi_R}, T_{\psi_R})$ with the arcs of $A_L \cup A_L'$ and vertices of $V_L \cup V_L'$ added, 
or that it is $D(N_{\mu}, T_{\mu})$ with the arcs of $A_L \cup A_R$ and vertices of $V_L \cup V_R$ added.

Let the embedding function $\emb$ be defined as follows. 
For a tree vertex $u$ in $T$, if $u \in V_L$ then let $\emb(u) = \emb_{\psi_L}(u)$, and similarly if $u\in V_R$ then let $\emb(u) = \emb_{\psi_R}(u)$.
If $u \in V_L'$, let $\emb(u) = \emb_{\mu}(u) = \emb_{\sigma_L}(u) = \emb_{\psi_L}(u)$ (note that $u$ is a vertex in $\Disp{\sigma_L}$ as $v \notin V_R'$, and so these terms are all well-defined and equal by construction).
Similarly if $u \in V_R'$, let $\emb(u) = \emb_{\mu}(u) = \emb_{\sigma_R}(u) = \emb_{\psi_R}(u)$.
Finally if $u\in V_S'$, let $\emb(u) = \emb_{\mu}(u) = \emb_{\sigma_L}(u) = \emb_{\sigma_R}(u) = \emb_{\psi_L}(u) = \emb_{\psi_R}(u)$.

For a tree arc $uv$ in $A_L$, let $\emb(uv) = \emb_{\psi_L}(uv)$, and similarly for $uv \in A_R$ let $\emb(uv) = \emb_{\psi_R}(uv)$.
For $uv \in A_L'$, let $\emb(uv) =  \emb_{\mu}(uv) = \emb_{\sigma_L}(uv) = \emb_{\psi_L}(uv)$ (note that $uv$ is an arc in $D(N_{\sigma_L}, T_{\sigma_L})$ as $uv \notin A_R'$).
Similarly for $uv \in A_R'$, let $\emb(uv) =  \emb_{\mu}(uv) = \emb_{\sigma_R}(uv) = \emb_{\psi_R}(uv)$. 
Finally for $uv \in A_S$, let $\emb(uv) =  \emb_{\mu}(uv) = \emb_{\sigma_L}(uv) = \emb_{\sigma_R}(uv) = \emb_{\psi_L}(uv) = \emb_{\psi_R}(uv)$.

Let the $(L \cup R \cup S, \{\future\})$-isolabelling $\iota$ be defined as follows.
For a vertex $v \in V_L$, let $\iota(v) = \iota_{\psi_L}(v)$. Similarly if $v \in V_R$, let $\iota(v) = \iota_{\psi_R}(v)$. Otherwise, $v$ is a vertex of $D(N_{\mu}, T_{\mu})$. If $\iota_{\mu}(v) = \pleft$ then let $\iota(v) = \iota_{\psi_L}(v)$ (recalling that $v \notin V_R'$, so $v$ is a vertex of $D(N_{\sigma_L}, T_{\sigma_L})$ and thus of $D(N_{\psi_L}, T_{\psi_L})$, with $\iota_{\sigma_L}(v) = \past$ and $\iota_{\psi_L}(v) \in L$).
Similarly if $\iota_{\mu}(v) = \pright$ then let $\iota(v) = \iota_{\psi_R}(v)$.
Finally if $\iota_{\mu}(v) \in S \cup \{\future\}$, then set $\iota(v) = \iota_{\mu}(v)$.

We will now show that $\psi$ is indeed a well-behaved $F$-partial solution; afterwards we will show that $\mu$ is the $\{L\to \pleft, R \to \pright\}$-restriction of $\psi$, from which it follows that $\mu$ is valid. (A similar argument can be used to show that $\psi_L$ is the $\{R \to \future\}$-restriction of $\psi$ and $\psi_R$ is the $\{L \to \future\}$-restriction of $\psi$, but we will not need these fact so we do not prove them here.)

\begin{claim}
$\psi$ is a well-behaved $F$-partial solution.
\end{claim}
\begin{claimproof}
    We first show that $\psi$ satisfies one of the properties of well-behaved $(L\cup R \cup S, \{\future\})$-containment structure.
    
\begin{itemize}    
    \item {\bf  For any $u,v \in V(D(N,T))$ with $\iota(u), \iota(v) \in L \cup R \cup S$, if $D(N,T)$ has a path from $u$ to $v$ then there is a path from $\iota(u)$ to $\iota(v)$ in  $\NinTin$:} To see this property, first consider the case that $uv$ is an arc in $D(N,T)$. 
            In this case it is sufficient to show that $\iota(u) = \iota_{\gecs}(u)$ and $\iota(v) = \iota_{\gecs}(v)$ for some $\gecs \in \{\mu, \psi_L, \psi_R\}$, as the property then follows from the fact that $\mu, \psi_L, \psi_R$ are all well-behaved.

            If $uv \in A_L$ then $u\in V(D(N_{\psi_L}, T_{\psi_L}))$ and either $u \in V_L$ or $\iota_{\sigma_L}(u) = \past$ and $\iota_{\mu}(u) = \pleft$; in either case $\iota(u) = \iota_{\psi_L}(u)$, and similarly $\iota(v) = \iota_{\psi_L}(v)$, so we let $\gecs = \psi_L$.
            Similarly if  $uv \in A_R$ then $\gecs = \psi_R$.
            For any arc $uv$ in $D(N_{\mu}, T_{\mu})$, we can let $\gecs = \psi_L$ if $\iota_{\mu}(u) = \iota_{\mu}(v) = \pleft$,  
            $\psi_R$ if $\iota_{\mu}(u) = \iota_{\mu}(v) = \pright$, and $\mu$ if $\iota_{\mu}(u), \iota_{\mu}(v) \in P$. If $\iota_{\mu}(u) = \future$ (resp. $\iota_{\mu}(v)$ = \future) then $\iota(u) = \future$ ($\iota(v) = \future$), so we do not need to consider this case. We also cannot have $\{\iota_{\mu}(u), \iota_{\mu}(v)\} = \{\pleft, \pright\}$ as $\mu$ is well-behaved. It remains to consider the case that one of $\iota_{\mu}(u), \iota_{\mu}(v)$ is $\pleft$ or $\pright$ and the other is in $S$; suppose w.l.o.g. that $\iota_{\mu}(u) = \pleft$ and $\iota_{\mu}(v) \in S$. In this case both $u$ and $v$ are vertices of $D(N_{\sigma_L}, T_{\sigma_L})$ and therefore $D(N_{\psi_L}, T_{\psi_L})$; moreover $\iota(u) = \iota_{\psi_L}(u)$ and $\iota_{\psi_L}(v) = \iota_{\sigma_L}(v) = \iota_{\mu}(v) = \iota(v)$. Thus we can let $\gecs = \sigma_L$.
            
            Now suppose for a contradiction that there exist $u,v \in V(D(N,T))$ with $\iota(u), \iota(v) \in L \cup R \cup S$ such that there is a path from $u$ to $v$ in $D(N,T)$, but no path from $\iota(u)$ to $\iota(v)$ exists in $\NinTin$.
            By the above discussion $uv$ cannot be an arc. Moreover every internal vertex $z$ on the shortest $u-v$ path must have $\iota(z) = \future$ (otherwise $\iota(z) \ \in L\cup R \cup S$ and $u,z$ forms a shorter path).
            Note that any such $z$ is also a vertex in $D(N_{\mu}, T_{\mu})$ with $\iota_{\mu}(z) = \future$ (no other possibility leads to $\iota(z) = \future$; in particular if $\iota_{\mu}(z)=\pleft$ then $\iota(z)=\iota_{\psi_L}(z)\in L$). Any arc incident to $z$ in $D(N,T)$ is also an arc in $D(N_{\mu}, T_{\mu})$ (such arcs cannot be in $A_L$ or $A_R$). So it follows that the path from $u$ to $v$ also exists in $D(N_{\mu}, T_{\mu})$. Finally, we must have $\iota_{\mu}(u), \iota_{\mu}(v) \in S$ (they cannot be in $\{\pleft, \pright\}$ as $u,v$ have neighbours labelled $\future$ by $\iota_{\mu}$ and $\mu$ is well-behaved). So it then follows that there is a path from $\iota(u)=\iota_{\mu}(u)$ to $\iota(v)=\iota_{\mu}(v)$ in $\NinTin$, as $\mu$ is well-behaved.

    \item {\bf $D(N,T)$ is a display graph:} 
        \begin{itemize}
            \item {\bf $D(N, T)$ is acyclic: }  
            Suppose for a contradiction that $D(N,T)$ contains a cycle.
            If $\iota(u) \in L \cup R \cup S$ for some vertex $u$ in this cycle, then as shown above the existence of a path from $u$ to $u$ in $D(N,T)$ implies the existence of a path from $\iota(u)$ to $\iota(u)$ in $\NinTin$, contradicting the fact that $\NinTin$ is acyclic.

            It remains to consider the case that every vertex $z$ in this cycle has $\iota(z) = \future$. However as argued previously, any $z$ with $\iota(z) = \future$ is also a vertex in $D(N_{\mu}, T_{\mu})$, and all its incident arcs in $D(N,T)$ are also arcs of $D(N_{\mu}, T_{\mu})$. Then we have that $D(N_{\mu}, T_{\mu})$ contains a cycle, a contradiction as $D(N_{\mu}, T_{\mu})$ is a display graph.

            \item {\bf $T$ is an out-forest:} As $D(N,T)$ is acyclic, it is remains to show that every vertex of $T$ has in-degree at most  $1$ in $T$.
            To do this, we will show something stronger: that for any vertex $v$ in $D(N,T)$, all of its incident arcs in $D(N,T)$ are arcs in $D(N_{\gecs}, T_{\gecs})$, for some $\gecs \in \{\mu, \psi_L, \psi_R\}$. Then the desired property immediately follows, as every tree vertex has at most one incoming tree-arc in $D(N_{\gecs}, T_{\gecs})$.
            
            So consider any vertex $v \in D(N,T)$. If $v \in V_L$, then the only incident arcs of $v$ are in $A_L$, and therefore all these arcs in $D(N_{\psi_L}, T_{\psi_L})$. Similarly if $v \in V_R$ then all incident arcs are in $D(N_{\psi_R}, T_{\psi_R})$.
            So now we may assume $v$ is a vertex of $\Disp{\mu}$.
            If $\iota_{\mu}(v) = \pleft$ then $v$ can have no incident arcs in $A_R'$ or $A_R$ (as any vertex incident to such an arc must be in $V_R$ or else have $\iota_{\sigma_R}(u) = \past$ and so $\iota_{\mu}(v) = \pright$). Then all incident arcs of $uv$ in $D(N,T)$ are also in $\Disp{\psi_L}$. Similarly if $\iota_{\mu}(v) = \pleft$ then all incident arcs are in $\Disp{\psi_R}$.
            Finally if $\iota_{\mu}(v) \in S \cup \{\future\}$, then again none of its incident arcs are in $A_L$ or $A_R$ (as this would require $\iota_{\sigma_L}(v) = \past$ or $\iota_{\sigma_R}(v) = \past$  and so $\iota_{\mu}(v) \in \{\pleft, \pright\}$), and so so all incident arcs are in $\Disp{\mu}$.
            
            \item {\bf Every vertex in $D(N, T)$ has in- and out-degree at most $2$ and total degree at most $3$:}
            This follows immediately from the previously-shown property that every vertex in $D(N,T)$ has all its incident arcs in one of $\Disp{\mu}, \Disp{\psi_L}, \Disp{\psi_R}$, and the fact that this constraint holds for each of these graphs.
            
            \item {\bf Any vertex in $V(N)\cap V(T)$ has out-degree $0$ and in-degree at most $1$ in each of $T_0$ and $N_0$:} 
            Again this follows from the fact that every vertex in $D(N,T)$ has all its incident arcs in one of $\Disp{\mu}, \Disp{\psi_L}, \Disp{\psi_R}$.
        \end{itemize}
    
    \item {\bf $\emb$ is an embedding function on $D(N,T)$:}
    \begin{itemize}
        \item {\bf For each $u\in V(T)$, $\emb(u)\in V(N)$ and,
        for each arc $uv \in A(T)$, $\emb(uv)$ is a directed $\emb(u)$-$\emb(v)$-path in $N$:}
        
        The fact that $\emb(u)\in V(N)$ follows immediately from the fact that $\emb(u) = \emb_{\gecs}(u)$ for some $\gecs\in \{\mu, \psi_L, \psi_R\}$, and so $\emb(u)$ is a network vertex.
        
        To see that $\emb(uv)$ is a path from $\emb(u)$ to $\emb(v)$, we will show that there is $\gecs \in \{\mu,\psi_L, \psi_R\}$ such that $\emb(uv) = \emb_{\gecs}(uv)$, $\emb(u) = \emb_{\gecs}(u)$ and $\emb(v) = \emb_{\gecs}(v)$. The result then follows from the fact that $\emb_{\gecs}$ is an embedding function.

     If $uv \in A_L$, then neither $u$ nor $v$ can be in $V_R$, 
     nor can they be in $V_R'$ (note that $u \in V_R'$ requires $\iota_{\mu}(u) \in \{\pright, \future\}$, but $uv \in A_L$ implies either $u \in V_L$ or $\iota_{\sigma_L}(u) = \past$ and thus $\iota_{\mu}(u) = \pleft$).
     It follows by construction of $\emb$ that $\emb(u) = \emb_{\psi_L}(u)$, $\emb(v) = \emb_{\psi_L}(u)$, and $\emb(uv) = \emb_{\psi_L}(uv)$, so we can let $\gecs = \psi_L$.
     Similarly if $uv \in A_R$ we can let $\gecs = \psi_R$.
     For $uv \in A_L'\cup A_R' \cup A_S$, we have that $u,v$ are vertices in $\Disp{\mu}$, and therefore not in $V_L$ or $V_R$. It follows by construction that $\emb(u) = \emb_{\mu}(u)$, $\emb(v) = \emb_{\mu}(u)$, and $\emb(uv) = \emb_{\mu}(uv)$, so we let $\gecs = \mu$.

        \item {\bf for any distinct $u,v \in V(T)$, $\emb(u)\neq \emb(v)$}: 
       We first show that for any $u \in V(T)$, $u \in V_L$ if and only if $\emb(u) \in V_L$. Recall that $V_L$ is the set of vertices in $\Disp{\psi_L}$ that become $\textsc{past}$-redundant after applying $\{L\to \past\}$, and observe that by the definition of $y$-redundancy, a tree vertex $u$ is $y$-redundant with respect to some containment structure $(D(N',T'), \emb', \iota')$ if and only if $\emb'(u)$ is $y$-redundant. Thus $u \in V_L$ if and only if $\emb(u) \in V_L$. By a similar argument, $u \in V_R$ if and only if $\emb(u) \in V_R$.

        Now suppose for a contradiction that there exist distinct $u,v \in V(T)$ with $\emb(u) = \emb(v)$. If $\emb(u) = \emb(v) \in V_L$, then also $u,v \in V_L$. Thus $u,v$ are distinct vertices in $\Disp{\psi_L}$ with $\emb_{\psi_L}(u) = \emb(u) = \emb(v) = \emb_{\psi_R}(v)$, a contradiction as $\emb_{\psi_L}$ is an embedding function.
        We get a similar contradiction if $\emb(u) \in V_R$.
        Finally if $\emb(u) \notin V_L \cup V_R$, then also $u,v\notin V_L\cup V_R$. Thus $u,v$ are distinct vertices in $\Disp{\mu}$ with $\emb_{\mu}(u) = \emb(u) = \emb(v) = \emb_{\mu}(v)$, again a contradiction as $\emb_{\mu}$ is an embedding function.
        
        \item {\bf for any $u \in V(T) \cap V(N)$, $\emb(u) = u$:} This follows immediately from the fact that $u$ is a vertex in $V(T_{\gecs})\cap V(N_{\gecs})$ for some $\gecs\in \{\mu, \psi_L, \psi_R\}$, and for such a $\gecs$ $\emb(u) = \emb_{\gecs}(u) = u$.

        \item {\bf the paths $\{\emb(uv) \mid uv\in A(T)\}$ are arc-disjoint:} 
        
        Similar to the proof that $\emb(u) \neq \emb(v)$ for $u \neq v$, we observe that if an arc $uv$ is $y$-redundant with respect to some $(D(N',T'), \emb', \iota')$, then so are all the arcs of $\emb'(uv)$.
        It follows by construction that if $uv \in A_L$ (resp. $A_R, AL',A_R', A_S$) then so are all arcs of $\emb(uv)$.
        So now suppose for a contradiction that there exist distinct tree arcs $uv$, $u'v'$ such that $\emb(uv), \emb(u'v')$ share an arc. Then $uv$ and $u'v'$ are both in the same set from $\{A_L, A_R, A_L', A_R', A_S\}$, and so in particular they are both in $\Disp{\gecs}$ for some $\gecs \in \{\mu, \psi_L, \psi_R\}$, with $\emb(uv) = \emb_{\gecs}(uv)$, $\emb(u'v') = \emb_{\gecs}(u'v')$. Then $\emb_{\gecs}(uv), \emb_{\gecs}(u'v')$ share an arc, a contradiction as $\emb_{\gecs}$ is an embedding function.

        \item {\bf for any distinct $p,q\in A(T)$, $\emb(p)$ and $\emb(q)$ share a vertex $z$
        only if $p$ and $q$ share a vertex $w$ with $z = \emb(w)$:}
        Here we use a couple of properties that have been proved earlier. Recall from the proof that $T$ is an out-forest, that for any vertex $v$ in $D(N,T)$, all its incident arcs in $D(N,T)$ are arcs in $\Disp{\gecs}$ for some $\gecs \in \{\mu, \psi_L, \psi_R\}$.
        Recall also, from the proof that the paths $\{\emb(uv) | uv \in A(T)\}$ are arc disjoint, that a tree arc $uv$ is in $A_L$ (resp. $A_R, A_L', A_R', A_S$) then so are all arcs in $\emb(uv)$.
        
        So now consider distinct arcs $p,q \in A(T)$ and suppose $\emb(p)$ and $\emb(q)$ share a vertex $z$. Suppose all incdent arcs of $z$ are in $\Disp{\mu}$ (the cases where all incident arcs are in $\Disp{\psi_L}$ or $\Disp{\psi_R}$ are similar). Then all incident arcs of $z$ from $\emb(p)$ and $\emb(q)$ are in $A_L'\cup A_R' \cup A_S$, which implies that $p,q$ are in $A_L'\cup A_R' \cup A_S$ as well.
        Thus $p,q$ are both arcs in $\Disp{\mu}$.
        It follows that $\emb(p) = \emb_{\mu}(p)$ and $\emb(q) = \emb_{\mu}(q)$ also share the vertex $z$, which implies that $p,q$ share a vertex $w$ with $z = \emb_{\mu}(w)$, as $\emb_{\mu}$ is an embedding function. Finally observe that $w \in V_L'\cup V_R'\cup V_S$ (as $w$ is in $\Disp{\mu}$) so $\emb(w) = \emb_{\mu}(v) = z$, as required.
    \end{itemize}

    \item {\bf $\iota$ is a $(L \cup R \cup S, \{\future\})$-isolabelling:}
    \begin{itemize}
        \item {\bf For $u \in V(D(N, T))$ with $\iota(u) \neq \future$, $\iota(u) \in V(\Nin)$ only if $u \in V(N)$ and $\iota(u) \in V(\Tin)$ only if $u \in V(T)$ :}
        This follows from the fact that $\iota(u) = \iota_{\gecs}(u)$ for some $\gecs \in \{\mu, \psi_L, \psi_R\}$.
        
        \item {\bf For $u,v \in V(D(N, T))$ with $\iota(u), \iota(v) \neq \future$, $\iota(u) = \iota(v)$ only if $u = v$ :}
        
        Consider some $u,v \in V(D(N,T))$ with $\iota(u) = \iota(v) \neq \future$.
        If $\iota(u) \in L$, then $\iota(u) = \iota_{\psi_L}(u)$, and $\iota(v)=\iota_{\psi_L}(v)$, as of the isolabelings $\iota_{\psi_L}, \iota_{\psi_R}, \iota_{\mu}$, only $\iota_{\psi_L}$ maps anything to $L$. Thus $\iota_{\psi_L}(u) = \iota_{\psi_L}(v) \in L$, from which it follows that $u=v$ since $\psi_L$ is an isolabelling.
        Similarly, if $\iota(u) \in R$ then $\iota_{\psi_R}(u) = \iota_{\psi_R}(v) \in R$ and so $u=v$.
        If $\iota(u) \in S$, then $u$ and $v$ must be vertices of $\Disp{\mu}$ with $\iota(u) = \iota_{\mu}(u)$ and  $\iota(v) = \iota_{\mu}(v)$ (other possibilities, such as $u \in V_L$ or $\iota_{\mu}(u)=\pright$, imply $\iota(u)$ or $\iota(v)$ are in in $L\cup R$). Then again as $\iota_{\mu}$ is an isolabelling, we have $u = v$.
        
        \item {\bf For $u,v \in V(D(N, T))$ with $\iota(u), \iota(v) \neq \future$, the arc $uv$ is in $D(N,T)$ if and only if $\iota(u)\iota(v)$ is in $\NinTin$ :}
        
        First suppose that $uv$ is an arc in $D(N,T)$, for $\iota(u), \iota(v) \neq \future$. We aim to show that there is $\gecs \in \{\psi_L, \psi_R, \mu\}$ such that $uv$ is an arc in $\Disp{\gecs}$ and $\iota(u) = \iota_{\gecs}(u)$, $\iota(v) = \iota_{\gecs}(v)$. Then $\iota(u)\iota(v) \in A(\NinTin)$ follows from the fact that $\gecs$ is well-behaved.
        
        If $uv \in A_L$ then $uv$ is an arc of $\Disp{\psi_L}$, and either $u \in V_L$ or $\iota_{\mu}(u) = \pleft$, which implies $\iota(u) = \iota_{\psi_L}(u)$. Similarly $\iota(v) = \iota_{\psi_L}(v)$, so we can let $\gecs = \psi_L$.
        By a similar argument, if $uv \in A_R$ then we let $\gecs = \psi_R$.
        For $uv \in A_L'$, we again have that $uv$ is an arc of $\Disp{\psi_L}$, and $\iota_{\mu}(u) = \pleft$, which implies $\iota(u)=\iota_{\psi_L}(u)$. Similarly $\iota(v) = \iota_{\psi_L}(v)$, so we can let $\gecs = \psi_L$. By a similar argument, if $uv \in A_R'$ then we let $\gecs = \psi_R$.
        For $uv \in A_S$, we make use of the fact that $\{\iota_{\mu}(u), \iota_{\mu}(v)\} \neq \{\pleft, \pright\}$ (as $\mu$ is well-behaved). So suppose w.l.o.g that $\iota_{\mu}(u), \iota_{\mu}(v) \neq \pright$. Then either $\iota_{\mu}(u) \in S$, in which case $\iota_{\psi_L}(u) = \iota_{\mu}(u) = \iota(u)$, or $\iota_{\mu}(u) = \pleft$, in which case $\iota(u) = \iota_{\psi_L}(u)$. Thus  in either case $\iota(u) = \iota_{\psi_L}(u)$, and similarly $\iota(v) = \iota_{\psi_L}(v)$. Note also that $uv \in A(\Disp{\psi_L})$. Thus we can let $\gecs = \psi_L$, as required.
        
        For the converse, suppose that $\iota(u)\iota(v)$ is an arc in $\NinTin$.
        By the properties of a tree decomposition, we cannot have $\iota(u) \in L$ and $\iota(v) \in R$ (or vice-versa).
        If $\iota(u), \iota(v) \in L$ then it must hold that $\iota(u) = \iota_{\psi_L}(u)$ and  $\iota(v) = \iota_{\psi_L}(v)$. Then as $\psi_L$ is well-behaved, $uv$ is an arc in $\Disp{\psi_L}$, which also implies that $uv$ is an arc in $D(N,T)$.
        A similar argument holds if $\iota(u), \iota(v) \in R$.
        If $\iota(u), \iota(v) \in S$ then it must hold that $\iota(u) = \iota_{\mu}(u)$ and $\iota(v) = \iota_{\mu}(v)$. Then again as $\mu$ is well-behaved, $uv$ is an arc in $\Disp{\mu}$ and thus $D(N,T)$.
        It remains to consider the case where one of $\iota(u), \iota(v)$ is in $S$ and the other is in $L$ or $R$. Suppose w.l.o.g. that $\iota(u) \in L, \iota(v) \in S$. Then $\iota(u) = \iota_{\psi_L}(u)$, and $\iota(v) = \iota_{\mu}(v)$, seemingly a problem. However, notice that as $\iota_{\mu}(v) \in S$, we also have that $v$ is a vertex of $\Disp{\sigma_L}$ and $\Disp{\psi_L}$, with $\iota_{\psi_L}(v) = \iota_{\sigma_L}(v) = \iota_{\mu}(v) = \iota(v)$.
        Thus $\iota_{\psi_L}(u)\iota_{\psi}(v)$ is an arc in $\NinTin$, from which it follows that $uv$ is an arc in $\Disp{\psi_L}$ and thus $D(N,T)$.

        \item {\bf $\iota(V(D(N,T))) \cap V(\NinTin) = L \cup R \cup S$ :}
        By construction, $\iota(v) \in L \cup R \cup S \cup \{\future\}$ for any $v \in V(D,T)$, so  $\iota(V(D(N,T))) \cap V(\NinTin) \subseteq L \cup R \cup S$. It remains to show that for any $z \in L \cup R \cup S$, there is $v \in V(D(N,T))$ for which $\iota(v) = z$.
        If $z \in L$, then as $\iota_{\psi_L}$ is a $(L\cup S, \{\future\})$-isolabelling, there is $v \in V(\Disp{\psi_L})$ such that $\iota_{\psi_L}(v) = z$. Then either $v \in V_L$, or $v$ is a vertex of $\Disp{\sigma_L}$ and  $\Disp{\mu}$ with $\iota_{\sigma_L}(v) = \past$, $\iota_{\mu}(v) = \pleft$ (since $\sigma_L$ is the $\{L\to \past\}$-restriction of $\psi_L$). In either case $\psi(v) = \psi_L(v)$, and so $\psi(v) = z$ as required.
        By a similar argument, if $z \in R$ then there is $v \in V(D(N,T))$ with $\iota(v) = \iota_{\psi_L}(v) = z$.
        Finally if $z \in S$, there is a vertex $v$ in $\Disp{\mu}$ with $\iota_{\mu}(v) = z$. As $\iota_{\mu}(v)\notin \{\pleft, \pright\}$, by construction we have $\iota(v) = \iota_{\mu}(v) = z$, as required.
        
    \end{itemize}
    
    \item {\bf Each vertex $u$ with $\iota(u) \in L \cup R \cup S$ has the same in- and out-degree in $D(N,T)$ as $\iota(u)$ in $\NinTin$:} 
    Recall, from the proof that $T$ is an out-forest, that any vertex $v$ in $D(N,T)$ has all its incident vertices contained in $\Disp{\gecs}$, for some $\gecs \in \{\psi_L, \psi_R, \mu\}$.
    Moreover if $u \in V_L$ or $\iota_{\mu}(u) = \pleft$ (in which case $\psi(u) = \psi_L(u)$ by construction) then all incident arcs of $u$ are in $\Disp{\psi_L}$. Thus $u$ has the same in- and out-degree in $D(N,T)$ as $\iota(u) = \iota_{\psi_L}(u)$ in $\NinTin$, since $\psi_L$ is a containment structure.
    A similar argument holds if $u \in V_R$ or $\iota_{\mu}(u) = \pright$.
    Finally if $\iota_{\mu}(u) \in S$, then all incident arcs of $u$ are in $\Disp{\mu}$, and so $u$ has the same degree in $D(N,T)$ as $\iota(u) = \iota_{\mu}(u)$ does in $\NinTin$.

    \item {\bf Each vertex $u$ of $T$ with $\iota(u) \neq \iota(\emb(u))$ has $2$ out-arcs in $D(N,T)$ :}
    
    In the case that $\iota(u) \in L \cup R \cup S$, this follows from previously-shown properties. In particular, $u$ is not in $V(T)\cap V(N)$, as this would imply $\emb(u) = u$ (because $\emb$ is an embedding function) and thus $\iota(u) = \iota(\emb(u))$.
    Then $u \notin V(N)$, which implies $\iota(u) \notin V(\Nin)$ as $\iota$ is an isolabelling. Thus we may assume $\iota(u)$ is an internal vertex of $\Tin$, which has out-degree $2$ as $\Tin$ is binary. As we have just shown that $u$ with $\iota(u) \in L \cup R \cup S$ has the same in- and out-degree in $D(N,T)$ as $\iota(u)$ in $\NinTin$, this implies that $u$ has out-degree $2$ in $D(N,T)$, as required.
    
    For the case that $\iota(u) = \future$, we observe that by construction, $\iota(u) = \iota_{\mu}(u)$.
    Furthermore $\iota_{\mu}(\emb_{\mu}(u)) \neq \future$ (as  $\emb_{\mu}(u) \in V_L' \cup V_R' \cup V_S$ and so $\emb(u) = \emb_{\mu}(u)$, and $\iota_{\mu}(\emb_{\mu}(u)) = \future$ would imply $\iota(\emb(u)) = \iota_{\mu}(\emb_{\mu}(u)) = \future = \iota(\emb(u))$, a contradiction).
    Thus as $\mu$ is a containment structure, we also have that $u$ has out-degree $2$ in $\Disp{\mu}$ and so $u$ has out-degree $2$ in $D(N,T)$ as well.
    
    \medskip 
    
    The above conditions show that $\psi$ is a $(L \cup R \cup S, \{\future\})$-containment structure, i.e. an $F$-partial solution. Next we show that $\psi$ is well-behaved:
    
    \item {\bf $D(N, T)$ contains no redundant arcs or vertices:} 
    Suppose for a contradiction that there is a redundant arc or vertex w.r.t. $\psi$. We will show that such an arc or vertex is also redundant w.r.t $\mu$, a contradiction as $\mu$ is well-behaved.
    
    Consider the case where a tree arc $uv$ is redundant w.r.t $\psi$. Then $uv$ is necessarily $\future$-redundant, which implies that $\iota(u) = \iota(v) = \future$.
    It follows by construction of $\iota$ that $u,v$ are vertices of $\Disp{\mu}$ with $\iota_{\mu}(u) = \iota_{\mu}(v) = \future$.
    In addition, we have that $\iota(z) = \future$ for any $z\in V(\emb(uv))$, which again implies $\iota_{\mu}(v) = \future$ for any such edge.
    Finally $\emb(uv) = \emb_{\psi}(uv)$ (as the arc $uv$ is in one of $A_L',A_R', A_S$) , and so $\iota_{\mu}(z) = \future$ for all $z \in \{u,v,\} \cup V(\emb_{\mu}(uv))$. It follows that $uv$ is $\future$-redundant w.r.t $\mu$, the desired contradiction.
    
    Essentially the same arguments can also be used for redundant network arcs and redundant tree or network vertices.

    \item {\bf For any $y,y' \in \{\future\}$ with $y \neq y$, there is no arc $uv$ in $D(N, T)$ for which $\iota(u) = y$ and $\iota(v) = y'$:} This follows immediately from the fact that  $|\{\future\}| = 1$, so there are no such $y,y'$.
    
    \item {\bf For any $u,v,$ in $V(D(N, T))$ with $\iota(v), \iota(v) \in L \cup R \cup S$, if $D(N, T)$ has a path from $u$ to $v$ then $\NinTin$ has a path from $\iota(u)$ to $\iota(v)$:} 
    This has already been shown, at the start of the proof for this claim.\claimqedhere
\end{itemize}
\end{claimproof}

We now have that $\psi$ satisfies all the conditions of a well-behaved $F$-partial solution.
Finally we need to show that $\mu$ is the $\{L \to \pleft, R \to \pright\}$-restriction of $\psi$. 
    
\begin{claim}
    $\mu$ is the $\{L \to \pleft, R \to \pright\}$-restriction of $\psi$.
\end{claim}    
\begin{claimproof} 
    Let $\mu'= (D(N_{\mu'},T_{\mu'}), \emb_{\mu'}, \iota_{\mu'})$ denote the $\{L \to \pleft, R \to \pright\}$-restriction of $\psi$. Then our aim is to show that $\mu = \mu'$.

\begin{itemize}
    \item {\bf $\Disp{\mu} = \Disp{\mu'}$:} To show this, it is enough to show that that $\Disp{\mu'}$ is $D(N,T)$ with the arcs of $A_L \cup A_R$ and vertices of $V_L \cup V_R$ removed.
    To do this, we will show that after applying $\{L \to \pleft, R \to \pright\}$, the $\pleft$-redundant arcs (resp. vertices) are exactly $A_L$ ($V_L$) and the $\pright$-redundant arcs (vertices) are exactly $A_R$ ($V_R$).
    
    Recall that for any arc or vertex $a$ in $D(N,T)$, $a$ is $y$-redundant w.r.t $(D(N,T), \emb, \iota')$ if and only if $\iota(Q_a) = \{y\}$, for some set $Q_a$ of vertices that depends only on $D(N,T)$ and $\emb$.
    Note that $Q_a$ always contains $a$ itself, if $a$ is a vertex, or the vertices of $a$, if $a$ is an arc
    
    So now let $Q_a$ be the set of vertices that determines whether $a$ is $y$-redundant for $(D(N,T), \emb)$, and let $Q_a'$ denote the set of vertices that determine whether $a$ is $y$-redundant for $(\Disp{\psi_L}, \emb_{\psi_L})$, where $a$ is any arc or vertex that in $\Disp{\psi_L}$.
    Note that $Q_a' \subseteq Q_a$. We will show first that $\iota(Q_a) \subseteq L$ if and only if $\iota_{\psi_L}(Q_a')\subseteq L$.

    Indeed, since $\iota(u) \in L$ implies $\iota_{\psi_L}(u) = \iota(u)$,
    if  $\iota(Q_a) \subseteq L$ then  $\iota_{\psi_L}(Q_a') \subseteq \iota(Q_a) \subseteq L$.
    For the converse, if $\iota_{\psi_L}(u) \in L$ then we also have $\iota(u) = \iota_{\psi_L}(u)$,
    so $\iota_{\psi_L}(Q_a') \subseteq L$ implies  $\iota(Q_a') \subseteq L$.
    It remains to consider the vertices of $Q_a\setminus Q_a'$.
    If $a=uv$ is a network arc, then
    $Q_a = \{u,v\} \cup Q_{a'}:a \in A(\emb(a')$ and
    $Q_a' = \{u,v\} \cup Q_{a'}':a \in A(\emb_{\psi_L}(a'))$.
    Then $Q_a = Q_a'$, unless there is an arc $a'$ in $D(N,T)$ with $a\in A(\emb(a'))$ and $a'$ not in $\Disp{\psi_L}$.
    Note however that this requires that $a' \in A_R' \cup A_R$,
    which as argued previously would imply that $A(\emb(a')) \subseteq A_R' \cup A_R$,
    a contradiction to the fact that $\iota(u), \iota(v) \in L$.
    If $a$ is a vertex, a similar argument applies: with $Q_a = Q_a'$
    unless $D(N,T)$ has arcs in $A_R\cup A_R'$ incident to $a$ or $a'$
    (with $a'$ a possible vertex such that $\emb(a)=a'$ or $\emb(a')=a$).
    But as $\iota(a), \iota(a') \in L$, this cannot happen. Thus in either case we have $Q_a'=Q_a \subseteq L$.

    We can now show that an arc or vertex $a$ in $D(N,T)$ is $\pleft$-redundant w.r.t $(D(N,T), \emb, \iota \circ g)$ if and only if $a$ is in $\Disp{\psi_L}$ and $a$ is $\past$-redundant w.r.t. $(\Disp{\psi_L}, \emb_{\psi_L}, \iota_{\psi_L}\circ g')$, where $g$ is the function $\{L \to \pleft, R \to \pright\}$, $g'$ is the function $\{L \to \past\}$. Indeed we may assume that all vertices in $a$ are mapped to elements of $L$, as otherwise neither side holds.
    Then it remains to observe that $a$  is $\pleft$-redundant w.r.t $(D(N,T), \emb, \iota \circ g)$ if and only if $g(\iota(Q_a)) = \{\pleft\} \Leftrightarrow \iota(Q_a) \subseteq L \Leftrightarrow \iota_{\psi_L}(Q_a') \subseteq L \Leftrightarrow g'(\iota_{\psi_L}(Q_a')) = \past \Leftrightarrow$ $a$ is $\past$-redundant w.r.t. $(\Disp{\psi_L}, \emb_{\psi_L}, \iota_{\psi_L}\circ g')$.
    
    As such, after applying $g = \{L \to \pleft, R \to \pright\}$ the $\pleft$-redundant arcs and vertices in $D(N,T)$ are exactly $A_L, V_L$.
    A similar argument shows that the $\pright$-redundant arcs and vertices are $A_R, V_R$.
    It follows that $\Disp{\mu'}$ is exactly $D(N,T)$ with $A_L,V_L,A_R,V_R$ removed, that is, $\Disp{\mu'} = \Disp{\mu}$.
    
    \item {\bf $\emb_{\mu} = \emb_{\mu'}$:} Here we use the fact that $\Disp{\mu} = \Disp{\mu'}$.
    Thus every tree arc $uv$ of $\Disp{\mu'}$ is in $A_L'\cup A_R' \cup A_S$, from which it follows by construction that $\emb(uv) = \emb_{\mu}(uv)$. Similarly as tree every vertex $u$ of $\Disp{\mu'}$ is in $V_L'\cup V_R' \cup V_S$, we have $\emb(u) = \emb_{\mu}(u)$.
    Then as $\emb_{\mu'}$ is the restriction of $\mu$ to the tree vertices of $V_L' \cup V_R' \cup V_S$ and the tree arcs of $A_L' \cup A_R' \cup A_S$, we have $\emb_{\mu'} = \emb_{\mu}$, as required.

    \item {\bf $\iota_{\mu} = \iota_{\mu'}$:} Consider any vertex $u$ in $\Disp{\mu}= \Disp{\mu}$.
    If $\iota_{\mu}(u) = \pleft$, then by construction $\iota(u) = \iota_{\psi_L}(u)$, which is in $L$ (as $\iota_{\sigma_L}(u) = \past$). Then by construction of $\mu'$, $\iota_{\mu'}(u) = \pleft$ as well.
    A similar argument shows that $\iota_{\mu'}(u) = \pright$ if $\iota_{\mu}(u) = \pright$.
    Finally if $\iota_{\mu}(u) \in S \cup \future$, then by construction $\iota(u) = \iota_{\mu}(u)$, and since $\iota(u) \notin L \cup R$, we  have  $\iota_{\mu'}(u) = \iota(u) = \iota_{\mu}(u)$.
    Thus in all case $\iota_{\mu'}(u) = \iota_{\mu}(u)$, and so $\mu = \mu$, as required.\claimqedhere
\end{itemize}    
\end{claimproof} 

As we have now shown that $\psi$ is a well-behaved $F$-partial solution and $\mu$ is the $\{L \to \pleft, R \to \pright\}$-restriction of $\psi$, we have that $\mu$ is valid, as required.
\end{proof}

From the three previous lemmas we immediately have \cref{lem:joinBagOverall}.

\subsection{Compressing signatures}\label{sec:compact}

So far, we have shown relations between valid signatures such that the validity of
any signature~$\sigma$ for a bag~$(P,S,F)$ in the tree decomposition of $\NinTin$
is determined by the validity of all signatures for the child bags of $(P,S,F)$.
The final step is to contract certain long paths that may occur in the 
signature.
This is necessary in order to bound the size, and thus the number of possible signatures, for a given bag.
This is summarized by the notion of \emph{compact} signatures, described below.

All our results relating the validity of well-behaved signatures also hold for compact signatures, that is, whether a compact signature is compact-valid can be determined by looking at the compact signatures for the child bags. (See~\cref{lem:compactLeafBag,lem:compactForgetBag,lem:compactIntroduceBag,lem:compactJoinBag}.)

\begin{definition}[compact form]
Let $\psi = (D(N,T), \emb, \iota)$ be a well-behaved $(S, \Y)$-containment structure.
Then the \emph{compact form of $\psi$}, denoted $c(\psi)$, is the $(S, \Y)$-containment signature derived from $\psi$ as follows:

For $y \in \Y$:
\begin{itemize}
    \item if there is a path $x_1,x_2,x_3$ in $N$ with $x_2$ having in-degree $1$ and out-degree $1$ in $N$, and $\iota(\{x_1,x_2,x_3\}) = \{y\},$ and $\emb(u) \neq x_2$ for any vertex $u$ in $T$, then delete vertex $x_2$ and arcs $x_1x_2$, $x_2x_3$ from $N$, and add the arc $x_1x_3$.
    For the arc $uv$ in $T$ for which $x_1x_2x_3$ is part of the path $\emb(uv)$, replace $x_1x_2x_3$ in $\emb(uv)$ with $x_1x_3$.
\end{itemize}

That is, we suppress a vertex $x$ in $N$ (and adjust $\emb$ as necessary), if $x$ has a single in-neighbour and out-neighbour, and all three are mapped to a label  $y\in Y$ by $\iota$, and no vertex of $T$ is mapped to $x$ by $\emb$.

We call such a path a \emph{long $y$-path} and refer to the above process as \emph{suppressing long $y$-paths}.

If $c(\psi) = \psi$ then we say $\psi$ is a \emph{compact} $(S, \Y)$-containment structure.

We call a compact signature $\sigma$ for $(P,S,F)$ \emph{compact-valid} if
there is a compact $F$-partial solution~$\psi$ such that
$\sigma = c(\sigma')$ for $\sigma'$ is the $(P\to \past)$-restriction of $\psi$.
\end{definition}

Note that by definition, any compact $(S,\Y)$-containment structure is also a well-behaved $(S,\Y)$-containment structure.

\begin{lemma}\label{lem:compactStructuresAreStructures}
If $\psi$ is a well-behaved $(S, \Y)$-containment structure, then  $c(\psi)$ is a compact $(S, \Y)$-containment structure.
\end{lemma}
\begin{proof}
Let $\psi = (D(N,T), \emb, \iota)$ and let $c(\psi) = (D'(N',T'), \emb', \iota')$.
It is clear from the construction of $c(\psi)$ that $D'(N',T')$ is a display graph and that $\emb'$ is an embedding of $T'$ into $N'$.
Moreover $\iota'(u) = \iota(u)$ for each vertex $u$ in $D'(N',T')$, and the only vertices of $D(N,T)$ that were suppressed in the construction of $D'(N,T')'$ were vertices $v$ for which $\iota(v) \in \{\past, \future\}$. As such $\iota'$ satisfies all the requirements for the isolabelling function in a $(S, \Y)$-containment structure, and so $(D'(N',T'), \emb', \iota')$ is a $(S, \Y)$-containment structure.

It is also clear that $c(\psi)$ is well-behaved, by construction and the fact that $\psi$ is well-behaved.

Finally observe that $D(N',T')$ contains no long $y$-paths for any $y \in \Y$, a all such paths are suppressed in the construction of $\psi'$.
\end{proof}

We will now show that for the purposes of our dynamic programming algorithm, it is enough to restrict our attention to compact signatures.
This allows us to get an FPT bound on the number of signatures we have to consider.
In order to do this, we define an analogue of ``restriction'' that allows us to derive compact $(S',\Y)$-containment structures
from $(S, \Y)$-containment structures.

\begin{definition}[compact restriction]\label{def:comprestrict}
Let~$\psi = (D(N,T), \emb, \iota)$ be an $(S, \Y)$-containment structure,
let~$S' \subseteq S$, and
let~$g:S \cup \Y \to S' \cup \Y$
be  a restriction function
(i.e. such that for all $v \in S \cup \Y$, $g(v)=v$ if $v \in S'$ and $g(v) \in \Y$ otherwise).
Then we define the \emph{compact-$g$-restriction} of $\psi$ to be the compact form of the $g$-restriction of $\psi$.

When $\sigma$ is a compact signature for $(P,S,F)$, we say that $\sigma$ is \emph{compact-valid} if there is a compact $F$-partial solution $\psi$ such that $\sigma$ is the compact-$\{P\to \past\}$-restriction of $\psi$.

Similarly, when $\mu$ is a compact reconciliation for $(L\cup R, S,F)$, we say $\mu$ is \emph{compact-valid} if there is a compact $F$-partial solution $\psi$ such that $\mu$ is the compact-$\{L\to \pleft, R\to \pright\}$-restriction of $\psi$.
\end{definition}

\begin{definition}
    Let $\sigma = (D(N,T),\emb, \iota)$ and $\sigma_0 = (D(N_0,T_0), \emb_0, \iota_0)$ be $(S,\Y)$-containment structures.
    We say $\sigma_0$ is a \emph{subdivision} of $\sigma$ if $\sigma_0$ can be derived from $\sigma$ as follows:
    \begin{itemize}
        \item For each network arc $uv$ in $D(N,T)$ with $\iota(u)=\iota(v)=y\in Y$, replace $uv$ with a path $u_1 = u, u_2, \dots, u_j=v$ for some $j \geq 2$,
        where $u_i$ is a new vertex for each $1<i<j$.
        \item For any new path $u_1=u,\dots, u_j=v$ created this way, if $uv$ is part of the path $\emb(u'v')$ for a tree arc $u'v'$ then replace $uv$ in this path with $u_1,\dots u_j$. 
        \item For any new path $u_1=u,\dots, u_j=v$ created this way, set $\iota(u_i) = \iota(u)$ for each $1<i<j$.
    \end{itemize}
\end{definition}

Observe that $\sigma = c(\sigma_0)$ if and only if $\sigma$ is compact and $\sigma_0$ is a subdivision of $\sigma$.

Observe also that for any $\sigma, \sigma_0$ such that $\sigma_0$ is a subdivision of $\sigma$, $c(\sigma) = c(\sigma_0)$.

\begin{lemma}\label{lem:restrictionPreservesSubdivision}
	Let $\sigma_0 = (D(N_0,T_0), \emb_0, \iota_0)$ be a subdivision of $\sigma = (D(N,T), \emb, \iota)$,
	and for each arc $uv$ in $A(D(N,T))$ let $P_{uv}$ denote the set of arcs in $D(N_0, T_0)$ on the path from $u$ to $v$ corresponding to the subdivision of $uv$. (Taking $P_{uv} = \{uv\}$ if $uv$ was not subdivided.)
	
	Let $\sigma_0' = (D(N_0', T_0'), \emb_0', \iota_0')$ be the $g$-restriction of $\sigma_0$ and let $\sigma' = (D(N',T'), \emb', \iota')$ be the $g$-restriction of $\sigma$, for some restriction function $g$.
	
	Then $A(D(N_0', T_0')) = \bigcup_{uv \in A(D(N',T'))}P_{uv}$. 
	Thus, $\sigma_0'$ is a subdivision of $\sigma'$.
\end{lemma}
\begin{proof}
	It is sufficient to show that for each $a \in A(D(N,T))$, any arc or internal vertex of $P_a$ in $\sigma_0$ is $y$-redundant w.r.t. $(D(N_0,T_0), \emb_0, \iota_0\circ g)$ if and only if $a$ is $y$-redundant w.r.t. $ (D(N,T), \emb, \iota \circ g)$
	
	This can be seen by construction of $\sigma_0$, in particular the fact that $P_a$ is part of the path $\emb_0(u'v')$ if and only if $a$ is part of the path $\emb(u'v')$, and the fact that $\iota_0(z) = \iota(u)=\iota(v)$ (and thus $g(\iota(u))) = g(\iota(u)) = g(\iota(v))$) for every vertex in $P_a$.
\end{proof}

\begin{lemma}\label{lem:compactRestriction}
	Let $\sigma_0$ be a $(S,\Y)$-containment structure, let $\sigma = c(\sigma_0)$, and let $\sigma_0'$ be the $g$-restriction of $\sigma_0$ for some restriction function $g:S \cup \Y \rightarrow S' \cup \Y$.
	Then $c(\sigma_0')$ is the compact-$g$-restriction of $\sigma$.
\end{lemma}
\begin{proof}
    Let $\sigma'$ be the $g$-restriction of $\sigma$.
    As $\sigma_0$ is a subdivision of $\sigma$ and $\sigma_0'$ is the $g$-restriction of $\sigma_0$,
    \cref{lem:restrictionPreservesSubdivision} implies that $\sigma_0'$ is a subdivision of $\sigma'$.
    It follows that $c(\sigma_0') = c(\sigma')$, which is the compact-$g$-restriction of $\sigma$ by construction.
\end{proof}

\begin{lemma}\label{lem:V-implies-CV}
    Let $\sigma$ be a well-behaved signature for a bag $(P,S,F)$.
	If $\sigma$ is valid then $c(\sigma)$ is compact-valid.
	
	Similarly if $\mu$ is a well-behaved reconciliation for a Join bag $(L\cup R, S, F)$ and $\mu$ is valid, then $c(\mu)$ is compact-valid.
\end{lemma}

\begin{proof}
    Suppose a signature $\sigma$ is valid and let $\psi$ be the $F$-partial solution such that $\sigma$ is the $\{P\to \past\}$-restriction of $\psi$.
    Then $c(\psi)$ is compact $F$-partial solution, and by \cref{lem:compactRestriction}, $c(\sigma)$ is the compact- $\{P\to \past\}$-restriction of $c(\psi)$. Thus $c(\psi)$ is compact-valid, as required.
    
    For the case that a reconciliation $\mu$ is valid, a similar argument holds, using the restriction function $\{L\to \pleft, R \to \pright\}$ instead of $\{P\to \past\}$.
\end{proof}

\begin{lemma}\label{lem:CV-implies-exists-V}
    Let $\sigma$ be a compact signature for a bag $(P,S,F)$.
	If $\sigma$ is compact-valid then there is a valid signature $\sigma_0$ for $(P,S,F)$ such that $\sigma = c(\sigma_0)$ is compact-valid.
	
	Similarly if $\mu$ is a compact reconciliation for a Join bag $(L\cup R, S, F)$ and $\mu$ is compact-valid, then there is a valid reconciliation $\mu_0$ for$(L\cup R, S, F)$ such that $\mu = c(\mu_0)$.
\end{lemma}
\begin{proof}
   Suppose a compact signature $\sigma$ is compact-valid and let $\psi$ be the compact $F$-partial solution such that $\sigma$ is the compact-$\{P\to \past\}$-restriction of $\psi$.
   Let $\sigma_0$ be the $\{P\to \past\}$-restriction of $\psi$. Then $\sigma_0$ is a valid signature for $(P,S,F)$, and by construction $\sigma = c(\sigma_0)$.
   
    For the case that a compact reconciliation $\mu$ is compact-valid, a similar argument holds.
\end{proof}

\begin{lemma}\label{lem:longPathLabels}
Let $\sigma'$ be the $(S_1\to y_1, \dots, S_j \to y_j)$-restriction of some $(S,\Y)$-containment structure $\sigma$.
If $\sigma$ is compact, then $\sigma'$ has a long $y$-path only if $y \in \{y_1, \dots, y_j$\}.
\end{lemma}
\begin{proof}
Let $\sigma = (D(N,T), \emb, \iota)$ and $\sigma' = D(N',T'), \emb', \iota')$.
    Suppose for a contradiction that $\sigma$ has a long $y$-path for some $y \notin \{y_1, \dots, y_j\}$.
    Thus in particular, there is a path $x_1,x_2,x_3$ in $N'$ with $x_2$ having in-degree and out-degree $1$, and $\iota'(x_1) = \iota(x_2) = \iota'(x_3) = y$.
    Then by construction, $x_1,x_2,x_3$ is also a path in $N$, and $\iota(x_1) = \iota(x_2) = \iota(x_3) = y$ (since the restriction function  $(S_1\to y_1, \dots, S_j \to y_j)$ does not assign $y$ to any vertex that was not already labelled $y$ by $\iota$).
    Furthermore $x_2$ cannot have any other incident arcs beside $x_1x_2$ and $x_2x_3$, as such arcs would not become redundant after applying $(S_1\to y_1, \dots, S_j \to y_j)$, and so such arcs would be in $N'$ as well.
    It follows that $\sigma$ has a long $y$-path, contradicting the assumption that $\sigma$ is compact.
\end{proof}

\begin{corollary}\label{cor:top-level-compact-sig}
$(\Nin,\Tin)$ is a \textsc{Yes}-instance of {\sc Tree Containment}
if and only if
there is a compact-valid signature $\sigma = (D(N,T), \emb, \iota)$ for $(V(\NinTin), \emptyset, \emptyset)$  with ${\iota}^{-1}(\future) = \emptyset$.
\end{corollary}
\begin{proof}
 By \cref{cor:top-level-sig}, we have that $(\Nin,\Tin)$ is a \textsc{Yes}-instance of {\sc Tree Containment}
if and only if
there is a valid signature $\sigma_0 = (D(N_0,T_0), \emb_0, \iota_0)$ for $(V(\NinTin), \emptyset, \emptyset)$  with ${\iota_0}^{-1}(\future) = \emptyset$.

So suppose that such a signature $\sigma_0$ exists.
Let $\sigma = c(\sigma_0)$, and observe that $\sigma$ also satisfies $\iota^{-1}(\future) = \emptyset$. Furthermore by \cref{lem:V-implies-CV}, $\sigma$ is compact-valid.

Conversely suppose there is a compact-valid signature $\sigma = (D(N,T), \emb, \iota)$ for $(V(\NinTin), \emptyset, \emptyset)$  with ${\iota}^{-1}(\future) = \emptyset$.
Then by \cref{lem:CV-implies-exists-V} there is a valid signature $\sigma_0$ for $(V(\NinTin), \emptyset, \emptyset)$ with $\sigma = c(\sigma_0)$. Then as $\sigma_0$ is a subdivision of $\sigma$, we also have ${\iota_0}^{-1}(\future) = \emptyset$. Thus by \cref{cor:top-level-sig}, $(\Nin,\Tin)$ is a \textsc{Yes}-instance of {\sc Tree Containment}.
\end{proof}

We are now ready to prove the compact equivalents of the main lemmas for Forget, Introduce and Join bags

\begin{lemma}[Leaf bag]\label{lem:compactLeafBag}
Let $(P,S,F)$ correspond to a Leaf bag in the tree decomposition i.e. $P =  S = \emptyset, F = V(\NinTin)$ and $(P,S,F)$ has  no children.
Then a compact signature $\sigma = (D(N,T), \emb, \iota)$ for $(P,S,F)$ is compact-valid if and only if $\iota^{-1}(\past) = \emptyset$.
\end{lemma}
\begin{proof}
Similar to the proof of \cref{lem:leafBag}, we have that if $\psi = (D(N',T'), \emb', \iota')$ is a compact $F$-partial solution and $\sigma$ is the compact-$\{P\to \past\}$-restriction of $\psi$, then because $P = \emptyset$ we must have $\iota(u) \neq \past$ for all $u \in V(D(N,T))$.
Conversely if $\iota^{-1}(\past) = \emptyset$ then $\sigma$ is a compact $F$-partial solution and also the compact-$\{P\to \past\}$-restriction of itself.
\end{proof}

\begin{lemma}[Forget bag]\label{lem:compactForgetBag}
Let $(P,S,F)$ correspond to a Forget bag in the tree decomposition with child bag $(P',S',F')$, i.e. $P = P' \cup \{z\}$, $S = S'\setminus \{z\}$ and $F = F'$.

Then a compact signature $\sigma$ for $(P,S,F)$ is compact-valid if and only if there is a compact-valid signature $\sigma'$ for $(P',S',F')$ such that $\sigma$ is the compact-$\{\{z\}\to \past\}$-restriction of $\sigma'$. 
\end{lemma}
c.f. \cref{lem:forgetBag}

\begin{proof}
Suppose first that $\sigma$ is compact-valid. Then exists a valid signature $\sigma_0$ for $(P,S,F)$ with $c(\sigma_0)=\sigma$ (\cref{lem:CV-implies-exists-V}). Then by \cref{lem:forgetBag}, there is a valid signature $\sigma_0'$ for $(P',S',F')$ such that $\sigma_0$ is the $\{\{z\} \to \past\}$-restriction of $\sigma_0'$. Then let $\sigma' = c(\sigma_0')$ and observe that $\sigma'$ is compact-valid (\cref{lem:V-implies-CV}). Furthermore by \cref{lem:compactRestriction}, $\sigma$ is the compact-$\{\{z\} \to \past\}$-restriction of $\sigma'$, as required.

Conversely, suppose there is a compact-valid signature $\sigma'$ for $(P',S',F')$ such that $\sigma$ is the compact-$\{\{z\}\to \past\}$-restriction of $\sigma'$. Then there is a valid signature $\sigma_0'$ for $(P',S',F')$ with $\sigma'=c(\sigma_0')$  (\cref{lem:CV-implies-exists-V}).
Let $\sigma_0$ be the $\{\{z\}\to \past\}$-restriction of $\sigma_0'$. Then by \cref{lem:forgetBag}, $\sigma_0$ is valid. Furthermore by \cref{lem:compactRestriction}, $c(\sigma_0)$ is the compact-$\{\{z\}\to \past\}$-restriction of $c(\sigma_0') = \sigma'$. That is $c(\sigma_0) = \sigma$. Then as $\sigma_0$ is valid, $\sigma$ is compact-valid (\cref{lem:V-implies-CV}).
\end{proof}

\begin{lemma}\label{lem:compactIntroduceBag}
Let $(P,S,F)$ correspond to an Introduce bag in the tree decomposition with child bag $(P',S',F')$, i.e. $P' = P$, $S' = S\setminus \{z\}$ and $F' = F \cup \{z\}$.

Then a compact signature $\sigma$ for $(P,S,F)$ is valid if and only if $\sigma'$ is a compact-valid signature for $(P',S',F')$, where $\sigma'$ is the compact-$\{\{z\}\to \future\}$-restriction of $\sigma$.
\end{lemma}
c.f. \cref{lem:introduceBag}

\begin{proof}
Suppose first that $\sigma$ is compact-valid.
Then there is a valid signature $\sigma_0$ for $(P,S,F)$ such that $c(\sigma_0) = \sigma$ (\cref{lem:CV-implies-exists-V}).
Then let $\sigma_0'$ be the $\{\{z\}\to \future\}$-restriction of $\sigma_0$, and observe that $\sigma_0'$ is valid by  \cref{lem:introduceBag}.
Then $c(\sigma_0')$ is compact-valid (\cref{lem:V-implies-CV}).
Furthermore by \cref{lem:compactRestriction}, $c(\sigma_0')$ is the compact-$\{\{z\}\to \future\}$-restriction of $c(\sigma_0) = \sigma$.
Then letting $\sigma' = c(\sigma_0')$ we have that $\sigma'$ is compact-valid and the compact-$\{\{z\}\to \future\}$-restriction  restriction of $\sigma$, as required.

For the converse, 
assume $\sigma' = (D(N',T), \emb', \iota')$ is compact-valid, and
let $\psi'$ be the compact $F'$-partial solution such that $\sigma'$ is the compact-$\{P\to \past\}$-restriction of $\psi'$.
	Let $\sigma_0'$ be the $\{P'\to \past\}$-restriction of $\psi'$, so $\sigma_0'$ is a subdivision of $\sigma'$.
	In addition let $\sigma''$ be the $\{\{z\} \to \future\}$-restriction of $\sigma$, so $\sigma''$ is also a subdivision of $\sigma'$.
	Let $A_p$ be the subset of arcs in $N'$ that are subdivided by one or more additional vertices to produce $\sigma_0'$.
	Similarly let $A_f$ be the subset of arcs in $N'$ that are subdivided by one or more additional vertices to produce $\sigma''$.
	
	By  \cref{lem:longPathLabels},  
	$\sigma''$ has no long $\past$-paths and $\sigma_0'$ has no long $\future$-paths.
	Thus the sets of arcs $A_p$ and $A_f$ are disjoint (as $\iota'(u)=\iota'(v)=\past$ for any $uv \in A_p$, and $\iota'(u)=\iota'(v)=\future$ for any $uv \in A_f$).
	So now for each arc $uv$ in $N'$, define $P_{uv}$ to be the path corresponding to $uv$ in $\sigma_0'$ if $uv \in A_p$, let $P_{uv}$ be the path corresponding to $uv$ in $\sigma''$ if $uv \in A_f$, and let $P_{uv}$ be the single arc $uv$ otherwise.
	
	Now define $\sigma^*$ to be the subdivision of $\sigma'$ derived by replacing every arc $uv$ in $N$ by $P_{uv}$.
	Now we have that $\sigma^*$ is also a subdivision of $\sigma_0'$ (by subdividing arcs of $A_p$) and also of $\sigma''$ (by subdividing arcs of $A_f$, and also that $\sigma' = c(\sigma^*$) (as $\sigma'$ is compact).
	(See \cref{fig:CompactIntroduceBag}.)

  \begin{figure}[t]
     \centering
\includegraphics[scale = 0.75]{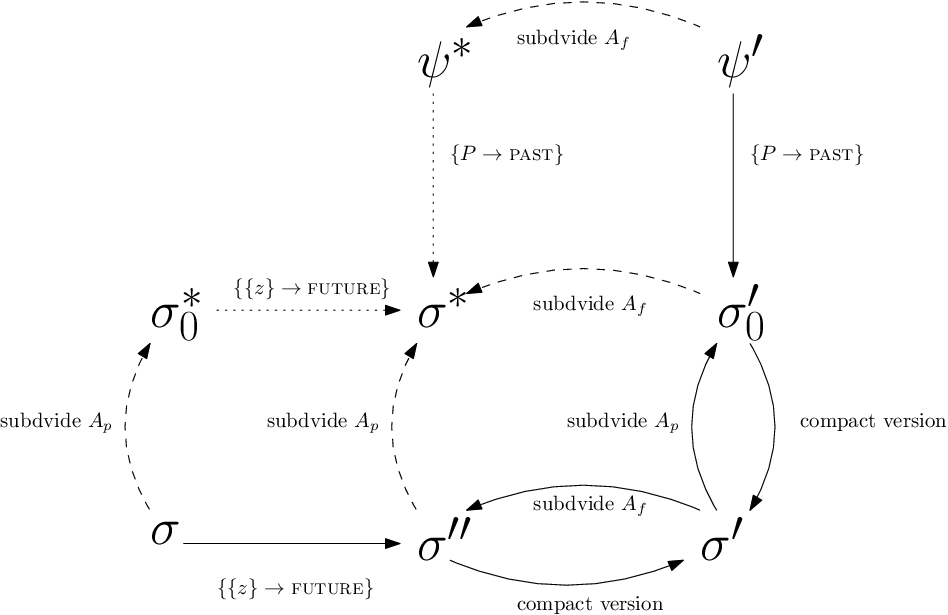}
     \caption{Illustration of proof for Introduce bags, compact-validity of $\sigma'$ implies compact-validity of $\sigma'$. Solid lines are  relations that we may assume; the dashed line shows the construction of $\sigma^*,\psi^*$ and $\sigma_0'$; dotted lines are restriction relations we can show using \cref{lem:restrictionPreservesSubdivision}.}
     \label{fig:CompactIntroduceBag}
 \end{figure}
	
	We now claim that $\sigma^*$ is a valid signature for $(P',S',F')$.
	To see this,
	Let $\psi^*$ be the subdivision of $\psi'$ derived by replacing $uv$ with $P_{uv}$ for any $uv \in A_f$ (observe that all arcs of $A_f$ appear in the display graph of $\psi$, as they are all arcs in $\sigma'$ and were not subdivided to create $\sigma_0'$).
	Then $\psi^*$ is also a well-behaved $F'$-partial solution. 
	Then using \cref{lem:restrictionPreservesSubdivision} and the fact that $\sigma_0'$ is the $\{P'\to \past\}$-restriction of $\psi'$, we have that the $\{P'\to \past\}$-restriction $\psi^*$ can be derived from $\sigma_0'$ by replacing the $uv$ with $P_{uv}$ for all $uv \in A_f$.
	That is, the $\{P'\to \past\}$-restriction $\psi^*$ is exactly $\sigma^*$. Thus $\sigma^*$ is indeed valid.
	
    Next, let $\sigma_0^*$ be the subdivision of $\sigma$ derived by replacing $uv$ with $P_{uv}$ for any $uv \in A_p$.
    Similar to the case with $\psi^*$, we can show using \cref{lem:restrictionPreservesSubdivision} that the $\{\{z\}\to \future\}$-restriction of $\sigma_0^*$ is $\sigma''$ with $uv$ replaced by $P_{uv}$ for all $uv \in A_p$, i.e. $\sigma^*$.
    So we now have that $\sigma^*$ is the $\{\{z\}\to \future\}$-restriction of $\sigma_0^*$ and also a valid signature for $(P',S',F')$. By \cref{lem:introduceBag}, this implies that $\sigma_0^*$ is valid. 
    Then $\sigma$, which is the compact version of $\sigma_0^*$, is compact-valid by \cref{lem:V-implies-CV}.
\end{proof}

\begin{lemma}\label{lem:compactJoinBag}
  Let $(L\cup R, S, F)$ be a Join bag with child bags $(L, S, F\cup R)$ and $(R,S,F \cup L)$ and
  let $\sigma$ be a compact signature for $(L\cup R, S, F)$.
  Then, $\sigma$ is compact-valid
  if and only if
  there is a compact-valid reconciliation $\mu$ for $(L\cup R, S, F)$
  such that $\sigma$ is the compact-$\{\{\pleft, \pright\} \to \past\}$-restriction of $\mu$.
\end{lemma}
c.f. \cref{lem:joinBag}.

\begin{proof}
 Suppose first that $\sigma$ is compact-valid. Then there is a valid signature $\sigma_0$ for $(L\cup R, S, F)$ such that $\sigma = c(\sigma_0)$ (\cref{lem:CV-implies-exists-V}.)
 By \cref{lem:joinBag}, there is a valid reconciliation $\mu_0$ for $(L \cup R, S, F)$ such that $\sigma_0$ is the $\{\{\pleft, \pright\} \to \past\}$-restriction of $\mu_0$.
 Then let $\mu = c(\mu_0)$, so $\mu$ is a compact-valid reconciliation (\cref{lem:V-implies-CV}). By \cref{lem:compactRestriction}, the compact-$\{\{\pleft, \pright\} \to \past\}$-restriction of $\mu = c(\mu_0)$ is $c(\sigma_0) = \sigma$, as required.
 
 Conversely, suppose that $\mu$ is a compact reconciliation for $(L\cup R, S, F)$ such that $\sigma$ is the compact-$\{\{\pleft, \pright\} \to \past\}$-restriction of $\mu$. As $\mu$ is compact-valid, there is a valid reconciliation $\mu_0$ for $(L \cup R, S, F)$ such that $\mu = c(\mu_0)$ (\cref{lem:CV-implies-exists-V}). Let $\sigma_0$ be the $\{\{\pleft, \pright\} \to \past\}$-restriction of $\mu$. Then by \cref{lem:joinBag}, $\sigma_0$ is a valid signature for $(L\cup R, S, F)$. Moreover by \cref{lem:compactRestriction}, $c(\sigma_0)$ is the compact-$\{\{\pleft, \pright\} \to \past\}$-restriction of $\mu = c(\mu_0)$, i.e. $c(\sigma_0) = \sigma$. Then as $\sigma_0$ is valid, $\sigma$ is compact-valid (\cref{lem:V-implies-CV}).
\end{proof}

\begin{lemma}\label{lem:compactReconciliationToChildren}
Let $(L\cup R, S, F)$ be a Join bag with child bags $(L, S, F\cup R)$ and $(R,S,F \cup L)$, and let $\mu$ be a compact reconciliation for $(L\cup R, S, F)$.
Let  $\sigma_L$ be the compact-$\{\pleft \to \past, \pright \to \future\}$-restriction of $\mu$, and $\sigma_R$ the compact-$\{\pright \to \past, \pleft \to \future\}$-restriction of $\mu$.
If $\mu$ is compact-valid, then $\sigma_L$ is a compact-valid signature for $(L, S, F \cup R)$ and 
$\sigma_R$  is a compact-valid signature for $(R, S, F \cup L)$.
\end{lemma}

c.f. \cref{lem:reconciliationToChildren}

\begin{proof}
	Suppose $\mu$ is compact-valid. 
	Then there is a valid reconciliation $\mu_0$ for $(L \cup R, S, F)$ such that $\mu = c(\mu_0)$ (\cref{lem:CV-implies-exists-V}).
	Now let $\sigma_L'$ be the $\{\pleft \to \past, \pright \to \future\}$-restriction of $\mu_0$. Then by \cref{lem:reconciliationToChildren}, $\sigma_L'$ is a valid signature for  $(L, S, F \cup R)$.
	Furthermore by \cref{lem:compactRestriction}, $c(\sigma_L')$ is the compact-$\{\pleft \to \past, \pright \to \future\}$-restriction of $c(\mu_0) = \mu$. That is, $c(\sigma_L) = \sigma_L$. It follows by \cref{lem:V-implies-CV} that $\sigma_L$ is compact-valid, as required.
	As similar argument shows that $\sigma_R$ is a compact-valid signature for  $(R, S, F \cup L)$.
\end{proof}

\begin{lemma}\label{lem:compactChildrentoReconciliation}
Let $(L\cup R, S, F)$ be a Join bag with child bags $(L, S, F\cup R)$ and $(R,S,F \cup L)$, and let $\mu$ be a compact reconciliation for $(L\cup R, S, F)$.
Let  $\sigma_L$ be the compact-$\{\pleft \to \past, \pright \to \future\}$-restriction of $\mu$, and $\sigma_R$ the compact-$\{\pright \to \past, \pleft \to \future\}$-restriction of $\mu$.
If  $\sigma_L$ is a compact-valid signature for $(L, S, F \cup R)$ and 
$\sigma_R$  is a compact-valid signature for $(R, S, F \cup L)$, then $\mu$ is compact-valid.
\end{lemma}
c.f.\ \cref{lem:childrentoReconciliation}

\begin{proof}
  In what follows, let $\sigma = (\Disp{\sigma}, \emb_{\sigma}, \iota_{\sigma})$, for any given containment structure $\sigma$.
  
  The proof is along similar lines to \cref{lem:compactIntroduceBag}, but with more containment structures. Before we can proceed, we first show that the the $\{\pleft \to \past, \pright \to \future\}$-restriction of $\mu$ is in fact compact, and therefore $\sigma_L$ is the $\{\pleft \to \past, \pright \to \future\}$-restriction of $\mu$, as well as the compact-$\{\pleft \to \past, \pright \to \future\}$-restriction of $\mu$. 
  
  Indeed, let $\sigma_L''$ denote the $\{\pleft \to \past, \pright \to \future\}$-restriction of $\mu$, and suppose for a contradiction that $\sigma_L''$ contains a long $y$-path for some $y \in \{\past, \future\}$. Thus there is a path $x_1,x_2,x_3$ in $N_{\sigma_L''}$ where $x_2$ has in-degree and out-degree $1$ and $\iota_{\sigma_L''}(x_1)= \iota_{\sigma_L''}(x_2) = \iota_{\sigma_L''}(x_3) = y$.
  If $y = \past$ then  $\iota_{\mu}(x_1)= \iota_{\mu}(x_2) = \iota_{\mu}(x_3) = \pleft$.
  Otherwise $y = \future$, and $\iota_{\mu}(x_1), \iota_{\mu}(x_2), \iota_{\mu}(x_3)$ are all in $\{\pright, \future\}$. But note that all three of $\iota_{\mu}(x_1), \iota_{\mu}(x_2), \iota_{\mu}(x_3)$ must be the same value, otherwise $\mu$ has an arc $uv$ with $\iota_{\mu}(\{u,v\}) = \{\pright, \future\}$ and $\mu$ is not well-behaved. 
 So if $x_2$ has in-degree and out-degree $1$ in $N_{\mu}$, then $\mu$ has a long $y$-path for $y=\iota_{\mu}(u)$, a contradiction as $\mu$ is compact.
 Otherwise, suppose $x_2$ has an incident arc $x_2z$ for $z\neq x_3$ (the case of an incident arc $zx_2$ is similar). As $x_2z$ is not an arc in $N_{\sigma_L''}$, it must hold that $\iota_{\mu}(z)=\iota_{\mu}(x_2)$ (otherwise the arc would not become redundant or $\mu$ is not well-behaved). Then since $x_2z$ is not redundant w.r.t $\mu$, there is some tree arc $uv$ in $T_{\mu}$ such that $x_2z$ is an arc of $\emb_{\mu}(uv)$. But this implies also that $x_1x_2$ is also arc of $\emb_{\mu}(uv)$, as $x_2 \neq \iota_{\mu}(u)$. Thus $x_2z$ is deleted in the construction of $\sigma_L''$ only if $uv$ is, which would in turn imply that $x_1x_2$ is deleted, again a contradiction.
  
  So we may assume that $\sigma_L''$ is compact and thus $\sigma_L'' = c(\sigma_L'') = \sigma_L$, so $\sigma_L$ is the $\{\pleft \to \past, \pright \to \future\}$-restriction of $\mu$.
  A similar argument shows that $\sigma_R$ is the $\{\pright \to \past, \pleft \to \future\}$-restriction of $\mu$.

  So now assume that $\sigma_L$ and $\sigma_R$ are both  compact-valid.
  Let $\psi_L$ denote the compact $F\cup R$-partial solution for which $\sigma_L$ is the compact $\{L \to \past\}$-restriction, and let $\sigma_L'$ be the $\{L \to \past\}$-restriction of $\psi_L$, so that $\sigma_L = c(\sigma_L')$.
  Similarly let  $\psi_R$ denote the compact $F\cup L$-partial solution for which $\sigma_R$ is the compact $\{R \to \past\}$-restriction, and let $\sigma_R'$ be the $\{R \to \past\}$-restriction of $\psi_R$, so that $\sigma_R = c(\sigma_R')$.
 (See \cref{fig:CompactReconciliation})

 \begin{figure}[t]
     \centering
\includegraphics[scale = 0.8]{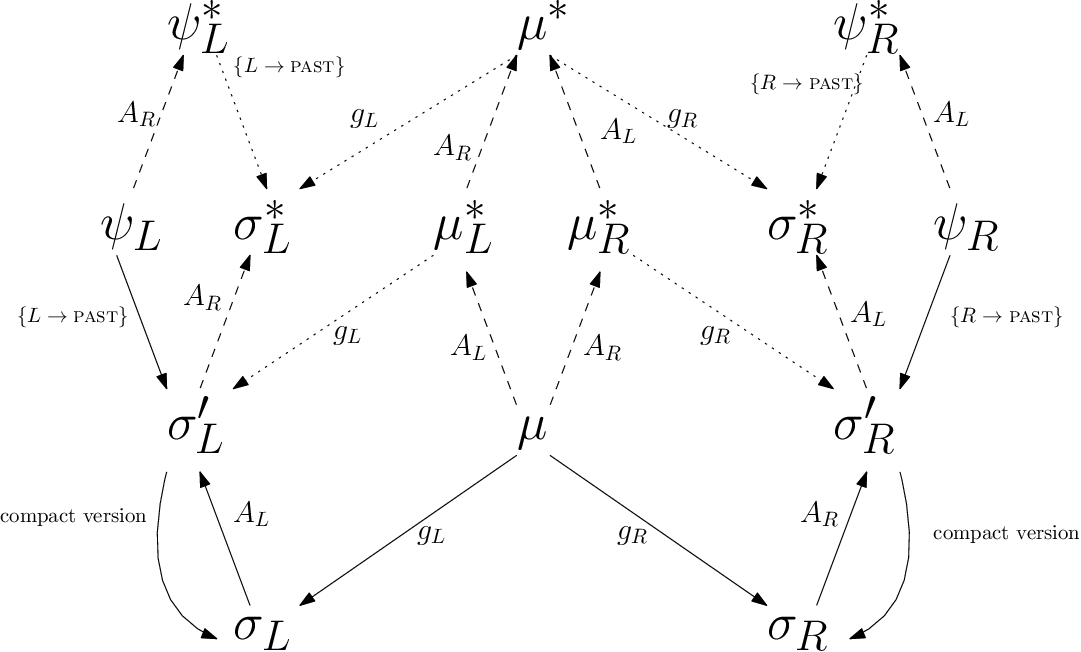}
     \caption{Illustration of proof of \cref{lem:compactChildrentoReconciliation}. Solid lines are  relations that we may assume; the dashed lines show the construction of $\sigma_L^*,\sigma_R^*,\psi_L^*,\psi_R^*,\mu_L^*, \mu_R^*, \mu^*$; dotted lines are restriction relations we can show using \cref{lem:restrictionPreservesSubdivision}.
     Here $g_L$ denotes the restriction function $\{\pleft \to \past, \pright \to \future\}$ and $g_R$ denotes $\{\pright \to \past, \pleft \to \future\}$; labels $A_L$ or $A_R$ indicate a subdivision of those arcs}
     \label{fig:CompactReconciliation}
 \end{figure}
 
 By \cref{lem:longPathLabels}, the only long-$y$-paths in $\sigma_L'$ are for $y = \past$.
 So $\sigma_L'$ is a subdivision of $\sigma_L$ in which the only subdivided arc are those $uv$ for which $\iota_{\sigma_L}(u) = \iota_{\sigma_L}(v) = \past$. So let $A_L$ the set of these arcs, and for each $uv \in A_L$ let $P_{uv}$ be the corresponding path in $N_{\sigma_L'}$.
 Similarly let $A_R$ be the set of arcs in $N_{\sigma_R}$ that are subdivided to produce $\sigma_R'$ from $\sigma_R$ (noting that $\iota_{\sigma_R}(u) = \iota_{\sigma_R}(v) = \past$ for any $uv \in A_R$), and let $P_{uv}$ be the corresponding path in $N_{\sigma_R'}$ for each $uv \in A_R$.
 Observe that since $\iota_{\sigma_L}(u) = \past$ implies $\iota_{\mu}(u) = \pleft$ and $\iota_{\sigma_R}(u) = \past$ implies $\iota_{\mu}(u) = \pright$, the arc sets $A_L$ and $A_R$ are disjoint.
 
 Now let $\sigma_L^*$ be the subdivision of $\sigma_L'$ (not $\sigma_L$) by replacing $uv$ with $P_{uv}$ for all $uv \in A_R\cap A(N_{\sigma_L'})$. Equivalently, $\sigma_L^*$ is the subdivision of $\sigma_L$ derived by replacing $uv$ with $P_{uv}$ for all arcs $uv$ in $A_L \cup (A_R \cap A(N_{\sigma_L}))$.
 Thus $c(\sigma_L^*) = \sigma_L$.
 Define $\sigma_R^*$ analogously.
 
 We claim that $\sigma_L^*$ is a valid signature for $(L, S, R\cup F)$.
 To see this, let $\psi_L^*$ be the subdivision of $\psi_L$ derived by replacing $uv$ with $P_{uv}$ for every 
 $uv$ in $A_L \cup (A_R \cap A(N_{\psi_L}))$.
 Then by \cref{lem:restrictionPreservesSubdivision} and the fact that $\sigma_L'$ is  the $\{L \to \past\}$-restriction
 of $\psi_L$, we have that $\sigma_L^*$ is  the $\{L \to \past\}$-restriction of $\psi_L^*$.
 Thus $\psi_L^*$ is valid, and a similar argument shows that $\psi_R^*$ is valid.
 
 Now let $\mu^*$ (respectively $\mu_L^*, \mu_R^*$) be the subdivision of $\mu$ derived by replacing $uv$ with $P_{uv}$ for all $uv$ in $A_L\cup A_R$ (respectively $A_L$, $A_R$).
 Note that $\mu^*$ is also a subdivision of both $\mu_L^*$ and $\mu_R^*$.
 
 Using a similar approach to before, by \cref{lem:restrictionPreservesSubdivision} and the fact that $\sigma_L$ is the $\{\pleft\to \past, \pright \to \future\}$-restriction of $\mu$, we can show that $\sigma_L'$ is the $\{\pleft\to \past, \pright \to \future\}$-restriction of $\mu_L^*$.
 Using this fact in turn together with \cref{lem:restrictionPreservesSubdivision}, we can show that $\sigma_L^*$ is the $\{\pleft\to \past, \pright \to \future\}$-restriction of $\mu^*$.
 Similarly we can show that $\sigma_R^*$ is the $\{\pright\to \past, \pleft \to \future\}$-restriction of $\mu^*$.
 
If  $\sigma_L^*$ is a valid signature for $(L, S, F \cup R)$ which is the $\{\pleft\to \past, \pright \to \future\}$-restriction of $\mu^*$,
and 
$\sigma_R^*$  is a valid signature for $(R, S, F \cup L)$ which is the $\{\pright\to \past, \pleft \to \future\}$-restriction of $\mu^*$, we can now apply \cref{lem:joinBag} to see that $\mu^*$ is a valid reconciliation for $(L\cup R, S, F)$.

It remains to observe that as $\mu^*$ is valid and $\mu = c(\mu^*)$, \cref{lem:V-implies-CV} implies that $\mu^*$ is compact-valid.
\end{proof}

From \cref{lem:compactJoinBag,lem:compactReconciliationToChildren,lem:compactChildrentoReconciliation} we have the following:
 
\begin{corollary}\label{cor:compactJoinBagOverall}
Let $(L\cup R, S, F)$ be a Join bag with child bags $(L, S, F\cup R)$ and $(R,S,F \cup L)$, and let $\sigma$ be a compact signature for $(L\cup R, S, F)$.
Then $\sigma$ is compact-valid if and only if there is a 
compact reconciliation $\mu$ for $(L\cup R, S, F)$, 
and compact-valid signatures $\sigma_L$ for $(L, S, F \cup R)$ and 
$\sigma_R$ for $(R, S, F \cup L)$, 
such that $\sigma$ is the compact-$\{\{\pleft, \pright\} \to \past\}$-restriction of $\mu$,
$\sigma_L$ is the compact-$\{\pleft \to \past, \pright \to \future\}$-restriction of $\mu$, and 
$\sigma_R$ is the compact-$\{\pright \to \past, \pleft \to \future\}$-restriction of $\mu$.
\end{corollary}

\section{Algorithm and running time}\label{sec:algorithm}

\cref{alg:treeContainment} gives a summary of our algorithm for \textsc{Tree Containment}.
\mjrev{To summarise:} we compute for each bag~$x = (P,S,F)$ in the tree decomposition, a set $CV_x$ of \mjnew{compact-valid signatures} for~$x$ \mjrev{- that is, compact signatures for which there exists a corresponding $F$-partial solution.}  
The algorithm uses \leojournal{compact-restrictions} to convert a compact signature of one bag into a compact signature for a different bag. \leojournal{Recall that} such a restriction works by mapping certain vertices to a label $\past$ or $\future$, removing redundant parts of the display graph and collapsing long paths, similarly to the process described in Section~\ref{sec:signatureOverview}. See Sections~\ref{sec:restriction} \mjrev{and~\ref{sec:compact}} for the formal \mjrev{definitions}. For join bags, the algorithm uses \leojournal{reconciliations},  3-way analogues of signatures, using labels $\{\pleft,\pright,\future\}$ instead of $\{\past, \future\}$, see Section~\ref{sec:joinBags}.

The correctness of the computation of the sets~$CV_x$ follows from 
\mjrev{\cref{lem:compactLeafBag,lem:compactForgetBag,lem:compactIntroduceBag,cor:compactJoinBagOverall}.}
The correctness of the last three lines, in which we return \textsc{true} if and only if there is a compact-valid signature $(D(N,T),\emb, \iota)$ for the root bag with $\iota^{-1}(\future) = \emptyset$, is a consequence of 
\mjrev{\cref{cor:top-level-compact-sig}.}

To show that the running time is bounded in a function in the treewidth of $N$, the main challenge is to
bound the number of compact signatures for a bag~$(P,S,F)$ by a function of $|S|$
(which, by \cref{thm:tw bound}, we may assume is at most $2tw(N)+1$).
In order to do this, we first bound the size of the display graph $D(N,T)$ in a signature by a function of $|S|$.
\mjrev{We will then use this to bound}
the number of possible display graphs, embedding functions and isolabellings, \mjrev{and thus the number of compact signatures}.

\begin{algorithm}[t!] 
  Compute $tw(\Nin)$\;
  \lIf{$tw(\NinTin) > 2 tw(\Nin)+1$}{\Return{\false}}
  Compute a nice minimum-width tree decomposition $\T$ of $\NinTin$\;
  
  \ForEach{bag $x=(P,S,F)$ of $\T$ in a bottom-up traversal}{
    \If{$x$ is a Leaf bag}{
      $CV_x\gets$ set of all compact signatures  $\sigma = (D(N,T), \emb, \iota)$  for $x$ with $\iota^{-1}(\textsc{past}) = \emptyset$\;
    }\ElseIf{$x$ is a Forget bag with child $y=(P\setminus\{z\},S\cup\{z\},F)$ in $\T$}{
          \ForEach{$\sigma \in CV_y$}{
            add the compact-$\{z\to\past\}$-restriction of $\sigma$ to $CV_x$
          }
    }\ElseIf{$x$ is an Introduce bag with child $y=(P,S\setminus\{z\},F\cup\{z\})$ in $\T$}{
      \ForEach{compact signature $\sigma$ of $x$}{
        \If{$CV_y$ contains the compact-$\{z\to\future\}$-restriction of $\sigma$}{
          add $\sigma$ to $CV_x$\;
        }
      }
    }\ElseIf{$x$ is a Join bag with children $y_L=(L,S,R\cup F)$ {\rm\&} $y_R=(R,S,L\cup F)$}{
      \ForEach{compact reconciliation $\mu$ for $x$}{
        $\sigma_L\gets\text{the compact-$\{\textsc{left}\to\past,\textsc{right}\to\future\}$-restriction of $\mu$}$\;
        $\sigma_R\gets\text{the compact-$\{\textsc{right}\to\past,\textsc{left}\to\future\}$-restriction of $\mu$}$\;
        \If{$\sigma_L\in CV_{y_L}$ and $\sigma_R\in CV_{y_R}$}{
          add the compact-$\{\{\textsc{left},\textsc{right}\}\to\past\}$-restriction of $\mu$ to $CV_x$\;
        }
      }
    }
  }
  \ForEach{$(D(N,T),\phi,\iota)\in CV_{\operatorname{root}(\T)}$}{
    \lIf{$\iota^{-1}(\future)=\emptyset$}{\Return{\true}}
  }
  \Return{\false}
  \caption{Tree Containment ($\NinTin$)}\label{alg:treeContainment}
\end{algorithm}

\begin{lemma}\label{lem:compactSignatureSize}
Any compact signature $\sigma = (D(N,T), \emb, \iota)$ for a bag $(P, S, F)$ satisfies\\ ${|V(D(N,T)))| \in O(|S|)}$.
Any compact reconciliation $\mu = (D(N,T), \emb, \iota)$ satisfies\\ $|V(D(N,T)))| \in O(|S|)$.
\end{lemma}

\begin{proof}
Let $\sigma = (D(N,T), \emb, \iota)$ be a compact signature for $(P,S,F)$.

We first bound the number of arcs in $D(N,T)$. To do this, let $A_S$ denote the subset of arcs in $D(N,T)$ incident to a vertex in $\iota^{-1}(S)$.
As there is only one vertex $u$ with $\iota(u) = s$ for each $s\in S$ and every vertex in $D(N,T)$ has total degree at most $3$, we have that $|A_S| \leq 3|S|$.
As $\emb(uv)$ and $\emb(u'v')$ are arc-disjoint for any distinct tree arcs $uv, u'v'$, it follows that there are at most $|A_S\cap A(N)|$ arcs $uv$ of $T$ for which $\emb(uv)$ contains an arc from $A_S$. There are at most $|A_S\cap A(T)|$ arcs in $T$ incident to a vertex from $\iota^{-1}(S)$. Thus there are at most $|A_S|$ arcs $uv$ in $T$
for which $\{u,v\}\cup V(\emb(uv))$ contains a vertex from $\iota^{-1}(S)$.

Every remaining arc $uv$ in $T$ has 
 $\iota(\{u,v\}\cup V(\emb(uv))) \subseteq \{\past, \future\}$.
 Furthermore as $\sigma$ is well-behaved, we must have that $\iota(u) = \iota(v)$, and $\iota(u')=\iota(v')$ for every arc $u'v'$ in the path $\emb(uv)$.
It follows that  for all but at most $|A_S|\leq 3|S|$ arcs of $T$, we have $\iota(u) = \iota(v) \in \{\past, \future\}$ and $\iota(V(\emb(uv))) = \{\past\}$ or $\iota(V(\emb(uv))) = \{\future\}$.

If $\iota(u) = \iota(v) =  \past$ and $\iota(V(\emb(uv))) = \{\past\}$, then $uv$ is redundant w.r.t $\{\past\}$, a contradiction as we may assume no valid signature has a $\{\past\}$-redundant arc.
Similarly we have a contradiction if $\iota(u) = \iota(v) =  \future$ and $\iota(V(\emb(uv))) = \{\future\}$.
So it must be the case that for all but at most $3|S|$ tree arcs $uv$, either $\iota(u) = \iota(v) =  \past$ and $\iota(\emb(uv)) = \{\future\}$, or $\iota(u) = \iota(v) =  \future$ and $\iota(\emb(uv)) = \{\past\}$.
In particular, we have that $\iota(\emb(v)) \neq \iota(v)$ and, so $v$ has out-degree $2$

It follows any lowest arc $uv$ in $T$ is one of the at most $|A_S|$  arcs $\{u,v\}\cup V(\emb(uv))$ that contains a vertex from $\iota^{-1}(S)$. Thus in total $T$ has at most $2|A_S| \leq 6|S|$ arcs.

To bound the arcs of $N$, observe that any arc $uv \in A(N)$ not in $A_S$ satisfies $\iota(u)=\iota(v) = \past$ or $\iota(u)=\iota(v) = \future$ (it cannot be that $\{\iota(u), \iota(v)\}= \{\past, \future\}$ as $\sigma$ is well-behaved). Then $uv$ must be part of the path $\emb(u'v')$ for some tree arc $u'v'$ (otherwise $uv$ is redundant w.r.t $\sigma$, a contradiction).
So now let $A_N'$ denote the set of arcs in $N$ that are part of a path $\emb(uv)$ for some tree arc $uv$ (note that $A_N'$ and $A_S$ are not necessarily disjoint but $A(N) \subseteq A_N' \cup  A_S$).

For any internal vertex $z$ on a path $\emb(uv)$, we must have that $z$ is incident to an arc from $A_S$. Indeed suppose this is not the case, then $z$ and its neighbours in $\emb(uv)$ form a path $x_1,x_2,x_3$ with $\iota(x_1)=\iota(x_2)=\iota(x_3) = y\in \{\past, \future\}$. This forms a long-$y$-path (contradicting the fact that $\sigma$ is compact), unless $z=x_2$ is incident to another arc in $A(N)$. But such an arc cannot be in $A_S$ by assumption, and also cannot be part of a path $\emb(u'v')$ for any tree arc $u'v'$ (as $\emb(uv)$ and $\emb(u'v')$ share a vertex $z$ only if $uv, u'v'$ share a vertex $w$ with $\emb(w)=z$).  Thus there is no other arc incident to $z$, and we have that $\sigma$ is not compact, a contradiction.

Thus we now have that every internal vertex of a a path $\emb(uv)$ is incident to an arc from $A_S$, and thus there are at most $2|A_S| \leq 6|S|$ such vertices. As there at most $|A(T)| \leq 2|A_S|$ paths $\emb(uv)$, we have that $|A_N'| \leq 2|A_S| + 2|A_S| \leq 12|S|$.

Thus in total, the number of arcs in $D(N,T)$ is at most $|A(T)| + |A_S| + |A_N'| \leq 2|A_S| + |A_S| + 2 |A_S| = 5|A_S| \leq 15|S|$.
It follows that the number of non-isolated vertices in $D(N,T)$ is at most $30|S|$. It remains to bound the number of isolated vertices.
 
There are at most $|S|$ isolated vertices $u$ in $V(D(N,T))$ for which $\iota(u) \in S$.
For the rest,
If $u \in V(T)$ and $\emb(u) \in V(N)$ are both isolated vertices then we have $\iota(u) = \iota(\emb(u) \in \{\past, \future\}$ (as $\iota(u)\neq \iota(\emb(u))$ would imply $u$ has out-degree $2$ and so $u$ and $\emb(u)$ are both redundant w.r.t $\sigma$.
Similarly if $v$ is an isolated network vertex with no $u\in V(T)$ for which $\emb(u) = v$, then $v$ is redundant w.r.t $\sigma$. As $\sigma$ has no redundant vertices (since $\sigma$ is well-behaved), it follows that every isolated vertex in $D(N,T)$ is either a tree vertex $u$ with $\emb(u)$ not isolated, a network vertex $v = \emb(u)$ for which $u$ is not isolated, or a vertex of $\iota^{-1}(S)$. As there are at most $30|S|$ non-isolated vertices, it follows that there are at most $30|S| + |S|$ isolated vertices.

Thus in total, $|V(D(N,T))| \leq 30|S| + 30|S| + |S| = 61|S| \in O(|S|)$. 

An identical argument holds for a reconciliation $\mu$.
\end{proof}

\begin{lemma}\label{lem:numCompactSignatures}
  Let $k$ be the width of the tree decomposition of $\NinTin$.
  Then, the number of compact signatures for a bag $(P,S,F)$
  and the number of compact reconciliations for a Join bag
  can be upper-bounded by $2^{O(k^2)}$.
\end{lemma}
\begin{proof}
Let $\sigma = (D(N,T), \emb, \iota)$ denote a compact signature for $(P,S,F)$
(The arguments for a reconciliation $\mu$ are similar).
By \cref{lem:compactSignatureSize}, $|V(D(N,T))| \in O(|S|)$ and,
by the properties of a tree decomposition, $|S|\leq k+1$, implying $|V(D(N,T))| \in O(k)$.
As such, an upper bound for the number of possible graphs $D(N,T)$ is $2^{O(k^2)} =: f_1(k)$.

For each vertex $u \in V(D(N,T))$, there are at most $(|S|-1)+2 \leq k+2$ possibilities for $\iota(u)$,
as $\iota(u)$ is either a vertex in $S$ other than~$u$ or one of the labels $\past$ and $\future$ ($\pleft$, $\pright$, and $\future$ for a reconciliations).
Thus, the number of possible isolabelings for a given display graph $D(N,T)$ is
$(k+3)^{O(k)} =: f_2(k)$

Now, to bound the number of possible embedding functions~$\emb$, observe that $\emb$ is fixed by
\begin{inparaenum}[(a)]
  \item specifying $\emb(u)$ for every \emph{tree vertex}~$u$ ($|V(N)|^{|V(T)|}\in k^{O(k)}$~possibilities) and
  \item the set of arcs in $N$ that appear in some path $\emb(uv)$ ($2^{|A(N)|}\in 2^{O(k^2)}$~possibilities) --
    indeed, if this set of arcs is chosen correctly, then it contains only one path from $\emb(u)$ to $\emb(v)$, which must be the path $\emb(uv)$.
\end{inparaenum}
Thus, for any fixed display graph $D(N,T)$,
the number of possible embedding functions is upper-bounded by $2^{O(k^2)}=:f_3(k)$.

Now, for any bag $(P,S,F)$, the number of possible choices for $\sigma = (D(N,T), \emb, \iota)$ is bounded by
$f(k) := f_1(k)f_2(k)f_3(k) = 2^{O(k^2)}\cdot 2^{O(k\log k)}\cdot 2^{O(k^2)}=2^{O(k^2)}$.
\end{proof}

\begin{lemma}\label{lem:algAnalysis}
  \cref{alg:treeContainment} is correct and runs in $2^{O(k^2)}\cdot |A(\Nin)|$~time, where $k = tw(\Nin)$.
\end{lemma}
\begin{proof}
First, note that, by \cref{thm:tw bound}, $\Nin$ does not display $\Tin$ unless $tw(\NinTin) \leq 2k+1$,
so we are safe to return \textsc{false} if $tw(\NinTin) \leq 2tw(\Nin)+1$ in line 2 of \cref{alg:treeContainment}.
Otherwise, for each bag $x$ in a nice tree decomposition of $\NinTin$, \cref{alg:treeContainment} calculates the set $CV_x$ of compact-valid signatures for~$x$.
In each case the set $CV_x$ is calculated using the previously-calculated set $CV_y$ for each child $y$ of $x$.
The correctness of this construction follows from \cref{lem:compactLeafBag} (for the Leaf bags), \cref{lem:compactForgetBag} (for Forget bags), \cref{lem:compactIntroduceBag} (for Introduce bags), and \cref{cor:compactJoinBagOverall} (for Join bags). 
Finally, the algorithm returns \textsc{true} if and only if there is a valid compact signature $(D(N,T), \emb, \iota)$ for the root bag of the tree decomposition, such that $\iota^{-1}(\future) = \emptyset$. The correctness of this follows from \cref{cor:top-level-compact-sig}.

To see the running time, first note that
a nice, minimum-width tree decomposition of $\NinTin$ with $O(|V(\NinTin)|) = O(|A(\Nin)|)$ bags
can be found in $2^{O(tw(\NinTin)^2)}|A(\Nin)|$, that is, $2^{O(k^2)}|A(\Nin)|$~time~\citep{bodlaender1996linear,Kloks1994Treewidth}.
By \cref{thm:tw bound}, we may assume $\NinTin$ has treewidth at most $2k+1$ and, thus,
$|S|\leq 2k+1$ for every bag $(P,S,F)$ in the decomposition.
Note that, computing any compact restriction of a signature~$\sigma$ can be done in polynomial time and,
by \cref{lem:numCompactSignatures}, the number of such signatures~$|CV_x|$ for a bag~$x$ is bounded by $2^{O(k^2)}$.
It is, thus, evident that, for any bag~$x=(P,S,F)$, the set $CV_x$ can be computed in $2^{O(k^2)}\cdot k^{O(1)}$ =$2^{O(k^2)}$~time (see \cref{alg:treeContainment}).
\end{proof}

\cref{lem:algAnalysis} immediately implies the following theorem.

\begin{theorem}\label{thm:TC-fpt}
  \textsc{Tree Containment} can be solved in $2^{O(tw(\Nin)^2)}\cdot|A(\Nin)|$~time.
\end{theorem}

\section{Future work}

Before implementing our dynamic programming algorithm, one should first try to reduce the constant in the bound on the number of possible signatures as much as possible. 
Such reductions may be possible for instance by imposing further
structural constraints on the signatures that need to be considered.
It would also be important to find ways of generating valid signatures for one bag directly from the valid signatures if its child bag(s), rather than generating all possible signatures and then removing the invalid ones.

From a theoretical point of view, there are many opportunities for future work.
First, there are multiple variants and generalizations of \textsc{Tree Containment} that deserve attention,
including non-binary inputs, and inputs consisting of two networks (i.e. where the task is to decide if a network is contained in a second network). For the latter problem our approach would have to be extensively modified, since our size-bound on the signatures heavily relies on $\Tin$ being a tree. In the case of non-binary inputs, it is likely that a similar approach to the one in this paper would allow us to get a size-bound on the tree side of each signature. However, more work would be needed in order to prove a bound on the network side, and the number of possible embeddings. Note in particular that, under our current approach for deriving signatures from (partial) solutions, all neighbours of $S$ are preserved. Without a bound on the degrees, the number of such vertices can be much larger than the treewidth of  $\NinTin)$. 
This can lead to an explosion in the length of `compact' paths and the number of possible embeddings that one may need to consider for a single bag.
Such explosions may be avoidable through clever bookkeeping, or it may be that they are unavoidable and can be used to force a $W[1]$-hardness reduction.

Second, a major open problem is whether the \textsc{Hybridization Number} problem is FPT with respect to the treewidth of the output network. Again there are different variants: rooted and unrooted, binary and non-binary, a fixed or unbounded number of input trees.
For some applications, the definition of an embedding has to be relaxed (allowing, for example, multiple tree arcs embedded into the same network arc)~\citep{huber2016folding,huber2021rigid}.
Other interesting candidate problems for treewidth-based algorithms include
phylogenetic network drawing~\citep{klawitter2020drawing}, orienting phylogenetic networks~\citep{huber2019rooting} and phylogenetic tree inference with duplications~\citep{van2019polynomial}.

\mjnew{Finally, we believe that the approach taken in this paper (applying dynamic programming techniques on a tree decomposition of single graph representing all the input data, with careful attention given to the interaction between past and future) could potentially have applications outside of phylogenetics, in any context where the input to a problem consists of two or more distinct partially-labelled graphs that need to be reconciled.}

\medskip

{\bf Acknowledgements:} Research of Leo van Iersel and Mark Jones was partially funded by Netherlands Organization for Scientific Research (NWO) Vidi grant 639.072.602 and KLEIN grant OCENW.KLEIN.125.
\bibliography{main}

%

\end{document}